\documentclass[12pt]{article}

\usepackage{amsmath}
\usepackage{graphics}
\usepackage{epsf}
\usepackage{graphicx}
\usepackage{amscd}
\usepackage{amsfonts}
\usepackage{latexsym}
\usepackage{multicol}
\usepackage[active]{srcltx}
\usepackage{setspace}

\newtheorem{theorem}{Theorem}

\newtheorem{assumption}{Assumption}

\newenvironment{proof}[1][Proof]{\noindent\textbf{#1.} }{\
\rule{0.5em}{0.5em}}

\begin{document}

\title{\LARGE 
Predictability Hidden by Anomalous Observations\thanks{\small We acknowledge the
financial support of the Swiss National Science Foundation (PDFM1-114533) and the Swiss Finance Institute. We thank participants at the International Conference
on Computational Management Science 2009 in Geneva, the
International Conference on Robust Statistics 2009 in Parma, the
International Conference on Computational and Financial Econometrics
2009 in Limassol, the Prospectus Workshop in Econometrics at
Yale University 2010, the seminar at the University of Ghent 2010, the
Conference on Resampling Methods and High Dimensional Data at
Texas A\&M University 2010, ECARES seminar at ULB 2010 in
Bruxelles, the seminar at the Universities of Orleans 2012, the seminar at the University of Frankfurt 2012, the seminar at the University Bocconi 2012,
the seminar at the University of St.Gallen 2012, the European Summer Symposium in Financial Markets 2012 in Gerzensee, the Econometric Society European Meeting 2012 in Malaga,
the seminar at Northwestern University 2012, the annual SoFiE Conference 2012 in Oxford, and the seminar at the HEC Montreal 2013 for helpful comments. Part of this research was performed when the first author was visiting Yale University as a post doctoral research fellow, and the second author University of Cergy-Pontoise. Correspondence address: Fabio
Trojani, Geneva Finance Research Institute and School of Economics and Management, University of Geneva, Bd du Pont d'Arve 40, CH-1211 Geneva, e-mail: Fabio.Trojani@unige.ch.}}

\date{\small{November 2016}}

\author{Lorenzo Camponovo \\  { \it \small  University of St.Gallen and University of Surrey} \\
Olivier Scaillet \\{\it \small University of Geneva and
Swiss Finance Institute}
\\
Fabio Trojani \\ {\it \small University of Geneva and Swiss Finance
Institute}}

\maketitle

\linespread{1.5}

\newpage

\begin{abstract} \noindent
Testing procedures for predictive regressions
with
lagged autoregressive variables imply
a suboptimal
inference in presence of small violations of ideal assumptions.
We propose a novel testing framework resistant to such violations, which is consistent with
nearly integrated regressors
and applicable to multi-predictor settings, when the data may only approximately follow a predictive
regression model.
The Monte Carlo evidence demonstrates large improvements
of our approach, while the empirical analysis produces
a strong robust evidence of market return predictability hidden by anomalous observations, both in- and out-of-sample,
using predictive variables such as the dividend yield or the volatility risk premium.

\linespread{5}

\phantom{P}



\end{abstract}

\newpage

\section{Introduction}
A large literature has investigated whether economic variables such as, e.g., the price-dividend
ratio, proxies of labour income, or the interest rate can predict stock returns.\footnote{See Rozeff (1984), Fama
and French (1988), Campbell and Shiller (1988), Nelson and Kim
(1993), Goetzmann and Jorion (1995), Kothari and Shanken (1997), Campbell and Yogo (2006),
Jansson and Moreira (2006), Polk, Thompson and Vuolteenaho (2006), Santos and Veronesi (2006),
Bollerslev, Tauchen and Zhou (2009), among others.} The econometric approach to test for predictability is mostly based
on a predictive regression of stock returns onto a set of lagged
financial variables; see, e.g., Stambaugh (1999). Important differences between testing approaches in the literature arise
because of the different test statistics, asymptotic theories or resampling approaches
used to test the null hypothesis of no predictability. These
differences lead in a number of cases to diverging results and
conclusions.

Mankiw and Shapiro (1986) and Stambaugh (1986) note that in a
setting with endogenous predictor and correlated innovations
standard asymptotic theory causes small sample biases that may imply
an overrejection of the hypothesis of no predictability. To
mitigate the problem, recent studies propose tests based on
bias-corrected estimators of predictive regressions. For instance,
Stambaugh (1999), and Amihud, Hurvich and Wang (2008) introduce
bias-corrected OLS estimators for the univariate and the
multi-predictor setting, respectively.

Recent work has also considered the issue of endogenous integrated or
nearly integrated predictors, following the evidence in Torous,
Valkanov and Yan (2004) that various variables assumed to predict
stock returns follow a local-to-unit root autoregressive process.
Lewellen (2004), Torous,
Valkanov and Yan (2004), and Campbell and Yogo (2006)
introduce testing procedures and more accurate unit-root and
local-to-unit root asymptotics for predictive regression models
with a single persistent predictor and correlated innovations.
More recently, Kostakis, Magdalinos and Stamatogiannis (2015) propose a new class of test statistics, by extending the instrumental variables approach in Magdalinos and Phillips (2009) to predictive regressions.


A general approach to obtain tests that are less susceptible
to finite sample biases or
assumptions on the form of their asymptotic distribution relies on nonparametric Monte Carlo simulation methods, such as the bootstrap or the subsampling.
Ang and Bekaert (2007) use the bootstrap to quantify the bias of parameter estimation in a regression of
stock returns on the lagged dividend yield and the interest rate.
In a multi-predictor setting with nearly integrated regressors,
Amihud, Hurvich and Wang (2008) compare the results of bootstrap tests to bias-corrected procedures and find the latter to have
accurate size and good power properties. Wolf (2000) introduces subsampling tests of stock return predictability in single-predictor models.

As shown in Hall and Horowitz (1996) and Andrews (2002), among others, a desirable property of bootstrap tests is that they may
provide asymptotic refinements of the sampling distribution of standard $t$-test statistics for testing the hypothesis of no predictability.\footnote{In the sense that the errors made in approximating the true finite-sample distribution of the $t$-test statistic are of lower order with respect to
the sample size than those implied by the conventional asymptotics.}
Moreover, as shown in Romano and Wolf (2001), Choi and Chue (2007), and Andrews and Guggenberger (2009, 2010), subsampling methods produce reliable inference also in predictive regression models with multiple nearly integrated predictors.

A common feature of all above approaches to test
predictability hypotheses is their reliance on
procedures that can be heavily influenced by a small
fraction of anomalous observations in the data.
For standard OLS
estimators and $t$-test statistics, this problem is
well-known since a long time; see, e.g., Huber (1981) for a review.
More recent research has also shown that
inference provided by bootstrap and subsampling tests
may be easily inflated by a small fraction of anomalous observations.\footnote{Following Huber  seminal work, several authors have emphasized the
potentially weak robustness features
of many standard asymptotic testing procedures;
see Heritier and Ronchetti (1994), Ronchetti and Trojani (2001),
Mancini, Ronchetti and Trojani (2005), and Gagliardini, Trojani and
Urga (2005), among others. Singh (1998), Salibian-Barrera and Zamar (2002), and Camponovo,
Scaillet and Trojani (2012), among others, highlight the failing robustness of the inference implied by
bootstrap and subsampling tests in i.i.d. settings.}
Intuitively, we explain this feature by the
too high
fraction of anomalous observations that is often simulated by conventional bootstrap and subsampling procedures, when compared to the actual
fraction of outliers in the original data. It is not possible to mitigate this problem simply by applying conventional bootstrap or subsampling methods to more robust estimators or test statistics. Resampling trimmed or winsorized estimators does not yield a robust resampling method (see Singh, 1988, Camponovo, Trojani and Scaillet, 2012, for detailed examples).  Hence, we consider in this paper a new robust resample methodology for time series, which allows us to develop more robust tests of predictability hypotheses in predictive regression settings.

Our robust predictive regression approach relies on robust weighted least-squares procedures that are data-driven and easily manageable.  This can be motivated economically with the presence of a time-varying ambiguity about predictive relations,
which is consistently addressed by ambiguity averse investors using
robust estimators that bound the effects of anomalous data features. Wrampelmeyer, Wiehenkamp and Trojani (2015)
show that different specifications of aversion to ambiguity in the literature imply robust optimal estimator choices related to
robust weighted least-squares. In this sense, our robust predictive regression testing approach is consistent with the preferences of investors that dislike a time-varying ambiguity in the data-generating processes for returns.
The data-driven weights in the procedure dampen, where necessary, the few data points that
are estimated as anomalous with respect to
the postulated predictive link. 
This feature automatically avoids, e.g., 
arguing ex ante that a large value of the predicted or the predictive variables is per se an anomalous observation, which is not in general
the case.  Indeed, observations linked to large values of both the predictive and the predicted variables
might be very informative about a potential predictability structure and discarding them in an ad hoc way might bias the inference. In a multivariate predictive regression setting, it is
even more difficult to determine with an informal approach
which subset of observations is
potentially anomalous, for example by eyeballing the data. A useful property of our methodology it that it embeds a formal data-driven identification of 
observations that can be excessively influential for the resulting inference on predictive relations.
The more detailed contributions to the literature are as follows.

%
%
%
%

First, using Monte Carlo simulations we find that the size and power of conventional hypothesis testing methods for predictive regressions, including
bias-corrected tests, tests implied by local-to-unity asymptotics, and conventional bootstrap and subsampling tests, are dramatically nonresistant
to even small fractions of anomalous observations in the data. Even though the test
probability of rejecting a null by chance alone
features
some degree of resistance in our Monte Carlo experiments, the test
ability to reject the null of no predictability when it is violated is in most
cases drastically reduced.

Second, we quantify
theoretically the robustness properties of subsampling and bootstrap tests in a time series context, borrowing from the concept of breakdown point,
which is a measure of the degree of resistance of a testing procedure to outliers; see, e.g., Hampel (1971), Donoho and Huber (1983), and Hampel, Ronchetti, Rousseeuw and Stahel (1986).
In Section 3.3 below, Theorem 1 (and its proof) for subsampling differs from Theorem 2 (and its proof) in Camponovo, Scaillet and Trojani (2012) valid for the i.i.d. case
since subsamples in a time series context are not generated in the same way in order to avoid breaking serial dependences. Theorem 1 (and its proof) for block bootstrap is also new, and not a straightforward extension of the results in Singh (1998) for the i.i.d. bootstrap.


Third, we develop a novel class of resampling tests of predictability,
which are resistant to anomalous observations and consistent with nearly integrated regressors
at sustainable computational costs.\footnote{Our robust resampling approach relies on the fast resampling idea put forward, among others, in Shao and Tu (1995), Davidson and McKinnon (1999), Hu and Kalbfleisch (2000), Andrews (2002), Salibian-Barrera and Zamar (2002), Goncalves and White (2004), Hong and Scaillet (2006), Salibian-Barrera, Van Aelst and Willems (2006, 2007), and Camponovo, Scaillet and Trojani (2012). The methodology is applicable to a wide set of bootstrap and subsampling simulation schemes in the literature.} 
We confirm by Monte Carlo simulations that these tests successfully limit the damaging effect of outliers, by preserving desirable finite sample
properties
in presence of anomalous observations.

Finally, we
provide a robust analysis of the recent empirical evidence on stock return predictability for
US stock market data. Following Cochrane (2008), our main purpose
is not to determine the best return-forecasting specification, but rather to study 
the amount of predictability resulting from very simple specifications, motivated by economic theory for the vast majority for the data.
We study single-predictor and multi-predictor
models, using several well-known
predictive variables suggested in the literature, such as
the lagged dividend yield, the difference between option-implied volatility and realized volatility (Bollerslev, Tauchen and Zhou, 2009), the interest rate, and the share of labor income to consumption (Santos and Veronesi, 2006). Our robust tests of predictability produce the following novel empirical evidence.

First, we
find that the dividend yield is a robust predictive variable of market returns, which is significant at the 5\% significance in all our regressions, for each subperiod, sampling frequency and forecasting horizon considered. In univariate regressions with monthly data, the lagged dividend yield is significant at the 5\% level according to the robust tests, in each window of 180 monthly observations from January 1980 to December 2010. In contrast, bias-corrected methods, local-to-unity asymptotics and conventional sabsampling tests produce a weaker and more ambiguous evidence overall, e.g., by not rejecting the null of no predictability
at the 10\% significance level
in the subperiod from January 1995 to December 2010. 
Multi-predictor regressions including
variance risk premium and labor income proxies confirm the
significant predictive power of the dividend yield. While the dividend yield is again significant at the 5\% level in all cases using the robust tests, it is significant only at the 10\% level
using the conventional
tests
in the sample period 1994-2009, within monthly predictive regressions including the difference of implied and realized volatility
as a predictive variable. It is not significant using conventional tests in the sample period 1955-2010 within quarterly predictive regression based on the share of labor income to consumption.
A weak evidence of predictability for future returns and cash flows is at odd with
the fundamental present-value relations linking expected returns, expected cash-flows and the price-dividend ratio; see also
Cochrane (2011). Therefore, our robust findings of market return predictability are consistent with the main logic of basic present-value models.

Second, we find that the difference between option-implied volatility and realized volatility is a robust predictive variable of future market returns at quarterly forecasting horizons. It is always significant at the 5\% significance level in each window of 180 observations, using both robust and nonrobust testing approaches.
This finding supports the remarkable market return forecasting ability of the variance risk premium, noted in Bollerslev, Tauchen and Zhou (2009) and confirmed in Bollerslev, Marrone and Zhou (2014) in an international context.

Third, using conventional testing approaches, we find that
the evidence of return predictability associated with the ratio of labor income to consumption is either absent or weak in the sample
periods 1955-2000 and 1965-2010, respectively. In contrast, the null of no predictability is always rejected at the 5\% significance level by our robust testing method, indicating that the weak and ambiguous evidence produced by the conventional tests is likely a consequence of their low power in presence of anomalous observations.

Fourth, we exploit the properties of our robust testing method to identify observations that might excessively influence the diverging conclusions of
conventional testing approaches. We find a fraction of less than about 5\%
of influential observations in the data, which tend to be more frequent
during the NASDAQ bubble and the more recent financial crisis. Such influential data points, including the Lehman Brothers default on September 2008,
the terrorist attack of September 2001, the Black Monday on October 1987, and the Dot-Com bubble collapse in August 2002,
are largely responsible for
the failure of conventional testing methods in uncovering the hidden predictability structures.

Finally, we find that our results cannot be substantially improved by specifying predictive relations with time-varying parameters, as we do not find evidence of structural breaks in our sample period after controlling for the impact of anomalous observations.
Motivated by the findings in Goyal and Welch (2003) and Campbell and Thompson (2008), we also show that the predictive relations detected by our robust approach generate incremental out-of-sample predictive power over a monthly forecasting horizon, improving on both the predictive relations estimated by conventional methods or those of a simple forecast based on the sample mean of market returns.

The rest of the paper is organized as follows. In Section 2, we introduce the usual
predictive regression model, and we illustrate by simulation the robustness problem of some of the recent tests of predictability proposed in the literature. In Section 3, we study theoretically the robustness properties of bootstrap and subsampling approximations. In Section 4, we introduce our robust approach, and develop robust bootstrap and subsampling tests of predictability. In Section 5, we apply our robust testing procedure to US equity data and reconsider some of the recent empirical evidence on market return
predictability. Section 6 concludes.

\section{Predictability and Anomalous Observations}\label{par}

In this section, we introduce the benchmark predictive regression model and a number of recent methods proposed for testing the predictability of stock returns. Through Monte Carlo simulations, we study the finite-sample properties of these testing procedures both in presence and absence of anomalous observations.
In Section \ref{pm}, we first introduce the model. In Section \ref{ri}, we focus on bias-corrected methods and testing procedures based on local-to-unity asymptotics.
Finally, in Section \ref{rm}, we consider testing approaches based on resampling methods.

\subsection{The Predictive Regression Model}\label{pm}
We consider the
predictive regression model,
\begin{eqnarray}
y_{t} & = & \alpha+\beta x_{t-1}+u_{t}, \label{pregmodel1}\\  x_{t}
& = & \mu + \rho x_{t-1}+v_{t},\label{pregmodel2}
\end{eqnarray}
where, $y_t$ denotes the stock return at time $t=1,\dots,n,$ and $x_{t-1}$ is an economic variable observed at time $t-1$, predicting
$y_t$. The parameters $\alpha\in\mathbb{R}$ and $\mu\in\mathbb{R}$ are the unknown intercepts of the linear regression model and the autoregressive model, respectively, $\beta\in\mathbb{R}$ is the unknown parameter of interest, $\rho\in\mathbb{R}$ is the unknown autoregressive coefficient, $u_t\in\mathbb{R}$, $v_t\in\mathbb{R}$ are error terms with $u_t=\phi v_t+e_t$, $\phi\in\mathbb{R}$, and $e_t$ is a scalar random variable.

In this setting, it is well-known that inference based on standard
asymptotic theory suffers from small sample biases, which may imply
an overrejection of the hypothesis of no predictability,
$\mathcal{H}_{0}: \beta_{0}=0$, where $\beta_0$ denotes the true value of the unknown parameter $\beta$; see Mankiw and Shapiro
(1986), and Stambaugh (1986), among others. Moreover, as emphasized in Torous, Valkanov, and Yan (2004), various state variables considered as predictors in model (\ref{pregmodel1})-(\ref{pregmodel2}) might be
well approximated by a nearly integrated process, which might
motivate a local-to-unity framework $\rho=1+c/n$, $c<0$, for the autoregressive coefficient of model (\ref{pregmodel2}), implying a nonstandard asymptotic distribution for the OLS estimator $\hat{\beta}_n$ of parameter $\beta$.

Several recent testing procedures have been proposed in order to
overcome these problems. Stambaugh (1999), Lewellen (2004), Amihud and Hurvich (2004), Polk, Thompson and Vuolteenaho (2006), and Amihud, Hurvich and Wang (2008, 2010), among others, propose bias-corrected procedures that correct the bias implied by the OLS estimator $\hat{\beta}_n$ of parameter $\beta$. Cavanagh, Elliott and Stock (1995),
Torous, Valkanov and Yan (2004), and Campbell and Yogo (2006), among others, introduce testing procedures based on local-to-unity asymptotics that provide more accurate approximations of the sampling distribution of the $t$-statistic $T_n=(\hat{\beta}_n-\beta_0)/\hat{\sigma}_n$ in nearly integrated settings, where $\hat{\sigma}_n$ is an estimate of the standard deviation of the OLS estimator $\hat{\beta}_n$. Kostakis, Magdalinos and Stamatogiannis (2015) also propose a new class of test statistics by extending the instrumental variables approach develop in Magdalinos and Phillips (2009) to predictive regressions.

\subsection{Bias Correction Methods and Local-to-Unity Asymptotic Tests}\label{ri}

A common feature of bias-corrected methods and inference based on local-to-unity asymptotics is a nonresistance to anomalous observations, which may lead to conclusions
determined by the particular features of a small subfraction of the
data.
Intuitively, this feature emerges because these approaches exploit statistical
tools that can be sensitive to
small deviations from the predictive regression model (\ref{pregmodel1})-(\ref{pregmodel2}). Consequently, despite the good accuracy under the strict model assumptions, these testing procedures may become less efficient
or biased even with a small fraction of anomalous observations in the data.

To illustrate the lack of robustness of this class of tests, we analyze through Monte Carlo simulation
the bias-corrected method proposed in Amihud, Hurvich and Wang (2008) and the Bonferroni approach for the local-to-unity asymptotic theory introduced in Campbell and Yogo (2006). We first generate $N=1,000$ samples
$z_{(n)}=\big(z_1,\dots,z_n\big)$, where $z_t=(y_t,x_{t-1})'$, of size $n=180$ according to model (\ref{pregmodel1})-(\ref{pregmodel2}), with $v_t\sim N(0,1)$, $e_t\sim N(0,1)$, $\phi=-1$, $\alpha=\mu=0$, $\rho\in\{0.9,0.95,0.99\}$, and $\beta_{0}\in\{0,0.05,0.1\}$.\footnote{These parameter choices are in line with the Monte Carlo setting studied, e.g., in Choi and Chue (2007). Unreported Monte Carlo results with $\phi=-2,-5$ are qualitatively very similar.} 
In a second step, to study the
robustness of the methods under investigation, we consider
replacement outliers random samples
$\tilde{z}_{(n)}=\big(\tilde{z}_1,\dots,\tilde{z}_n\big)$, where $\tilde{z}_t=(\tilde{y}_t,x_{t-1})'$ is
generated according to,
\begin{equation}\label{conta}
\tilde{y}_{t} = (1-p_{t})y_{t}+p_{t}\cdot y_{3max},
\end{equation}
with $y_{3max}=3\cdot\max(y_{1},\dots,y_{n})$ and $p_{t}$ is an
i.i.d. $0-1$ random sequence, independent of process
(\ref{pregmodel1})-(\ref{pregmodel2}) such that
$P[p_{t}=1]=\eta$. The probability of contamination by outliers is set to
$\eta=4\%$, which is a small contamination of the original sample, compatible with the features of the real data set analyzed in the empirical study in Section \ref{spm1}.\footnote{For the monthly data set in Section \ref{spm1}, the estimated fraction of anomalous observations
in the sample period 1980-2010
is less than about $3.87\%$.}

We study the finite sample
properties
of tests of the null hypothesis ${\cal
H}_0:\beta_{0}=0$ in the predictive regression model. In the first two rows of Figure \ref{poweracy}, we plot
the empirical frequency of rejection of null hypothesis ${\cal
H}_0$ for the bias-corrected method proposed in Amihud, Hurvich and Wang (2008) and the Bonferroni approach for the local-to-unity asymptotic theory introduced in Campbell and Yogo (2006), respectively, with respect to different values of the alternative
hypothesis $\beta_{0}\in\{0,0.05,0.1\}$, and different degree of persistence of the predictors
$\rho\in\{0.9,0.95,0.99\}$. The nominal significance level of the test is $10\%$.

The results for different degree of persistence are qualitatively very similar.
In the Monte Carlo simulation with noncontaminated samples (straight line), we find that the fraction of null hypothesis rejections
of all procedures
is quite close to the nominal level
$10\%$ when $\beta_{0}=0$. As expected, the power of the tests increases for increasing values of
$\beta_{0}$. 
In the simulation with contaminated samples (dashed line), the size of all tests remains quite close to the nominal significance level. In contrast, the presence of anomalous observations dramatically deteriorates the power of both procedures. Indeed, for $\beta_{0}>0$, the frequency of rejection of the null hypothesis for both tests is much lower than in the noncontaminated case. 
Unreported Monte Carlo results for the instrumental variable approach proposed in Kostakis, Magdalinos and Stamatogiannis (2015) produce similar findings.

The results in Figure \ref{poweracy} highlight the
lack of resistance to anomalous data of bias-corrected methods and inference based on local-to-unity asymptotics. Because of a small fraction
of anomalous observations, the testing procedures become unreliable, and are unable to reject the null hypothesis of no predictability,
even for large values of $\beta_0$.
This is a relevant aspect for applications, in which typically the statistical evidence of predictability is weak.

To overcome this robustness problem, a natural approach is to develop more resistant versions of the nonrobust tests considered in our Monte Carlo exercise. However, this task may be hard to achieve in general.\footnote{To robustify the bias-corrected procedure in Amihud, Hurvich and Wang (2008),
we would need to derive an expression for the bias of robust
estimators of regressions, and then derive the asymptotic distribution of such bias-corrected robust estimators.
For nearly integrated settings, a robustification of the procedure proposed in Campbell and Yogo (2006) would require a
not obvious extension of the
robust local-to-unity asymptotics developed in Lucas (1995, 1997) for the predictive regression model.}
A more general approach to obtain tests that are less susceptible
to finite sample biases or assumptions on their asymptotic distribution can rely on nonparametric Monte Carlo simulation methods.
We address these methods in the sequel.

\subsection{Bootstrap and Subsampling Tests}\label{rm}

Nonparametric Monte Carlo simulation methods, such as the bootstrap and the subsampling, may provide improved inferences in predictive regression model (\ref{pregmodel1})-(\ref{pregmodel2}) both in stationary or nearly integrated settings. As shown in Hall and Horowitz (1996) and Andrews (2002), for
stationary data the block bootstrap may yield improved
approximations to
the sampling distribution of the
standard $t$-statistics for testing predictability,
having asymptotic errors of lower order in sample size. Moreover, as shown in Choi and Chue (2007) and Andrews and Guggenberger (2010), we can use the subsampling to produce correct inferences in nearly integrated settings.
We first introduce block bootstrap and subsampling procedures. We then focus on
predictive regression model (\ref{pregmodel1})-(\ref{pregmodel2}) and study by Monte Carlo simulation the degree of resistance to anomalous observations of bootstrap and subsampling tests of predictability, both in stationary and nearly integrated settings.

Consider a random sample $z_{(n)}=(z_{1},\dots,z_{n})$ from a time series of random vectors $z_i\in\mathbb{R}^{d_z}$, $d_z\ge 1$, and a general
statistic
$T_{n}:=T(z_{(n)})$. Block bootstrap procedures split the
original sample $z_{(n)}$ into overlapping blocks of size
$m<n$. From these blocks, bootstrap samples $z_{(n)}^*$ of size $n$ are randomly generated.\footnote{See, e.g., Hall
(1985), Carlstein (1986), K\"unsch (1989) and Andrews (2004), among others. Alternatively, it is possible to construct the bootstrap samples
using nonoverlapping blocks.}
Finally,
the empirical distribution of statistic $T (z^*_{(n)})$ is used to estimate the sampling distribution of $T(z_{(n)})$.
Similarly, the
more recent subsampling method applies statistic $T$ directly to overlapping random blocks $z_{(m)}^*$ of size $m$ strictly less than $n$.\footnote{See Politis,
Romano and Wolf (1999), among others.}
Then, the empirical distribution of statistic $T (z^*_{(m)})$ is used to estimate the sampling distribution of $T(z_{(n)})$, under the assumption that
the impact of the block size is asymptotically negligible ($m/n\to 0$).

In the predictive regression model (\ref{pregmodel1})-(\ref{pregmodel2}), the usual $t$-test statistic for testing the null
of no predictability is $T_n=(\hat{\beta}_n-\beta_0)/\hat{\sigma}_n$.
Therefore, we can define
a block bootstrap test of the null hypothesis with the block bootstrap statistic
$T_{n,m}^{B\ast}=(\hat{\beta}_{n,m}^{B\ast}-\hat{\beta}_n)/\hat{\sigma}_{n,m}^{B\ast}$,
where $\hat{\sigma}_{n,m}^{B\ast}$ is an estimate of the standard deviation of the OLS estimator $\hat{\beta}_{n,m}^{B\ast}$ in
a random bootstrap sample of size $n$, constructed using blocks of size $m$.
Similarly, we can define a subsampling test of the same null hypothesis with the subsampling statistic $T_{n,m}^{S\ast}=(\hat{\beta}_{n,m}^{S\ast}-\hat{\beta}_n)/\hat{\sigma}_{n,m}^{S\ast}$,
where $\hat{\sigma}_{n,m}^{S\ast}$ is now an estimator of the standard deviation of the OLS estimator $\hat{\beta}_{n,m}^{S\ast}$ in
a random overlapping block of size $m<n$.

It is well-known that OLS estimators and empirical averages are very sensitive to even small fractions
of anomalous observations in the data; see, e.g., Huber (1981). Since bootstrap and subsampling tests rely on such statistics, inference based on these methods may inherit the lack of robustness. To verify this intuition, we study the finite-sample properties
of bootstrap and subsampling tests of predictability in presence of anomalous observations through Monte Carlo simulations. First we consider stationary settings.
To this end, we generate $N=1,000$ samples
$z_{(n)}=\big(z_1,\dots,z_n\big)$, where $z_t=(y_t,x_{t-1})'$, of size $n=180$ according to model (\ref{pregmodel1})-(\ref{pregmodel2}), with $v_t\sim N(0,1)$, $e_t\sim N(0,1)$, $\phi=-1$, $\alpha=\mu=0$, $\rho\in\{0.3,0.5,0.7\}$, and $\beta_{0}\in\{0,0.1,0.2\}$. We then consider also contaminated samples $\tilde{z}_{(n)}=\big(\tilde{z}_1,\dots,\tilde{z}_n\big)$ according to (\ref{conta}). We test the null hypothesis ${\cal
H}_0:\beta_{0}=0$, using symmetric bootstrap and subsampling confidence intervals for parameter $\beta$ under different values of the alternative
hypothesis $\beta_{0}\in\{0,0.1,0.2\}$.\footnote{
Section \ref{near} shows how to construct symmetric confidence intervals for the parameter of interest, based on resampling distributions.
For the selection of the block size $m$, we use the standard data-driven method proposed in Romano and Wolf (2001).}

In the first and second rows of 
Figure \ref{powerstat}, we plot the empirical frequencies of rejection of null hypothesis ${\cal
H}_0$, using the subsampling and the bootstrap, respectively. 
The nominal significance level of the test is $10\%$.
For $\beta_0=0$, the size of the tests is close to the nominal level $10\%$, while for $\beta_0=0.2$ the power increases. On the other hand, in presence of contamination the power of the tests dramatically decreases.  Indeed, for $\beta_{0}>0$, also in this case the frequency of rejection of the null hypothesis for both tests is much lower than in the noncontaminated case. 


We also study the robustness properties of the subsampling in nearly integrated settings, under
the same simulation setting of the previous section.
In the third row of Figure \ref{poweracy}, we plot
the empirical frequencies of rejection of null hypothesis ${\cal
H}_0$
for different values of the alternative
hypothesis $\beta_{0}\in\{0,0.05,0.1\}$. The nominal significance level of the test is $10\%$, as before.
With noncontaminated samples (straight line), we find
for all values of $\beta_{0}\in\{0,0.05,0.1\}$ that the frequency of rejection of subsampling tests is close to the one of the bias-corrected method and the Bonferroni approach in the previous section. For $\beta_0=0$, the size of the tests is close to the nominal level $10\%$, while for $\beta_0=0.1$ the power increases.
On the contrary, also in this case, 
contaminations with anomalous observations
strongly deteriorate the power of the
tests. 

In summary, the results in Figures \ref{poweracy} and \ref{powerstat} show that bootstrap and subsampling tests inherit, and to some extent exacerbate, the lack of robustness of OLS estimators for predictive
regressions.
To robustify the inference produced by
resampling methods,
a natural idea is to
apply conventional bootstrap and subsampling simulation schemes to a more robust statistic,
such as, e.g., a robust estimator of linear regression.
Unfortunately, as
shown in Singh (1998), Salibian-Barrera and Zamar (2002), and Camponovo, Scaillet and Trojani (2012)
for i.i.d. settings,
resampling a robust statistic does not yield a robust inference, because conventional bootstrap and subsampling procedures have an intrinsic nonresistance to outliers.
Intuitively, this problem arises because the fraction of anomalous observations generated
in bootstrap and subsampling blocks is often much higher than the
fraction of outliers in the data. To solve this problem, it is necessary
to address more systematically the robustness of resampling methods for time series.

\section{Resampling Methods and Quantile Breakdown Point}\label{sec}

We characterize theoretically the robustness of bootstrap and subsampling tests in predictive regression settings. Section \ref{bpq}
introduces the notion of a quantile breakdown point, which
is a measure of the global resistance
of a resampling method to anomalous observations.
Section \ref{bpointmc} quantifies and illustrates the quantile breakdown point of conventional bootstrap and subsampling tests in predictive regression models.
Finally, Section \ref{tlor} derives explicit bounds for quantile breakdown points, which quantify 
the degree of resistance to outliers of bootstrap and subsampling tests for predictability, before applying them to the data.

\subsection{Quantile Breakdown Point}\label{bpq}

Given a random sample $z_{(n)}$
from a sequence of random vectors $z_i\in\mathbb{R}^{d_z}$, $d_z\ge 1$,
let $z^{\ast}_{(n)}=(z_{1}^{\ast},\dots,z_{n}^{\ast})$ denote a block bootstrap sample, constructed using overlapping blocks of size $m$. Similarly, let $z^{\ast}_{(m)}=(z_{1}^{\ast},\dots,z_{m}^{\ast})$ denote an overlapping subsampling block.
The construction of blocks is a key difference with respect to the i.i.d. setting, and implies a different extension of available results on breakdown properties for i.i.d. data.
We denote by $T_{n,m}^{K\ast}$, $K=B,S$, the corresponding block bootstrap and subsampling statistics, respectively.\footnote{We focus for brevity on one-dimensional real-valued
statistics. However, as discussed for instance in Singh (1998) in the i.i.d.
context, we can extend our results for time series to multivariate and scale statistics.}
For $t\in(0,1)$, the
quantile $Q_{t,n,m}^{K\ast}$ of $T_{n,m}^{K\ast}$ is defined by
\begin{equation}\label{quantile}
Q_{t,n,m}^{K\ast} =\inf \{x \arrowvert P^{\ast}(T_{n,m}^{K\ast}\le x) \ge
t\},
\end{equation}
where $P^{\ast}$ is the probability measure induced by the block bootstrap or the subsampling method and, by definition, $\inf (\emptyset) = \infty$. Quantile $Q_{t,n,m}^{K\ast}$ is effectively a useful nonparametric estimator of the corresponding finite-sample quantile of statistic $T(z_{(n)})$.
We characterize the robustness properties of block bootstrap and subsampling by the breakdown point $b_{t,n,m}^{K\ast}$ of the quantile (\ref{quantile}), which is defined as the smallest fraction of outliers in the
original sample such that $Q_{t,n,m}^{K\ast}$ diverges to infinity.


Borrowing the notation in Genton and Lucas (2003), we
formally define the breakdown point of the $t$-quantile $Q^{K\ast}_{t,n,m}:=Q_{t,n,m}^{K\ast}(z_{(n)})$ as,
\begin{equation}\label{break2}
b^{K\ast}_{t,n,m}
:=\frac{1}{n}\cdot\Bigg[\inf_{\{1\le p\le \lceil
n/2\rceil\}}\big\{p\big\vert \textrm{there exists } z_{(n,p)}^{\zeta}\in
\mathcal{Z}_{(n,p)}^{\zeta}\textrm{ such that }
Q^{K\ast}_{t,n,m}(z_{(n)}+z_{(n,p)}^{\zeta})=+\infty\big\}\Bigg],
\end{equation}
where $\lceil x\rceil =\inf\{n\in\mathbb{N}\vert x\le n\}$, and $\mathcal{Z}_{(n,p)}^{\zeta}$ is the set of all
$n$-samples $z_{(n,p)}^\zeta$ with exactly $p$ nonzero components that are $d_z$-dimensional outliers of size $\zeta\in{\bar{\mathbb{R}}^{d_z}}$.\footnote{When $p>1$, we
do not necessarily assume outliers $\zeta_{1},\dots,\zeta_{p}$ to be
all equal to $\zeta$, but we rather assume existence of constants
$c_{1},\dots,c_{p}$, such that $\zeta_{i}=c_{i}\zeta$.
To better capture the presence of outliers in predictive regression models, our definitions for the breakdown point and the set $\mathcal{Z}_{(n,p)}^{\zeta}$ of all $n$-components outlier samples are slightly different from those proposed in Genton and Lucas (2003) for general settings. However, we can modify our results to cover alternative definitions of breakdown point and outlier sets $\mathcal{Z}_{(n,p)}^{\zeta}$.}
Literally, $b^{K\ast}_{t,n,m}
$ is the smallest fraction of anomalous observations
of arbitrary size,
in a generic outlier-contaminated sample $z_{(n)}+z_{(n,p)}^\zeta$,
such that quantile $Q^{K\ast}_{t,n,m}$, estimated by bootstrap or subsampling Monte Carlo simulation schemes, can become meaningless.

Intuitively, when a breakdown
occurs, inference about
the distribution of $T(z_{(n)})$ based on
bootstrap or subsampling tests
becomes pointless. Estimated test
critical values may be arbitrarily large and confidence intervals be arbitrarily wide. In these cases,
the size and power of bootstrap and subsampling tests can collapse to zero or one in presence of anomalous observations, making these inference procedures useless.
Therefore, providing theory (see Theorem \ref{bsubboot} below) for quantifying $b^{K\ast}_{t,n,m}$ in general for bootstrap and subsampling tests of predictability, in dependence of the statistics and
testing approaches used, is key in order to understand which approaches ensure some resistance to anomalous observations and which do not, even before looking at the data.

\subsection{Quantile Breakdown Point and Predictive Regression}\label{bpointmc}

The quantile breakdown point of conventional
block bootstrap and subsampling tests for predictability in Section \ref{rm} depends directly on the breakdown properties of
OLS estimator $\hat{\beta}_n$. The breakdown point $b$ of a statistics $T_n=T(z_{(n)})$ is simply
the smallest fraction of outliers in the original sample such that
the statistic $T_n$ diverges to infinity; see, e.g., Donoho and Huber (1983) for the formal definition.
We know $b$ explicitly in some cases and we can gauge its value most of the time, for
instance by means of simulations and sensitivity analysis. Most nonrobust
statistics, like OLS estimators for linear regression, have a breakdown point $b = 1/n$.
Therefore, the breakdown point of conventional block bootstrap and subsampling quantiles in
predictive regression settings also equals $1/n$. In other words, a single anomalous observation in the original data is sufficient
to produce a meaningless inference implied by bootstrap or subsampling quantiles in standard tests of predictability.

It is straightforward to illustrate these features in a Monte Carlo simulation that quantifies
the sensitivity of block bootstrap and subsampling quantiles to data contaminations
by a single outlier, where the size of the outlier is increasing.
We first simulate $N=1,000$ random samples
$z_{(n)}=\big(z_1,\dots,z_n\big)$ of size $n=120$, where $z_t=(y_t,x_{t-1})'$ follows model (\ref{pregmodel1})-(\ref{pregmodel2}), $v_t\sim N(0,1)$, $e_t\sim N(0,1)$, $\phi=-1$, $\alpha=\mu=0$, $\rho=0.9$, and $\beta_{0}=0$. For each Monte Carlo sample, we define in a second step
\begin{eqnarray}  \label{maxy}
y_{\max}=\arg\max_{y_{1},\dots,y_{n}}\{w(y_{i})\vert w(y_{i})=
y_{i}-\beta_0x_{i-1}, \text{under}\,\mathcal{H}_0: \beta_0 = 0\}\ ,
\end{eqnarray}
and we modify $y_{\max}$ over the interval
$[y_{\max},y_{\max}+5] $. This means that we contaminate the predictability relationship by an anomalous observation for only one single data point in the full sample. We study
the sensitivity of the Monte Carlo average length of confidence intervals for parameter $\beta$, estimated by the standard block bootstrap and the subsampling.
This is a natural exercise, as the length of the confidence interval for parameter $\beta$ is in a one-to-one relation with the critical value of the test of the null of no predictability (${\cal H}_0$: $\beta_0=0$).
For the sake of comparison, we also consider confidence intervals implied by
the bias-corrected testing method in Amihud, Hurvich and Wang (2008) and the Bonferroni approach proposed in Campbell and Yogo (2006).

For all tests under investigation,
in the first 2 rows of Figure \ref{sensacy}, we plot the relative increase of the average confidence interval length in our Monte Carlo simulations, under
contamination by a single outlier of increasing size.
We find that all sensitivities are basically linear in the size of the outlier, confirming that a single anomalous observation can have an arbitrarily large
impact on the critical values of those tests and make the test results potentially useless, as implied by their quantile breakdown point of $1/n$.


\subsection{Quantile Breakdown Point Bounds}\label{tlor}

To obtain bootstrap and subsampling tests with more favorable breakdown properties, it is necessary to apply resampling procedures to a robust
statistic with nontrivial breakdown point ($b>1/n$), such as, e.g., a robust estimator of linear regression.
Without loss of generality, let $T_n=T(z_{(n)})$ be a statistic with breakdown point $1/n < b\le 0.5$.

In the next theorem, we compute
explicit
quantile breakdown point bounds, which
characterize the
resistance of bootstrap and subsampling tests to anomalous observations, in dependence of relevant
parameters, such as
$n$, $m$, $t$, and $b$.\footnote{Similar results can be obtained
for the subsampling and the block bootstrap based on nonoverlapping blocks. The results for the
block bootstrap can also be modified to cover
asymptotically equivalent variations, such as the stationary
bootstrap of Politis and Romano (1994).} 
\begin{theorem}\label{bsubboot}
Let $b$ be the breakdown point of $T_{n}$, $t\in(0,1)$, and $r= \lceil n/m\rceil$. The
quantile breakdown points $b^{S\ast}_{t,n,m}$ and $b^{B\ast}_{t,n,m}$ satisfy the following bounds,
\begin{eqnarray*}
&&\frac{\lceil mb\rceil}{n}\le  b^{S}_{t,n,m} \le \frac{1}{n}\cdot\bigg[\inf_{\{p\in\mathbb{N},
p\le r-1\}}\bigg\{p\cdot \lceil mb\rceil \bigg\vert
p>\frac{(1-t)(n- m+1)+\lceil mb\rceil-1}{m}\bigg\}\bigg],\label{fibps}\\
&&\frac{\lceil mb\rceil}{n}\le b^{B}_{t,n,m} \le
\frac{1}{n}\cdot\bigg[\inf_{\{p_{1},p_{2}\}}\bigg\{p=p_{1}\cdot
p_{2}\bigg\vert P\bigg(
BIN\bigg(r,\frac{mp_{2}-p_{1}+1}{n-m+1}\bigg)\ge\bigg\lceil\frac{nb}{p_{1}}\bigg\rceil\bigg)>1-t\bigg\}\bigg],\phantom{P}\label{fibpb}
\end{eqnarray*}
where $p_{1},p_{2}\in\mathbb{N}$, with $p_{1}\le m, p_{2}\le r-1$, and $BIN(N,q)$
denotes a binomially distributed variable with parameters $N\in\mathbb N$ and $q\in (0,1)$.
\end{theorem}

In Theorem \ref{bsubboot}, the term $\frac{(1-t)(n- m+1)}{m}$ represents the number of
degenerated subsampling statistics necessary in order to cause the breakdown of
$Q_{t,n,m}^{S\ast}$, while $\frac{\lceil mb\rceil}{n}$ is the fraction of
outliers which is sufficient to cause the breakdown of statistic $T$
in a block of size $m$. The breakdown point formula
for the i.i.d. bootstrap derived in Singh (1998) emerges as a special case of
the second inequality in Theorem \ref{bsubboot}.

We quantify the implications of Theorem \ref{bsubboot} by computing in Table \ref{table1} lower and upper bounds for the
breakdown point of subsampling and bootstrap quantiles, using a sample
size $n=120$, and a maximal statistic breakdown point ($b=0.5$).
We find that even for a highly robust statistic with maximal breakdown point, the
subsampling implies a very low quantile breakdown point,
which increases with the block size but is also very far from the maximal
value $b=0.5$.
For
instance, for a block size $m=10$,
the $0.95$-quantile breakdown point
is between $0.0417$ and $0.0833$.
In other words, even though a statistic is resistant to large fractions of anomalous observations, the implied
subsampling quantile can collapse with just 5 outliers out of 100 observations.\footnote{
This breakdown point is also clearly lower than
in the i.i.d. case; see  Camponovo, Scaillet and Trojani (2012). For
instance, for $m=10$, the $0.95$-quantile breakdown point of the
overlapping subsampling is 0.23 in i.i.d. settings.
Since in a time series setting the number of possible
subsampling blocks of size $m$ is typically lower than the number of
i.i.d. subsamples of size $m$, the breakdown of a statistic
in one random block tends to have a larger impact on the subsampling
quantile than in the i.i.d. case.}
Similar results arise for the bootstrap quantiles. Even though the bounds are less sharp than for the subsampling, quantile breakdown points are again clearly smaller than the breakdown point of the statistic used.\footnote{
These quantile breakdown point bounds are again clearly lower than in the i.i.d. setting.
For instance,
for $m=30$, the $0.95$-quantile breakdown point for time series is less than
$0.25$, but it is $0.425$ for i.i.d. settings, from the results in Camponovo, Scaillet and Trojani (2012).}

Overall, the results in Theorem \ref{bsubboot} imply that subsampling and bootstrap tests for time series feature an intrinsic non-resistance to anomalous observations, which cannot be avoided, simply by applying conventional resampling approaches to more robust statistics.

\section{Robust Resampling Methods}
When using a robust statistic
with large breakdown point, the bootstrap and subsampling
still
imply an important nonresistance to anomalous observations, which is consistent with our Monte Carlo results in the predictive regression model.
To overcome the problem, it
is necessary to introduce a novel class of more robust resampling tests in the time series context.\footnote{We develop such robust
methods borrowing from the fast resampling approach considered,
among others, in Shao and Tu (1995), Davidson and McKinnon (1999),
Hu and Kalbfleisch (2000), Andrews (2002), Salibian-Barrera and
Zamar (2002), Goncalves and White (2004), Hong and Scaillet (2006),
 Salibian-Barrera,
Van Aelst and Willems (2006, 2007), and Camponovo, Scaillet and
Trojani (2012).} Section \ref{rpr} introduces our robust bootstrap and subsampling approaches, and Section \ref{robb} demonstrates theoretically their favorable breakdown properties. Section \ref{near} characterizes the asymptotic validity of the robust subsampling in both stationary and nonstationary settings.
Finally, in Section \ref{mcrs} we study the accuracy of our approach through Monte Carlo simulations.

\subsection{Robust Predictive Regression and Hypothesis Testing}\label{rpr}

We develop a new class of easily applicable robust resampling tests for the null hypothesis of no predictability in predictive regression models.
To this end, first we focus on robust estimators with nontrivial breakdown point $b>1/n$.
Several such estimators are available in the literature.
Among those estimators, a convenient choice is the Huber estimator of regression, which ensures together good robustness properties and moderate computational costs.

Let $\theta=(\alpha,\beta)'$ and $w_{t-1}=(1,x_{t-1})'$, 
given a
positive constant $c$, the robust Huber estimator $\hat{\theta}_{n}^{R}$ is the $M$-estimator that solves the equation
\begin{equation}\label{huber1}
\psi_{n}(z_{(n)},\hat{\theta}_{n}^{R}):=\frac{1}{n}\sum_{t=1}^{n}
g(z_t,\hat{\theta}_n^R)\cdot h_c(z_t,\hat{\theta}_n^R)=0,
\end{equation}
where the functions $g$ and $h_c$ are defined as
\begin{eqnarray}
g(z_t,\theta) & := & 
(y_{t}-w_{t-1}'\theta)w_{t-1}, \\
h_c(z_t,\theta)  & := & \min\left(1,\frac{c}{\Vert (y_{t}-w'_{t-1}\theta)w_{t-1}\Vert}\right).\label{huberw}
\end{eqnarray}
In Equation (\ref{huber1}), we can write the Huber estimator $\hat{\theta}_n^R$ as a weighted least squares estimator with data-driven weights $h_c$ defined by (\ref{huberw}).
By design, the Huber weight $0\le h(z_t,\theta)\le 1$ reduces the influence of potential anomalous observations on the estimation results. 
Equation (\ref{huber1}) is an estimating function and not the way we define the predictive relationship.
Weights below one indicate a potentially anomalous data-point, while weights equal to one indicate unproblematic observations for the postulated model.
Therefore, the value of weight (\ref{huberw}) provides a useful way for highlighting potential anomalous observations that might be excessively influential for the
fit of the
predictive regression model; see, e.g., Hampel, Ronchetti, Rousseeuw and Stahel (1986).

Constant $c>0$ is useful in order to tune the degree of resistance to anomalous data of estimator $\hat{\theta}_{n}^{R}$ in relevant applications, and it can be determined in a fully data-driven way.\footnote{By extending the calibration method proposed in Romano and Wolf (2001) for both the selection of the block size $m$ and the degree of robustness $c$.}
Note that 
the norm of function $\psi_{n}$ in Equation (\ref{huber1}) is bounded (by constant $c$),
and the breakdown point of estimator $\hat{\theta}_{n}^{R}$ is maximal ($b=0.5$, see, e.g., Huber, 1981).

%
%

Conventional bootstrap and subsampling
solve equations $\psi_{k}(z^{\ast}_{(k)},\hat{\theta}_{n,m}^{R\ast})=0$, 
with $k=n$ (bootstrap) and $k=m$ (subsampling)
for each random bootstrap sample $z_{(n)}^{\ast}$ and subsampling random sample
$z_{(m)}^{\ast}$, respectively, which can be a computationally demanding task. Instead, we
consider a standard Taylor expansion of (\ref{huber1}) around
the true parameter $\theta_{0}$,
\begin{equation}\label{taylor}
\hat{\theta}_{n}-\theta_{0}=-[\nabla_{\theta}\psi_{n}(z_{(n)},\theta_{0})]^{-1}\psi_{n}(z_{(n)},\theta_{0})+o_{p}(1),
\end{equation}
where $\nabla_{\theta}\psi_{n}(z_{(n)},\theta_{0})$ is the
derivative of function $\psi_{n}$ with respect to parameter $\theta$. Based on
this expansion, we can use
$-[\nabla_{\theta}\psi_{k}(z_{(k)}^{\ast},\hat{\theta}_{n}^R)]^{-1}\psi_{k}(z_{(k)}^{\ast},\hat{\theta}_{n}^R)$
as an approximation of $\hat{\theta}_{n,m}^{R\ast}-\hat{\theta}_{n}^R$ in
the definition of the resampling scheme estimating the sampling
distribution of $\hat{\theta}_{n}-\theta_{0}$. This approach avoids computing
$\hat{\theta}^{R\ast}_{n,m}$ in random samples, which is a markable computational advantage that produces a fast numerical procedure. This is an important improvement over
conventional resampling schemes,
which can easily become unfeasible when applied to robust
statistics.
Let
\begin{eqnarray*}
\hat{\Sigma}_{n}^R & = & [\nabla_{\theta}\psi_{n}(z_{(n)},\hat{\theta}_{n}^R)]^{-1} 
\left(\frac{1}{n}\sum_{t=1}^{n}
g(z_t,\hat{\theta}_n^R)g(z_t,\hat{\theta}_n^R)'\cdot h_c(z_t,\hat{\theta}_n^R)^2\right)
[\nabla_{\theta}\psi_{n}(z_{(n)},\hat{\theta}_{n}^R)]^{-1}, \\
\hat{\Sigma}_{k}^{R\ast} & = & [\nabla_{\theta}\psi_{k}(z_{(k)}^{\ast},\hat{\theta}_{n}^R)]^{-1} 
\left(\frac{1}{k}\sum_{t=1}^{k}
g(z_t^{\ast},\hat{\theta}_n^R)g(z_t^{\ast},\hat{\theta}_n^R)'\cdot h_c(z_t^{\ast},\hat{\theta}_n^R)^2\right)
[\nabla_{\theta}\psi_{k}(z_{(k)}^{\ast},\hat{\theta}_{n}^R)]^{-1},
\end{eqnarray*}
with $k=n,m$.
Following this fast resampling approach, we can finally estimate the sampling distribution of $\sqrt{n}[\hat{\Sigma}_n^R]^{-1/2}(\hat{\theta}_n^R-\theta_0)$
with the distribution 
\begin{equation}\label{fastpr1}
L_{n,m}^{R\ast}(x)=\frac{1}{N}\sum_{i=1}^{N}\mathbb{I}
\bigg(\sqrt{k}[\hat{\Sigma}_{k,i}^{R\ast}]^{-1/2}\bigg(-[\nabla_{\theta}\psi_{k}(z_{(k),i}^{\ast},\hat{\theta}_{n}^R)]^{-1}\psi_{k}(z_{(k),i}^{\ast},\hat{\theta}_{n}^R)\bigg)\le
x\bigg),
\end{equation}
where $N$ denote the number of possible random samples.
In the next section, we analyze the breakdown properties of the robust fast resampling procedure.

\subsection{Robust Resampling and Quantile Breakdown Point}\label{robb}

A closer inspection of quantity
$[\hat{\Sigma}_{k,i}^{R\ast}]^{-1/2}[\nabla_{\theta}\psi_{k}(z_{(k),i}^{\ast},\hat{\theta}_{n}^R)]^{-1}$
$\psi_{k}(z_{(k),i}^{\ast},\hat{\theta}_{n}^R)$ in Equation (\ref{fastpr1})
reveals important implications for the breakdown properties of the robust fast resampling distribution ($\ref{fastpr1}$).
Indeed, this
quantity can degenerate only when either (i) matrix $\hat{\Sigma}_{k,i}^{R\ast}$ is singular,
(ii) matrix
$\nabla_{\theta}\psi_{k}(z_{(k),i}^{\ast},\hat{\theta}_{n}^R)$ is singular or
(iii) estimating function $\psi_{k}(z_{(k),i}^{\ast},\hat{\theta}_{n}^R)$ is not bounded. However, since we are making use
of a robust (bounded) estimating function, situation (iii) cannot
arise.
Therefore, we intuitively expect the breakdown of the quantiles of robust subsampling distribution ($\ref{fastpr1}$) to
arise only when conditions (i) or (ii) are realized.\footnote{Unreported Monte Carlo simulations show that the
application of our robust resampling approach to an $M$-estimator with nonrobust (unbounded) estimating function
does not solve the robustness problem, consistent with
our theoretical results in Section \ref{tlor}.}
We borrow from 
this intuition and in the next theorem, we compute the quantile breakdown point of resampling distribution (\ref{fastpr1}).
\begin{theorem}\label{ftheo}
For simplicity, let $r=n/m\in\mathbb{N}$. The $t$-quantile breakdown points $b_{t,n,m}^{RB\ast}$ and
$b_{t,n,m}^{RS\ast}$ of the robust block bootstrap and robust subsampling distributions, respectively, are given by
\begin{eqnarray}
b_{t,n,m}^{RS\ast} & = & \frac{1}{n}\bigg[\inf_{\{p\in\mathbb{N},p\le n-m+1\}} \bigg\{m+p\bigg\vert\ p>(1-t)(n-m+1)-1\bigg\}\bigg],\label{fbps}\\
b_{t,n,m}^{RB\ast} & = & \frac{1}{n}\bigg[\inf_{\{p\in\mathbb{N},p\le n-m+1\}} \bigg\{m+p\bigg\vert P\bigg(BIN\bigg(r,\frac{p+1}{n-m+1}\bigg)=r\bigg)>1-t\bigg\}\bigg],\label{fbpb}
\end{eqnarray}
where $BIN(N, q)$ is a Binomial random variable with parameters $N$ and $q\in (0,1)$.
\end{theorem}

The quantile breakdown points of the robust bootstrap and subsampling approach are
often much higher than the one of conventional bootstrap and subsampling. Table \ref{table111}
quantifies these differences, 
confirming that the robust bootstrap and subsampling quantile breakdown points in Table \ref{table111} are
considerably larger than those in Table \ref{table1} for conventional bootstrap and subsampling methods.

\subsection{Robust Subsampling and Asymptotic Size}\label{near}

In this section, we show that the robust subsampling can provide an asymptotically valid method for introducing inference in both stationary and nonstationary predictive regression models. 
However, to achieve this objective, the definition of robust subsampling tests requires some cares. Indeed, only tests based on symmetric robust subsampling confidence intervals ensure a correct inference.
To this end, we consider the predictive regression model (\ref{pregmodel1})-(\ref{pregmodel2}) under the following assumptions.
\begin{assumption}\label{asso}
\item[(i)] $\rho\in(-1,1]$, and $u_t=\phi v_t+e_t$, $\phi\in\mathbb{R}$, where $v_t$ and $e_t$ are independent.
\item[(ii)] $v_t$ and $e_t$ are strictly stationary with $E[v_t]=0$, and $E[e_t]=0$. Furthermore, for $\epsilon>0$, $E[\vert v_t\vert^{2+\epsilon}]<\infty$, and $E[\vert e_t\vert^{2+\epsilon}]<\infty$.
\item[(iii)] $v_t$ and $e_t$ are strong mixing with mixing coefficients $\alpha_{v,m}$ and $\alpha_{e,m}$, respectively, that satisfy $\sum_{m=1}^{\infty}\alpha_{k,m}^{\epsilon/(2+\epsilon)}<\infty$, $\epsilon>0$, $k=v,e$.
\item[(iv)] $\theta_0=(\alpha_0,\beta_0)'\in\Theta_0$ is the unique solution of $E[\psi_n(z_{(n)},\theta_0)]=0$, and the set $\Theta_0\subset\mathbb{R}^2$ is compact.
\end{assumption}

Assumption \ref{asso} provides a set of conditions also adopted in Choi and Chue (2007) to prove the validity of the subsampling in nearly integrated settings, and in Lucas (1995) to derive the limit distribution of robust M-estimators in integrated settings.

Consider the statistic $T_n^R=\sqrt{n}(\hat{\beta}_n^R-\beta_0)/\sigma_n^R$, where $\sigma_n^R$ is the square-root of the second diagonal component of $\Sigma_n^R$. When $\vert\rho\vert<1$, then $T_n^R$ converges in distribution to a standard normal, see Choi and Chue (2007). On the other hand, when $\rho=1$, the limit distribution of $T_n^R$ is nonstandard and depends on nuisance parameters that have to be simulated, see Lucas (1995). Because of this discontinuity in the limit distribution, conventional bootstrap methods are inconsistent. 
Also subsampling approximations may suffer from a lack of uniform convergence, see Andrews and Guggenberger (2009). To verify the uniform validity of inference based on the robust subsampling, we follow the same approach adopted in Andrews and Guggenberger (2009), and focus on the quantiles of statistics $T_n^R$, $-T_n^R$, and $\vert T_n^R\vert$. More precisely, in Figure \ref{asize}, we simulate the $0.95$-quantiles of the limit distribution of these statistics 
for different values of the degree of persistence $\rho=1-c/n$, with $c\in[0,10]$, and covariance parameter of the error terms $\phi\in\{0,-1,-2,-5\}$.

In Figure \ref{asize}, we can observe that the graphs of the 0.95--quantile for different values of $\phi$ have
similar shapes, and are monotone in $c$. In particular, the $0.95$--quantiles
of the limit distribution of $T_n^R$ and $\vert T_n^R\vert$ are decreasing, while those of $-T_n^R$ are increasing. 
Therefore, using the same arguments adopted in Section 7 in Andrews and Guggenberger (2009), we can conclude that upper and symmetric conventional subsampling confidence intervals have correct asymptotic size, while lower and equal-tailed conventional subsampling confidence intervals have incorrect size asymptotically. More precisely, given $t\in(0,1)$ let $CI_{t,\vert\cdot\vert}$ denote a $t$-confidence interval obtained by inverting the conventional subsampling approximation of the sampling distribution of statistic $\vert T_n^R\vert$. Then, 
$\lim_{n\to\infty} \inf_{\rho} P(\beta_0\in CI_{t,\vert\cdot\vert})=t$,
i.e., conventional symmetric confidence intervals ensure a correct asymptotic size uniformly in the degree of persistence $\rho$. Importantly, because of the negligible remainder term in the Taylor approximation (\ref{taylor}), these results hold also for our symmetric robust subsampling confidence intervals.

\subsection{Monte Carlo Evidence}\label{mcrs}

To quantify the implications of Theorem \ref{ftheo}, we can study the sensitivity of confidence intervals estimated by the robust bootstrap and subsampling,
with respect to contaminations by anomalous observations of increasing size.
To this end, we consider the same Monte Carlo setting of Section \ref{bpointmc}. In the last row of Figure \ref{sensacy}, we plot the percentage increase of the length in the average estimated confidence interval, with respect to contaminations of the available data by a single anomalous observation of increasing size for the robust bootstrap and subsampling, respectively. In evident contrast to the findings for conventional testing procedures, Figure \ref{sensacy} shows that the inference implied by our robust
 approach
 is largely insensitive to outliers, with a percentage increase in the average confidence interval length that is less than $1\%$, even for an outliers of size $y_{max}+5$.

The robustness of our approach has favorable implications
for the power of bootstrap and subsampling tests in presence of
anomalous observations.
For the same Monte Carlo setting of  Sections \ref{ri} and \ref{rm},
Figures \ref{poweracy} and \ref{powerstat} show
that in presence of noncontaminated samples (straight line) the frequencies of null hypothesis rejections of robust bootstrap and subsampling tests are again very close to those
observed for nonrobust methods. 
This means that the asymptotic efficiency loss of robust estimators in the absence of anomalous observations do not seem to reduce the performance of the robust bootstrap and subsampling with respect to nonrobust procedures.
However,
in presence of anomalous observations
(dashed line),
robust bootstrap and subsampling tests still provide
an accurate empirical size close to the actual nominal level,
as well as a power curve that is close to the one obtained in the noncontaminated Monte Carlo simulation. 

\section{Empirical Evidence of Return Predictability}\label{eres}

Using our robust resampling tests, we revisit the recent empirical evidence on return
predictability for US stock market data from a robustness perspective.
We study single-predictor and multi-predictor settings, using several well-known predictive variables suggested in the literature,
such as the lagged dividend yield, the difference between option-implied volatility and realized volatility
(Bollerslev, Tauchen and Zhou, 2009), and the share of labor income to consumption (Santos
and Veronesi, 2006). Because of the high persistence of dividend yields, in this empirical analysis we do not consider bootstrap procedures. We compare the evidence produced by
our robust subsampling tests of predictability
with the results of recent testing methods proposed in the literature, including
the bias-corrected method in Amihud, Hurvich and Wang (2008), the Bonferroni approach for local-to-unity asymptotics
in Campbell and Yogo (2006), and conventional subsampling tests.

The dividend yield is the most common predictor of future stock returns, as suggested by a simple present-value logic.\footnote{See, e.g., Rozeff (1984), Campbell and Shiller (1988), Fama and French (1988), Stambaugh (1999), Lewellen (2004), Torous, Valkanov and Yan (2004), Lettau and Ludvigson (2005), and Campbell and Yogo (2006).} However, its forecasting ability
has been called into question, e.g., by the ambiguous empirical evidence of
studies not rejecting the null of no predictability for a number of forecasting horizons and sample periods; see, e.g., Goyal and Welch (2003), and Ang and Bekaert (2007), among others. Whether these ambiguous results are
related to the weakness of
conventional tests in
detecting predictability structures masked by
anomalous observations, is an empirical question that we can analyze using our robust testing method.

The empirical study is articulated in three parts.
Section \ref{spm1}
studies the forecast ability of the lagged dividend yield for explaining
monthly S$\&$P 500 index returns, in
a predictive regression model with a single predictor.
This study allows us to compare the results of our methodology with those of the
Bonferroni approach for local-to-unity asymptotics, which is applicable to univariate regression settings.
Instead, Section \ref{tpm2} considers models with several predictive variables. In Section \ref{dyvrp}, we test
the predictive power of the dividend yield and the variance risk premium,
for quarterly S$\&$P 500 index returns sampled at a monthly frequency in periods marked by a financial bubble and a financial crisis. Section \ref{dyli} tests the predictive power of the dividend yield and the ratio of labor income to consumption for predicting
quarterly value-weighted CRSP index returns.\footnote{We also consider regressions with
three predictive variables that additionally incorporate interest rate proxies. We discuss below the results, but we do not report the details for brevity.}

\subsection{Single-Predictor Model}\label{spm1}

We consider monthly S$\&$P 500 index returns from Shiller (2000),
$R_{t}=(P_{t}+d_{t})/P_{t-1}$,
where $P_{t}$ is the end of month real stock price and $d_{t}$
the real dividend paid during month $t$. Consistent with the literature, the
annualized dividend series $D_{t}$ is defined as,
\begin{equation}\label{dividend}
D_{t}=d_{t}+(1+r_{t})d_{t-1}+(1+r_{t})(1+r_{t-1})d_{t-2}+\dots+(1+r_{t})\dots(1+r_{t-10})d_{t-11},
\end{equation}
where $r_{t}$ is the one-month maturity Treasury-bill rate.
We estimate
the predictive regression model
\begin{equation}\label{pregmodel31}
\ln(R_{t}) = \alpha+\beta
\ln\bigg(\frac{D_{t-1}}{P_{t-1}}\bigg)+\epsilon_{t}\ ;\ t=1,\dots,n ,
\end{equation}
and test the null of no predictability, ${\cal
H}_0:\beta_{0}=0$.

We collect monthly observations in the sample
period 1980-2010 and estimate the predictive regression model using rolling windows of 180 observations.
Table \ref{tablef1} reports the detailed point estimates and test results for the different testing procedures in the four subperiods 1980-1995, 1985-2000, 1990-2005, 1995-2010.

We find that while the robust subsampling tests
always clearly reject the hypothesis of no predictability at the 5\%-significance level,
the conventional testing approaches produce a weaker and more ambiguous predictability evidence.
For instance, the subsampling tests
cannot reject ${\cal H}_0$ at the $10\%$ significance level in
subperiod 1985-2000, while
the bias-corrected method and the Bonferroni approach fail to reject ${\cal H}_0$ at the $10\%$ significance level in the subperiod
1995-2010.

It is interesting to study to which extent anomalous observations in sample periods 1985-2000 and 1995-2010 might have caused the diverging conclusions of robust and nonrobust testing methods. We exploit the properties of our robust testing method to identify
such data points. Figure \ref{huberweights} plots the time series of Huber weights estimated by the robust estimator (\ref{huber1}) of the predictive regression model (\ref{pregmodel31}).

We find that subperiod 1998-2002 is characterized by a cluster
of infrequent anomalous
observations, which are likely related to the abnormal stock market performance during the NASDAQ bubble in the second half of the 1990s.
Similarly, we find a second cluster of anomalous observations in subperiod 2008-2010, which is linked to the extraordinary events of the recent financial crisis.
Overall, anomalous observations are less than 4.2\% of the whole data sample, and they
explain the failure of conventional
testing methods in uncovering hidden predictability structures in these sample periods.

We find that the most influential observation
before 1995 is November 1987, following the Black Monday on October 19 1987. During the subperiod 1998-2002, the most influential observation is October 2001, reflecting the impact on financial markets of the terrorist attack
on September 11 2001. Finally, the most anomalous observation in the whole sample period 1980-2010 is October 2008, following the Lehman Brothers default on September 15 2008. 

To investigate the potential presence of time-varying parameters in the predictive regression model (\ref{pregmodel31}), we test formally for the presence of structural breaks.
We apply both the standard Wald test statistic proposed by Andrews (1993), and its robust version introduced in Gagliardini, Trojani and Urga (2005). 
Asymptotic critical values of these test statistics are provided
in Andrews (1993). To improve on the inference of asymptotic tests, we also follow Diebold and Chen (1996) and Gagliardini, Trojani and Urga (2005), and implement nonrobust and robust resampling tests of structural breaks for our
predictive regression model.
Using all methods, we never reject the null hypothesis of no structural break at the $10\%$ significance level in our sample period.
Therefore, the lack of predictability produced in some cases by the standard approach cannot be explained by a structural break in a significant subset of the the data. This evidence
supports the presence of a small subset of influential anomalous observations
as a plausible explanation for the diverging conclusions 
of classical and robust predictive regression methods.

Finally, we study the out-of-sample accuracy of predictive regressions estimated by nonrobust and robust methods. Borrowing from Goyal and Welsh (2003) and Campbell and Thompson (2008), we introduce the out-of-sample $R^2_{OS}$ statistics, defined as
\begin{equation}\label{r2os}
R^2_{OS}=1-\frac{\sum_{t=1}^{T} (y_t-\hat{y}_{t,ROB})^2}{\sum_{t=1}^{T} (y_t-\hat{y}_{t,OLS})^2},
\end{equation}
where $\hat{y}_{t,ROB}$ and $\hat{y}_{t,OLS}$ are the fitted values from a predictive regression estimated up to period $t-1$, using the robust Huber estimator and the OLS estimator, respectively. 
Whenever statistic $R^2_{OS}$ is positive,
the robust approach
yields a lower average mean squared prediction error than the nonrobust method, providing more accurate out-of-sample forecasts. As reported in Table \ref{tableos}, we obtain $R^2_{OS}=0.51\%$. Therefore, 
besides the more robust in-sample results, our robust approach also yields better out-of-sample predictions. To compare the out-of-sample accuracy of the nonrobust and robust approaches with respect to the simple forecast based on the sample mean of market returns, we consider also the 
out-of-sample $R^2_{OS,K}$ statistic, defined as
\begin{equation}\label{r2os2}
R^2_{OS,K}=1-\frac{\sum_{t=1}^{T} (y_t-\hat{y}_{t,K})^2}{\sum_{t=1}^{T} (y_t-\bar{y}_t)^2},
\end{equation}
where $\bar{y}_{t}$ is the historical average return estimated through period $t-1$, and $K=ROB,OLS$. As reported in Table \ref{tableos}, we obtain $R^2_{OS,ROB}=4.04\%$, and
$R^2_{OS,OLS}=3.51\%$. Therefore, both nonrobust and robust methods provide more accurate out-of-sample predictions than simple forecast based on the sample mean of market returns.

\subsection{Two-Predictor Model}\label{tpm2}

We extend our empirical study to two-predictor regression models.
This approach has several purposes. First, we can
assess the incremental predictive ability of the dividend yield, in relation to other well-known competing predictive variables. Second, we can verify the power properties of robust subsampling tests in settings with several predictive variables.

Section \ref{dyvrp} borrows from Bollerslev, Tauchen and Zhou (2009) and studies the joint predictive
ability of the dividend yield and the variance risk premium. Section \ref{dyli} follows the two-predictor model in Santos and Veronesi (2006), which considers the ratio of labor income to consumption as an additional predictive variable to the dividend yield.

\subsubsection{Bollerslev, Tauchen and Zhou}\label{dyvrp}

We consider again monthly S$\&$P 500 index and dividend data between January 1990 and December 2010, and test the predictive regression model:
\begin{equation}\label{pregmodel3f}
\frac{1}{k}\ln(R_{t+k,t}) = \alpha+\beta_1
\ln\bigg(\frac{D_{t}}{P_{t}}\bigg)+\beta_2
VRP_t
+\epsilon_{t+k,t},
\end{equation}
where $\ln(R_{t+k,t}):=\ln(R_{t+1})+\dots+\ln(R_{t+k})$ and the variance risk premium $VRP_t:=IV_t-RV_t$ is defined by
the difference of
the S\&P 500 index option-implied volatility at time $t$, for one month maturity options, and
the ex-post realized return variation over the period $[t -1, t ]$.
Bollerslev, Tauchen and Zhou (2009) show that the variance risk premium is the most
significant
predictive variable of market returns over a quarterly horizon. Therefore, we test
the predictive regression model (\ref{pregmodel3f}) for $k=3$.

Let $\beta_{01}$ and $\beta_{02}$ denote the true values of parameters $\beta_1$ and $\beta_2$, respectively. Using
the subsampling tests, as well as our robust subsampling tests, we first test
the null hypothesis of no return predictability by the dividend yield, ${\cal
H}_{01}:\beta_{01}=0$.


Table \ref{tablef23} collects the detailed point estimates and testing results.
We find again that the robust tests always clearly reject the null of no predictability at the 5\%-significance level.
In contrast, the conventional 
subsampling tests
produce weaker and more ambiguous results, with uniformly lower $p$-values (larger confidence intervals) and a nonrejection of the null of no predictability at the $5\%-$level in period 1994-2009. Since the Bonferroni approach
in Campbell and Yogo (2006) is defined for single-predictor models, we cannot
apply this method in
model (\ref{pregmodel3f}). Unreported
results for the multi-predictor testing method
in Amihud, Hurvich and Wang
(2008) show that for data windows following window 1993-2008 the bias-corrected method cannot
reject null hypothesis
${\cal
H}_{01}$ at the 10\% significance level.

By inspecting the Huber weights (\ref{huberw}), implied by the robust estimation of the predictive regression model (\ref{pregmodel3f}), we find again a cluster of infrequent anomalous observations, both during the NASDAQ bubble and the recent financial crisis. In this setting, the most influential observation is still October 2008, reflecting the Lehman Brothers default on September 15 2008. 


Table \ref{tablef23} reports the estimates and testing results for parameter $\beta_{02}$. In contrast to the previous evidence,
we find that all tests under investigation clearly reject ${\cal H}_{02}$ at the 5\%-significance level, thus confirming
the remarkable return forecasting ability of the variance risk premium noticed in Bollerslev, Tauchen and Zhou (2009), as well as the international evidence reported in Bollerslev, Marrone, Xu and Zhou (2014).\footnote{Besides the two-predictor model (\ref{pregmodel3f}), we also consider the three-predictor model
\begin{equation}\label{pregmodel3fint}
\frac{1}{k}\ln(R_{t+k,t}) = \alpha+\beta_1
\ln\bigg(\frac{D_{t}}{P_{t}}\bigg)+\beta_2VRP_t+\beta_3 LTY_t
+\epsilon_{t+k,t},
\end{equation}
where $LTY_t$ is the detrended long-term yield, defined as the ten-year Treasury yield minus its trailing twelve-month moving averages. Again, using the standard subsampling and the robust subsampling, we find evidence in favor of predictability at $5\%$ significance level for the variance risk premium for the sample period 1990-2010. In contrast, all tests do not reject the null hypothesis of no predictability at $10\%$ significance level for the detrended long-term yield. Finally, both conventional and robust tests reject the null hypothesis of no predictability at the $5\%$ significance level for the dividend yield. The comparison of these empirical results with those obtained in the two-predictor model (\ref{pregmodel3f}) again confirms the reliability of our robust tests and the (possible) failure of nonrobust procedures in uncovering predictability structures in presence of anomalous observations.}
Finally, also for this predictive regression model, we do not find evidence of structural breaks at the $10\%$ significance level. Moreover, we obtain out-of-sample statistics $R^2_{OS}=1.40\%$ and $R^2_{OS,ROB}=5.70\%$,
indicating again an improved out-of-sample predictive power for our robust approach.


\subsubsection{Santos and Veronesi}\label{dyli}

We finally focus on the two-predictor regression model proposed in Santos and Veronesi (2006):
\begin{equation}\label{pregmodel3}
\ln(R_{t}) = \alpha+\beta_1
\ln\bigg(\frac{D_{t-1}}{P_{t-1}}\bigg)+\beta_2 s_{t-1}+\epsilon_{t},
\end{equation}
where $s_{t-1}=w_{t-1}/C_{t-1}$ is the share of labor income to consumption.
We make use of quarterly returns on the value weighted CRSP index, which includes NYSE, AMEX, and NASDAQ stocks, in the sample period Q1,1955-Q4,2010.
The dividend time-series is
also obtained from CRSP, while the risk free rate is the three-months Treasury bill rate.
Labor income and consumption are obtained from the Bureau of Economic
Analysis.\footnote{As in Lettau and Ludvigson (2001), labor income is defined as wages and salaries, plus transfer payments, plus other labor income, minus personal contributions for social insurance, minus taxes. Consumption is defined as nondurables plus services.}

Let $\beta_{01}$ and $\beta_{02}$ denote the true values of parameters $\beta_1$ and $\beta_2$, respectively. Using
subsampling tests, as well as our robust testing method, we first test
the null hypothesis of no predictability by the dividend yield, ${\cal
H}_{01}:\beta_{01}=0$.
Table \ref{tablef45} collects detailed point estimates and test results for the four subperiods 1950-1995, 1955-2000, 1960-2005, 1965-2010. We find again that our robust tests always clearly reject ${\cal H}_{01}$ at the $5\%$-significance level.
In contrast, conventional
tests produce more ambiguous results, and cannot reject at the 10\%-significance level the null hypothesis
${\cal H}_{01}$ for subperiod 1955-2000.

Table \ref{tablef45} summarizes estimation and testing results
for parameter $\beta_{02}$.
While the conventional tests produce a weak and mixed evidence of return predictability using labor income proxies, e.g., by not rejecting ${\cal
H}_{02}$ at the 10\%-level in subperiod 1950-1995, the robust tests
produce once more a clear and consistent predictability evidence for all sample periods.

The clusters of anomalous observations (less than $4.6\%$ of the data in the full sample), highlighted by the estimated weights in Figure \ref{huberweights2}, further indicate that conventional
tests might
fail to uncover hidden
predictability structures using samples of data that include observations from the NASDAQ bubble or the recent financial crisis, a feature that
was noted
in Santos and Veronesi (2006) and Lettau and Van Nieuwerburgh (2007) from a completely different angle.
In such contexts, the robust subsampling tests are again found to control well the potential damaging effects of anomalous observations, by
providing a way to consistently uncover hidden predictability features also when the data may only approximately follow the given predictive regression model.
We do not find evidence of structural breaks in the predictive relation at the $10\%$ significance level, while we obtain an out-of-sample statistic $R^2_{OS}=1.13\%$, indicating that our robust approach improves the
out-of-sample predictions of classical predictive regression methods.
However, in this case the out-of-sample statistic $R^2_{OS,ROB}=-2.73\%$ shows no improvement over quarterly forecasts provided by standard sample mean of market returns.

\section{Conclusion}
A large literature studies the predictive
ability of a variety of
economic variables for future market returns.
Several useful
testing approaches for testing the null of no predictability
in
predictive regressions
with correlated errors
and nearly integrated regressors have been proposed, including tests
that rely
on nonparametric Monte Carlo simulation methods, such as the bootstrap and subsampling.
All these methods improve
on the
conventional asymptotic tests
under
the ideal assumption of an exact predictive regression model. However,
we find
by Monte Carlo evidence that even small violations of such assumptions, generated by a small fraction of anomalous observations, can result in
large deteriorations in the reliability of all these tests.

To systematically
understand the problem, we characterize theoretically the robustness properties of resampling tests of
predictability in a time series context, using the concept of quantile breakdown point, which is a measure of the global resistance of a testing procedure to outliers. We obtain general quantile breakdown point formulas, which highlight an important nonresistance of
these tests
to anomalous observations that might infrequently contaminate the predictive regression model, thus
confirming the fragility
detected in our Monte Carlo study.

We propose a more robust
testing method for
predictive regressions
with correlated errors
and nearly integrated regressors, by introducing a novel general
class of fast and robust resampling procedures for predictive regression models at sustainable
computational costs.
The new resampling tests
are resistant to anomalous observations in the data and imply more robust confidence intervals and inference results.
We demonstrate by
Monte Carlo simulations their good resistance to outliers and their
improved
finite-sample properties in presence of anomalous observations.

In our
empirical study for US stock market data,
we study single-predictor and multi-predictor
models, using well-known predictive variables in the literature, such as
the market dividend yield, the difference between index option-implied volatility and realized volatility
(Bollerslev, Tauchen and Zhou, 2009), and the share of labor income to consumption (Santos
and Veronesi, 2006).

First, using the robust tests we find clear-cut
evidence
that the dividend yield is a robust predictive variable for market returns,
in each subperiod and for each sampling
frequency and forecasting horizon considered.
In contrast, tests
based on
bias-corrections, local-to-unity asymptotics, or standard subsampling procedures provide more ambiguous findings, by not rejecting the null of no predictability in a number of cases.

Second, we find that the difference between option-implied volatility and realized volatility is a robust
predictive variable of future market returns at quarterly forecasting horizons,
both using robust
and nonrobust testing methods. This finding confirms the remarkable return forecasting ability of
the variance risk premium, first noticed
in Bollerslev, Tauchen and Zhou (2009).

Third, we find that conventional testing approaches deliver
an ambiguous evidence of return predictability
by proxies of labor income, which is either absent or weak in the
sample periods 1955-2000 and 1965-2010, respectively.
In contrast, the null of no predictability
is always clearly rejected using the robust testing approach,
indicating that
the weak findings of the conventional tests are
likely deriving from their low ability to detect predictability structures
in
presence of small sets of anomalous observations.

Fourth, we exploit the properties of our robust tests
to identify potential anomalous observations
that might explain the diverging conclusions of robust and nonrobust
methods.
We find a fraction of less than about 5\% of anomalous observations in the data,
which tend to cluster during the NASDAQ bubble and the more recent financal crisis.
Anomalous data points, including the Lehman Brothers default on September 2008, the
terrorist attack of September 2001, the Black Monday on October 1987, and the Dot-Com bubble
collapse in August 2002, are responsible for the failure of conventional testing methods in uncovering
the hidden predictability structures for these sample periods.

Fifth, we find that the different conclusions of our robust approach with respect to conventional methods cannot be explained by the presence of time-varying predictive regression parameters, as we do not find any evidence of structural breaks in predictive relations over our sample period. Moreover, we find that the out-of-sample predictions for monthly market
returns of our robust approach are more accurate than the ones given by conventional predictive regressions
and sample mean market return forecasts.

Finally, while our subsampling tests have been developed in the context of standard predictive systems with autocorrelated regressors, our approach is extendable also to more general settings, including potential nonlinear predictive relations or unobserved state variables.
For instance, van Binsbergen and Koijen (2010) propose a latent-variable approach and a Kalman filter to estimate a present value model with hidden and persistent expected return and dividend growth, in order to formulate powerful tests for the joint predictability of stock returns and dividend growth. The application of our robust subsampling tests in the context of such present value models is an interesting avenue for future research.

\newpage

{\noindent \textbf{\Large Appendix A: Mathematical Proofs}} \newline
\newline

\begin{proof}[Proof of Theorem \protect\ref{bsubboot}]
We first consider the subsampling. 
The value $\frac{\lceil mb\rceil}{n}$ is the smallest fraction of
outliers, that causes the breakdown of statistic $T$
in a block of size $m$. Therefore, the first inequality is
satisfied.

For the second inequality, 
we denote by
$z^{N}_{(m),i}=(z_{(i-1)m+1},\dots,z_{im})$, $i=1,\dots,r$ and
$z^{O}_{(m),i}=(z_{i},\dots,z_{i+m-1})$, $i=1,\dots,n-m+1$, the
nonoverlapping and overlapping blocks of size $m$, respectively.
Given the original sample $z_{(n)}$, for the first nonoverlapping
block $z^{N}_{(m),1}$, consider the following type of contamination:
\begin{equation}\label{co}
z^{N}_{(m),1}=(z_{1},\dots,z_{m-\lceil mb\rceil},C_{m-\lceil
mb\rceil+1},\dots,C_{m}),
\end{equation}
where $z_{i}$, $i=1,\dots,m-\lceil mb\rceil$ and $C_{j}$,
$j=m-\lceil mb\rceil+1,\dots,m$, denote the noncontaminated and
contaminated points, respectively. By construction, the first
$m-\lceil mb\rceil+1$ overlapping blocks $z^{O}_{(m),i}$,
$i=1,\dots,m-\lceil mb\rceil+1$, contain $\lceil mb\rceil$ outliers.
Consequently, $T(z^{O}_{(m),i})=+\infty$, $i=1,\dots,m-\lceil
mb\rceil+1$. Assume that the first $p<r-1$ nonoverlapping blocks
$z^{N}_{(m),i}$, $i=1,\dots,p$, have the same contamination as in
(\ref{co}). Because of this contamination, the number of
statistics $T^{S\ast}_{n,m}$ which diverge to infinity is $mp-\lceil
mb\rceil+1$.

$Q_{t,n,m}^{S\ast}=+\infty$, when the proportion of statistics
$T^{S\ast}_{n,m}$ with $T^{S\ast}_{n,m}=+\infty$ is larger than
$(1-t)$. Therefore,

\begin{equation*}
b^{S\ast}_{t,n,m}\le \inf_{\{p\in\mathbb{N}, p\le
r-1\}}\bigg\{p\cdot \frac{\lceil mb\rceil}{n} \bigg\vert
\frac{mp-\lceil mb\rceil+1}{n-m+1}>1-t\bigg\}.
\end{equation*}

Finally, we consider the bootstrap case.
The proof of the first inequality 
follows the same lines as the
proof for the subsampling case. We focus on the
second inequality.

Consider
$z^{N}_{(m),i}$, $i=1,\dots,r$. Assume that $p_{2}$ of these
nonoverlapping blocks are contaminated with exactly $p_{1}$ outliers
for each block, while the remaining $(r-p_{2})$ are noncontaminated
($0$ outlier), where $p_{1},p_{2}\in\mathbb{N}$ and $p_{1}\le m$,
$p_{2}\le r-1$. Moreover, also assume that the contamination of the $p_{2}$ contaminated
blocks has the structure defined in (\ref{co}). The block bootstrap constructs a $n$-sample randomly selecting
with replacement $r$ overlapping blocks of size $m$. Let $X$ be the
random variable which denotes the number of contaminated blocks in
the random bootstrap sample. It follows that $X\sim
BIN(r,\frac{mp_{2}-p_{1}+1}{n-m+1})$.

By Equation (\ref{break2}), $Q_{t,n,m}^{B\ast}=+\infty$, when the
proportion of statistics $T^{B\ast}_{n,m}$ with
$T^{B\ast}_{n,m}=+\infty$ is larger than $(1-t)$. The smallest number
of outliers such that $T^{B\ast}_{n,m}=+\infty$ is by definition
$nb$. Let $p_{1},p_{2}\in\mathbb{N}, p_{1}\le m,
p_{2}\le r-1$. Consequently,
\begin{equation*}
b^{B\ast}_{t,n,m}\le
\frac{1}{n}\cdot\bigg[\inf_{\{p_{1},p_{2}\}}\bigg\{p=p_{1}\cdot
p_{2}\bigg\vert P\bigg(
BIN\bigg(r,\frac{mp_{2}-p_{1}+1}{n-m+1}\bigg)\ge\bigg\lceil\frac{nb}{p_{1}}\bigg\rceil\bigg)>1-t\bigg\}\bigg].
\end{equation*}
This concludes the proof of Theorem \ref{bsubboot}.
\end{proof}
\vspace{0.5cm}

\begin{proof}[Proof of Theorem \protect\ref{ftheo}]
Since the estimating function $\psi_{n}$ is bounded, it turns out that
\begin{equation}\label{term}
[\hat{\Sigma}_{k,i}^{R\ast}]^{-1/2}[\nabla_{\theta}\psi_{k}(z_{(k),i}^{\ast},\hat{\theta}_{n}^R)]^{-1}
\psi_{k}(z_{(k),i}^{\ast},\hat{\theta}_{n}^R),
\end{equation}
may degenerate only when (i) $det \left(\hat{\Sigma}_{k,i}^{R\ast}\right)=0$ or (ii) $det\left(\nabla_{\theta}\psi_{k}(z_{(k),i}^{\ast},\hat{\theta}_{n}^R)\right)=0$.
Consider the function
\begin{equation}\label{huber2p}
f(z_t,\theta)=(y_{t}-\theta'
w_{t-1})w_{t-1}\cdot\min\Bigg(1,\frac{c}{\vert\vert (y_{t}-\theta'
w_{t-1})w_{t-1} \vert\vert}\Bigg).
\end{equation}
Using some algebra, we can show that
\begin{equation}\label{der}
\nabla_{\theta}f(z_t,\theta)=\left\{ \begin{array}{ll}
-(1,x_{t-1})'(1,x_{t-1}), & \textrm{if } \vert\vert (y_{t}-\theta'
w_{t-1})w_{t-1} \vert\vert\le c,\\ \mathbb{O}_{2\times2}, & \textrm{if }  \vert\vert (y_{t}-\theta'
w_{t-1})w_{t-1} \vert\vert> c,
\end{array} \right.
\end{equation}
where $\mathbb{O}_{2\times2}$ denotes the $2\times2$ null matrix. It turns out that by construction the matrix
$\nabla_{\theta}(\psi_{k}(z_{(k),i}^{\ast},\hat{\theta}_n^R))$ is semi-positive definite, and in particular $det(\nabla_{\theta}(\psi_{k}(z^{\ast}_{(k),i},\hat{\theta}_n^R))=0$, only when $\vert\vert (y_{t}-\hat{\theta}'^R_n
w_{t-1})w_{t-1} \vert\vert> c$, for all the observations $(y_t,w_{t-1})'$ in the random sample $z^{\ast}_{(k),i}$.

%

For the original sample, consider following type of contamination
\begin{equation}\label{cof}
z_{(n)}=(z_{1},\dots,z_{j},C_{j+1},\dots,C_{j+p},z_{j+p+1},\dots,z_n),
\end{equation}
where $z_{i}$, $i=1,\dots,j$ and $i=j+p+1,\dots,n$ and $C_{i}$,
$i=j+1,\dots, j+p$, denote the noncontaminated and
contaminated points, respectively, where $p\ge m$. It turns out that all the $p-m+1$ overlapping blocks of size $m$
\begin{equation}
(C_{j+i},\dots, C_{j+i+m-1}),
\end{equation}
$i=1,\dots,p-m+1$ contain only outliers. Therefore, for these $p-m+1$ blocks we have that $det\left(\nabla_{\theta}\psi_{m}(C_{j+i},\dots, C_{j+i+m-1},\hat{\theta}_{n}^R)\right)=0$, i.e., some components of vector (\ref{term}) may degenerate to infinity. Moreover,
$Q_{t,n,m}^{RS\ast}=+\infty$ when the proportion of statistics
$T^{RS\ast}_{n,m}$ with $T^{RS\ast}_{n,m}=+\infty$ is larger than
$(1-t)$. Therefore, $b^{RS\ast}_{t,n,m}= \inf_{\{p\in\mathbb{N}, m\le p\le
n-m+1\}}\bigg\{\frac{p}{n} \bigg\vert
\frac{p-m+1}{n-m+1}>1-t\bigg\}$, which proves the result in Equation (\ref{fbps}).

For the result in Equation (\ref{fbpb}), note that because of the contamination defined in (\ref{cof}), by construction we have $p-m+1$ overlapping blocks of size $m$ with exactly $m$ outliers, and $n-(p-m+1)$ blocks with less than $m$ outliers. Let $X$ be the random variable which denotes the number of full contaminated blocks in the random bootstrap sample. It follows that $X\sim BIN\left(r,\frac{p-m+1}{n-m+1}\right)$. To imply (i) or (ii), all the random observations $(z_1^{\ast},\dots,z_k^{\ast})$ have to be outliers, i.e., $X=r$.
By Equation (\ref{break2}), $Q_{t,n,m}^{RB\ast}=+\infty$, when the
proportion of statistics $T^{RB\ast}_{n,m}$ with
$T^{RB\ast}_{n,m}=+\infty$ is larger than $(1-t)$. Consequently,
\begin{displaymath}
b^{RB\ast}_{t,n,m}=
\frac{1}{n}\cdot\bigg[\inf_{\{p\in\mathbb{N}, p\le
n-m+1\}}\bigg\{p\bigg\vert P\bigg(
BIN\bigg(r,\frac{p-m+1}{n-m+1}\bigg)=r\bigg)>1-t\bigg\}\bigg].
\end{displaymath}
This concludes the proof.
\end{proof}

\vspace{0.5cm}

\newpage

\begin{figure}[!h]
\center $ \begin{array}{ccc}
\includegraphics[height=1.5in,width=2in]{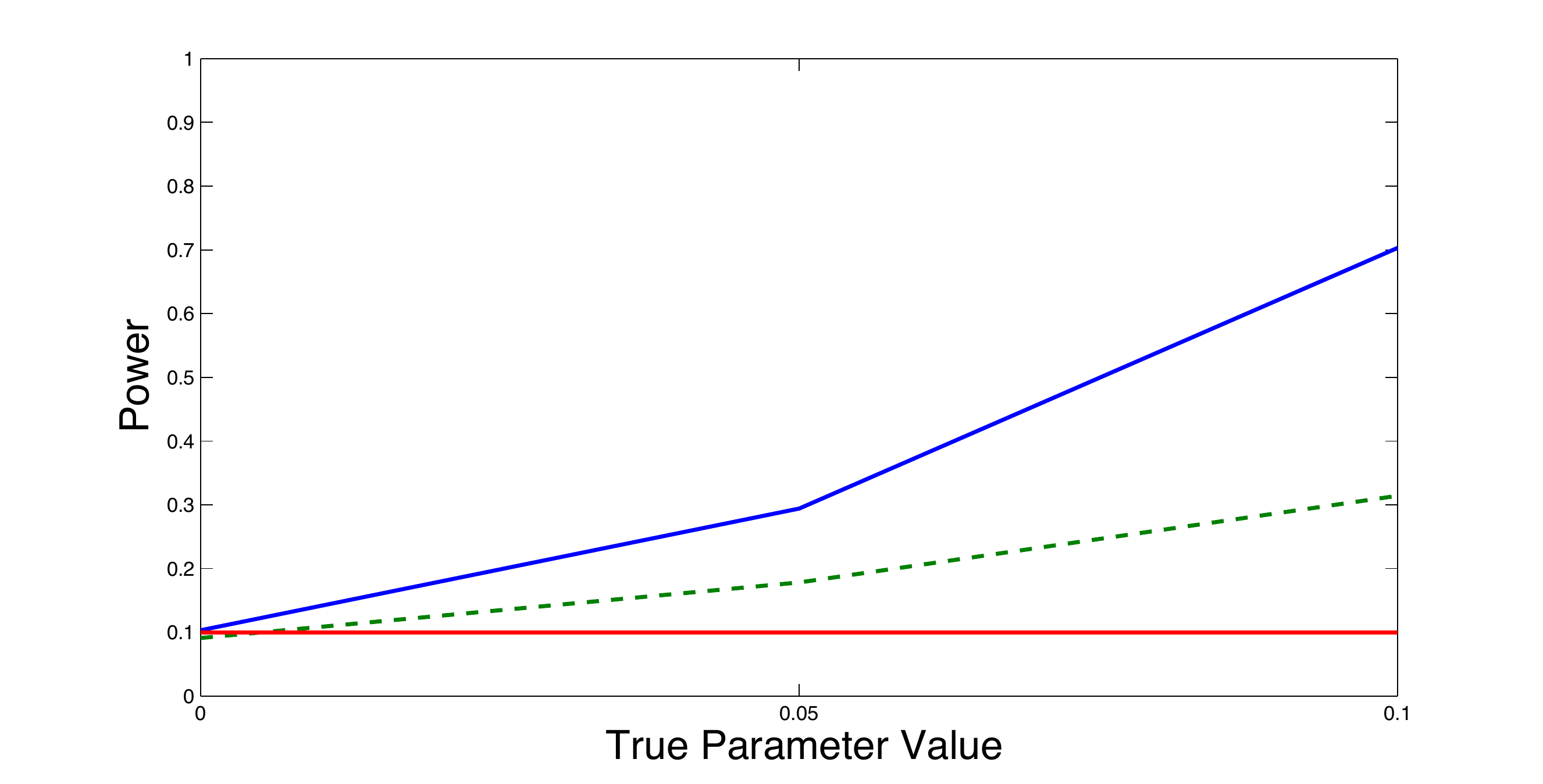} & 
\includegraphics[height=1.5in,width=2in]{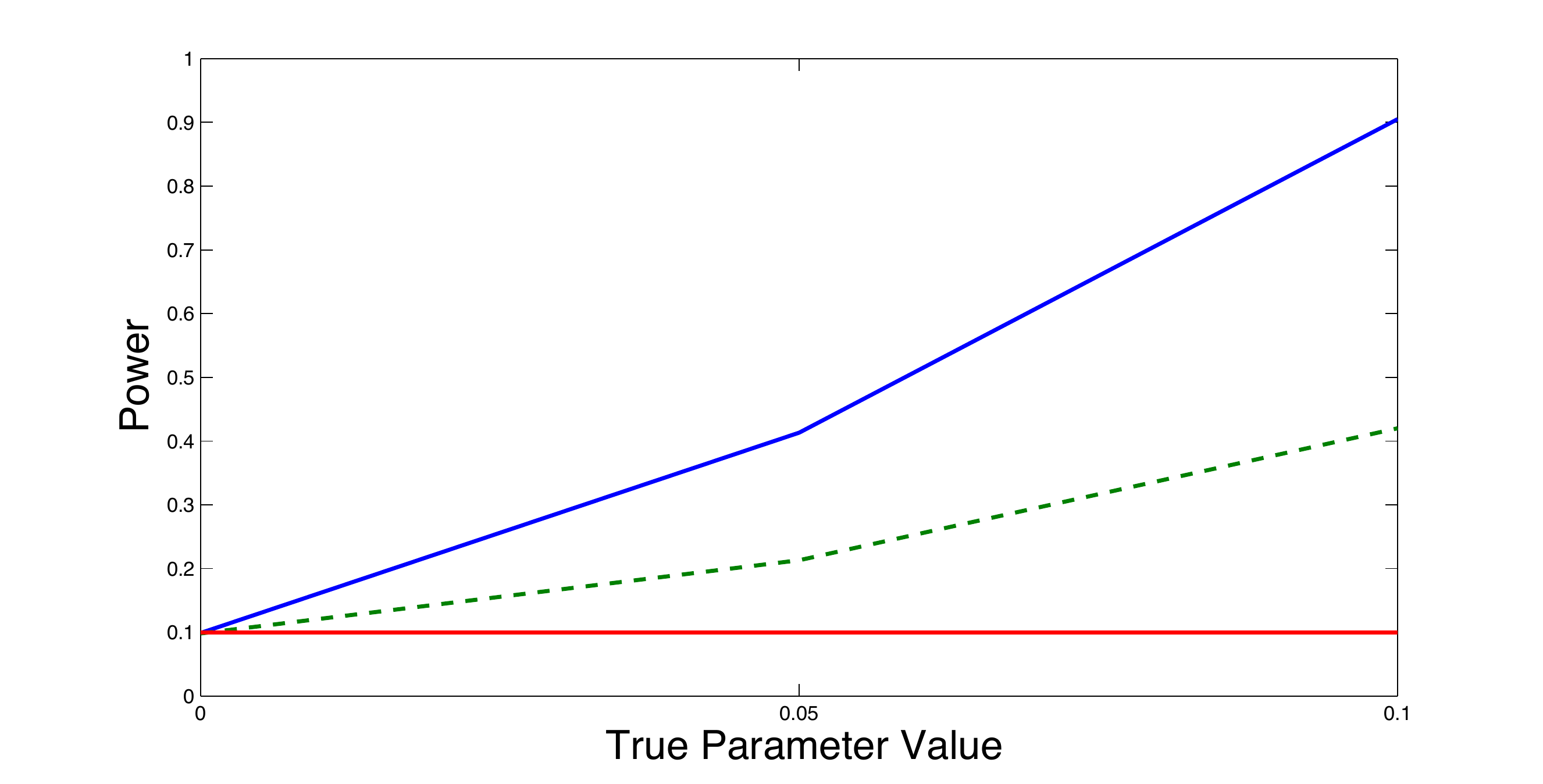} &
\includegraphics[height=1.5in,width=2in]{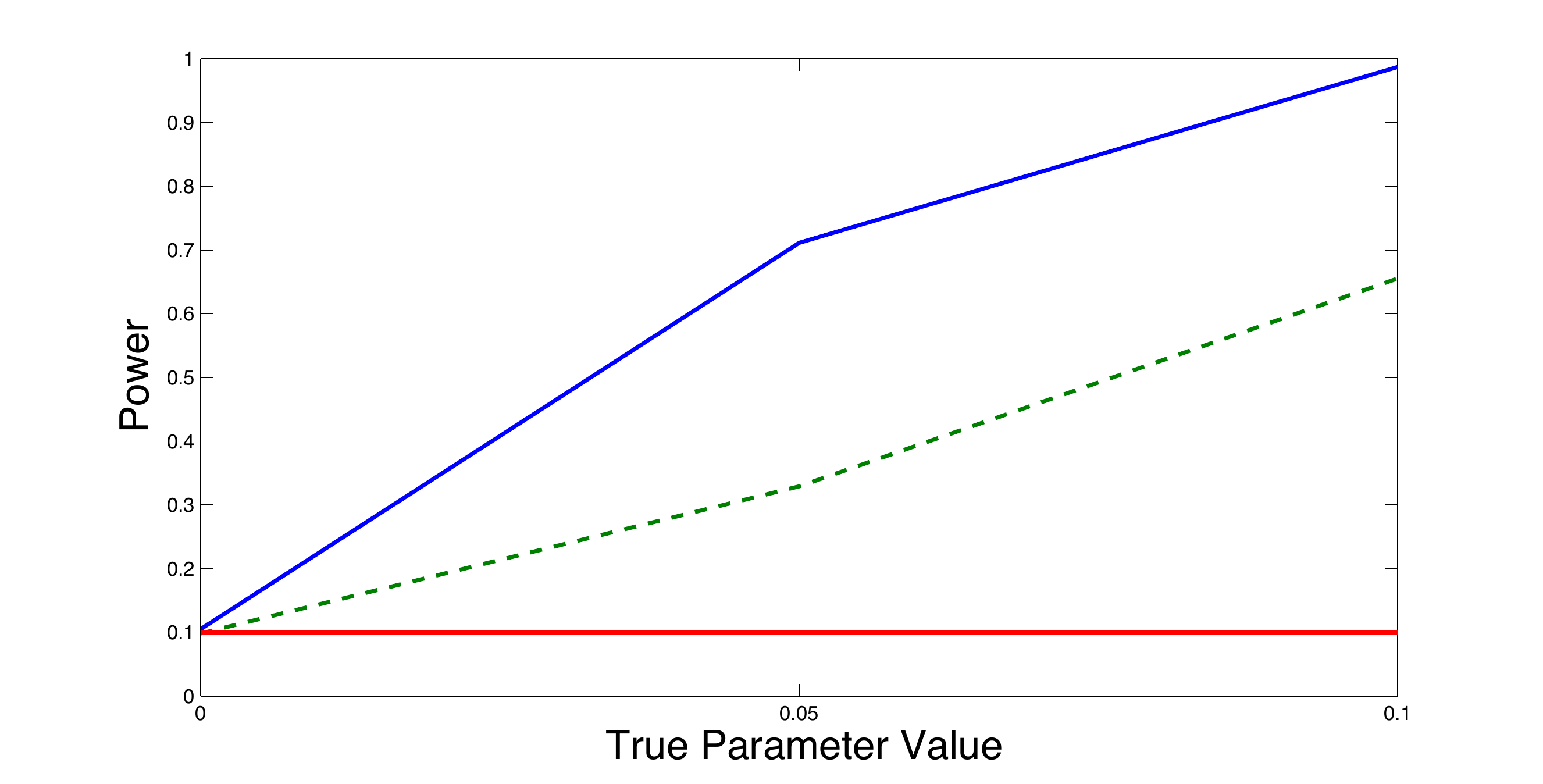} \\
\includegraphics[height=1.5in,width=2in]{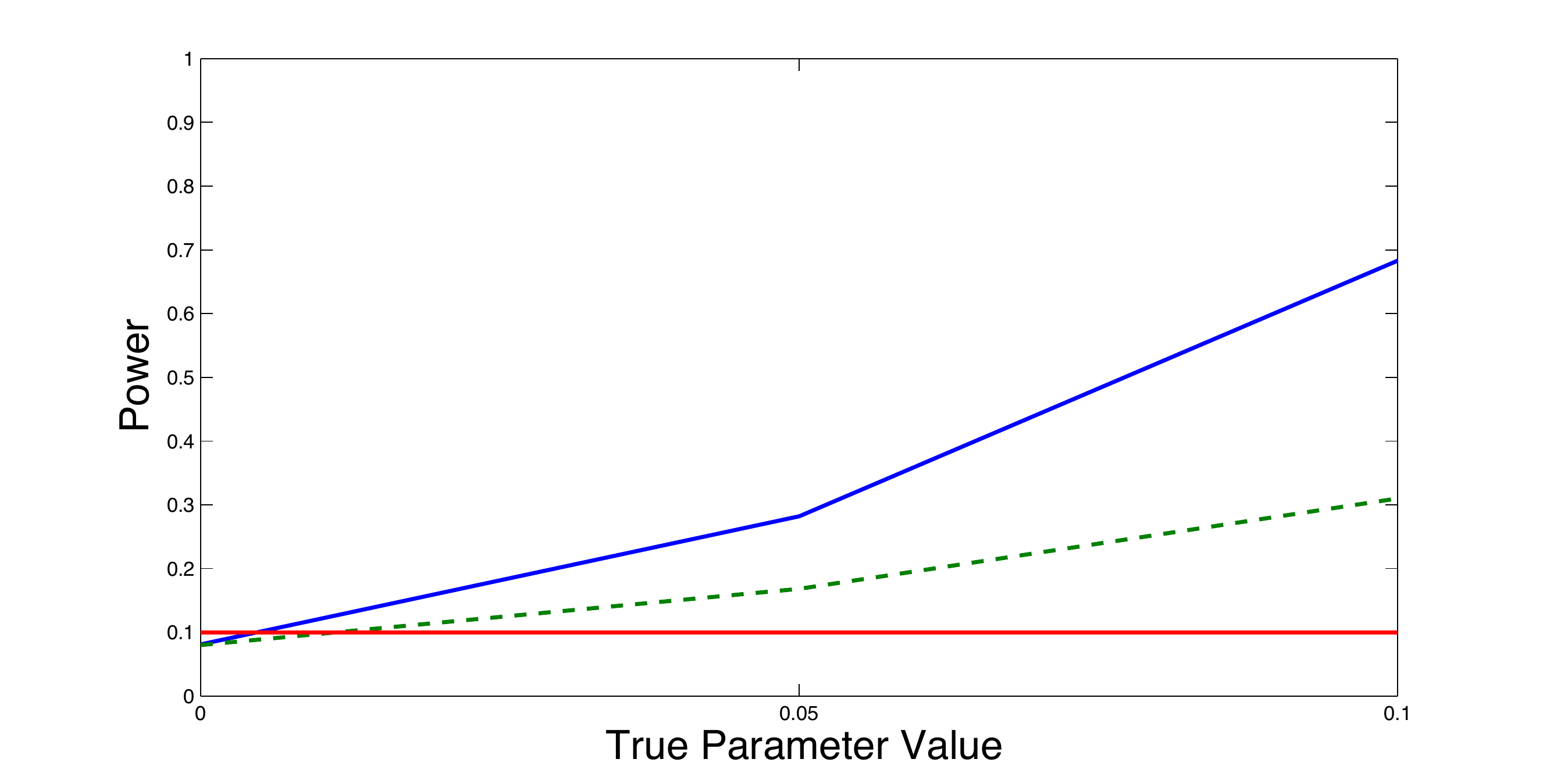} & 
\includegraphics[height=1.5in,width=2in]{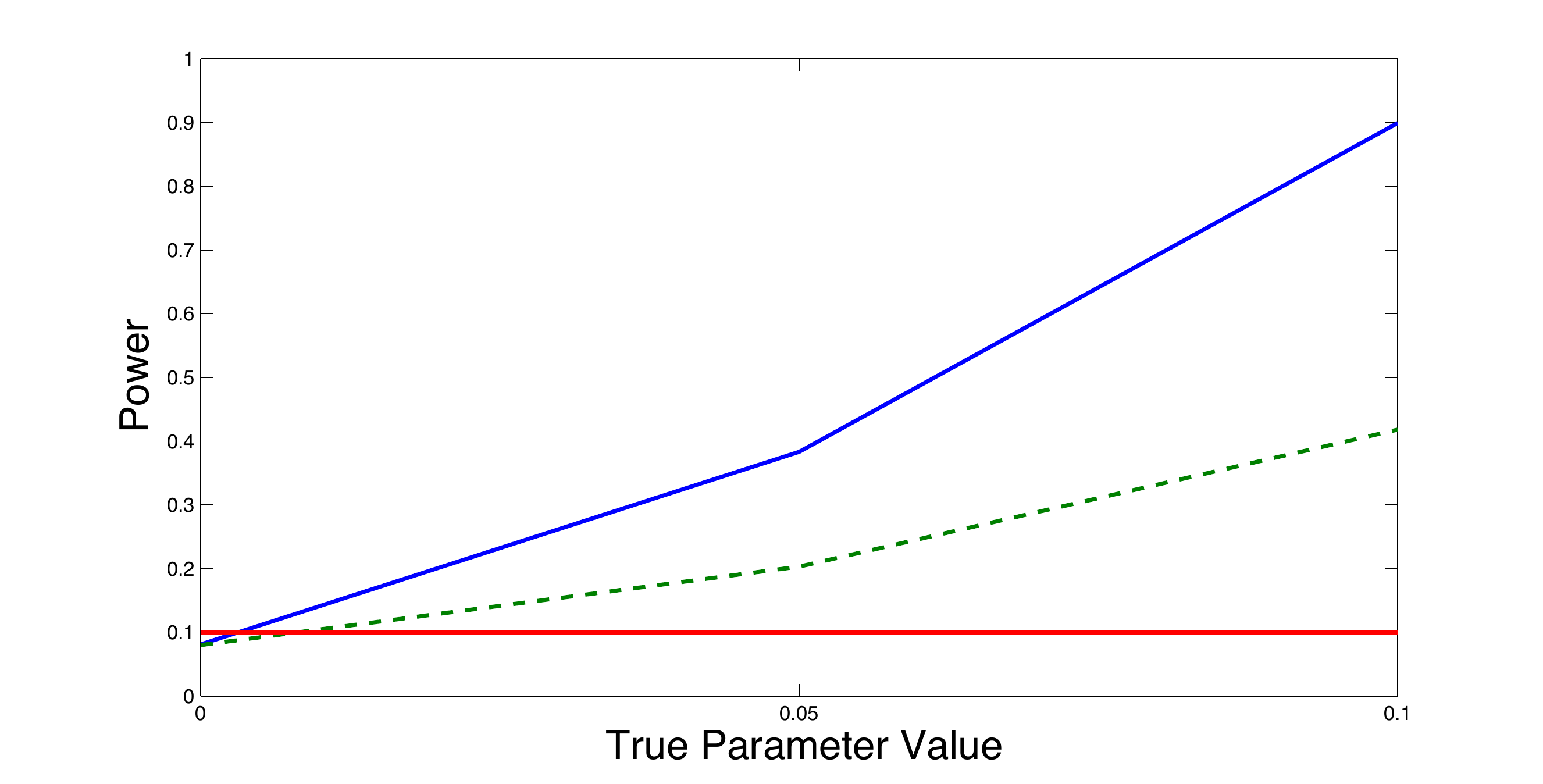} &
\includegraphics[height=1.5in,width=2in]{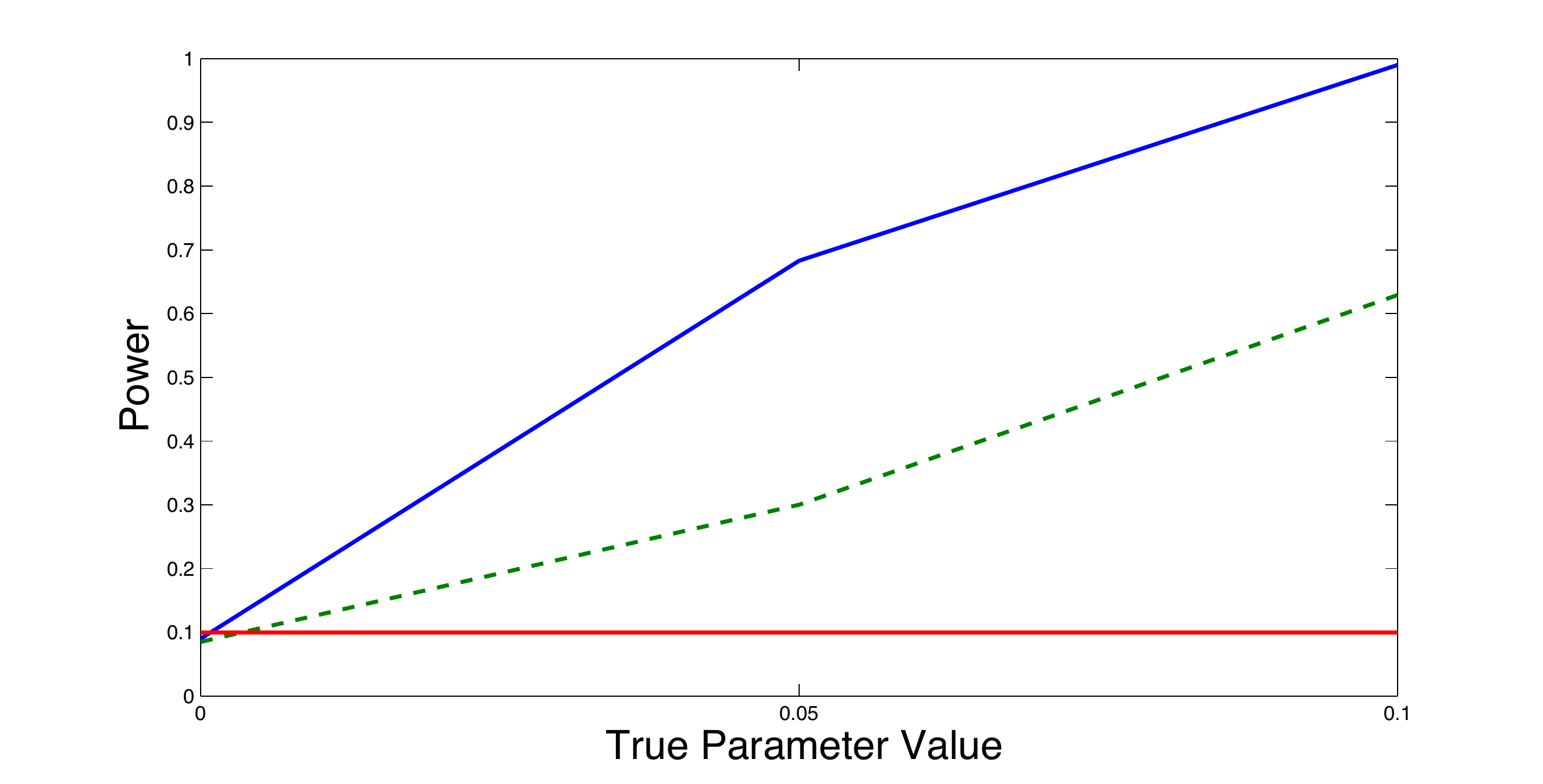} \\
\includegraphics[height=1.5in,width=2in]{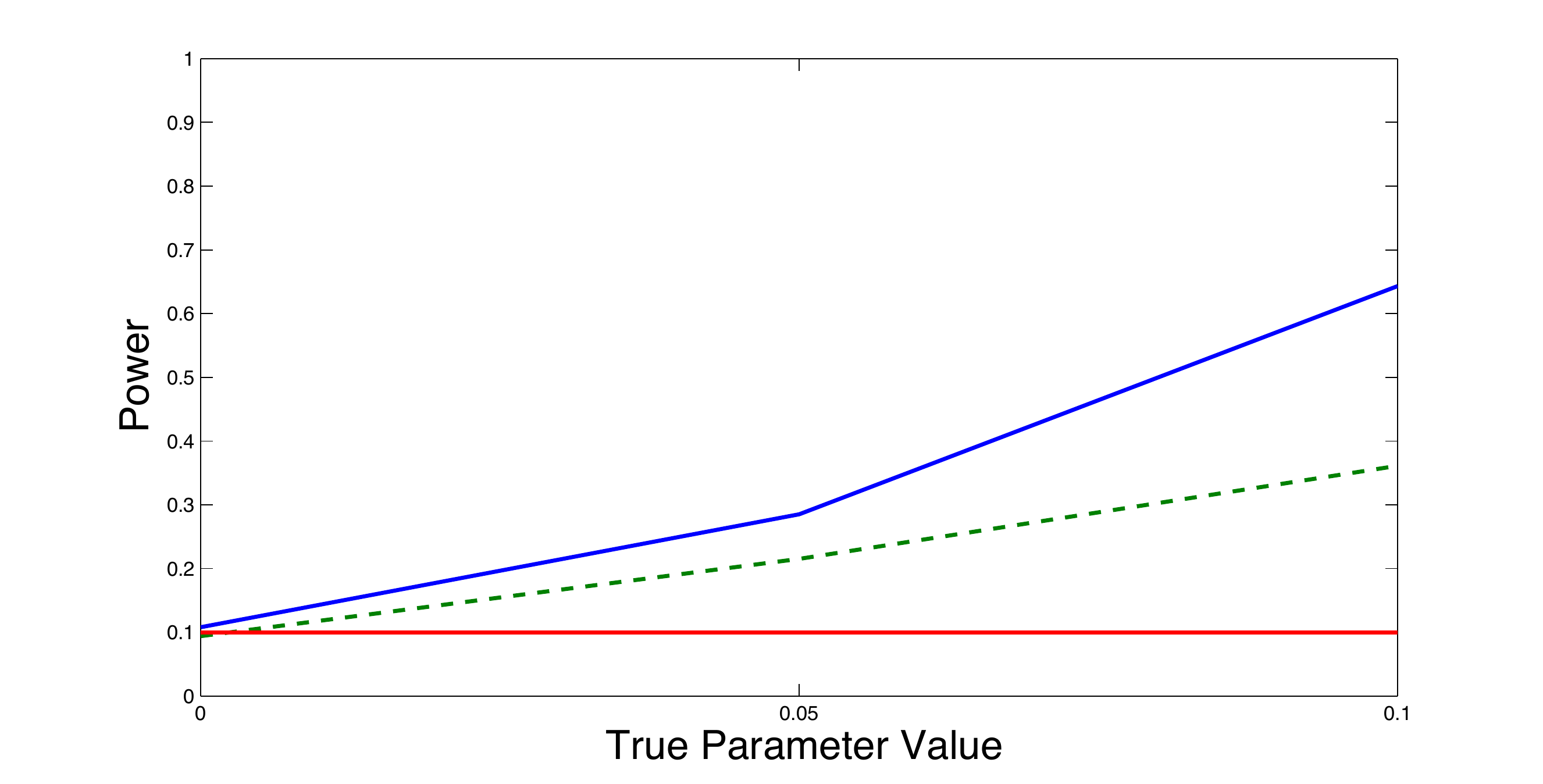} & 
\includegraphics[height=1.5in,width=2in]{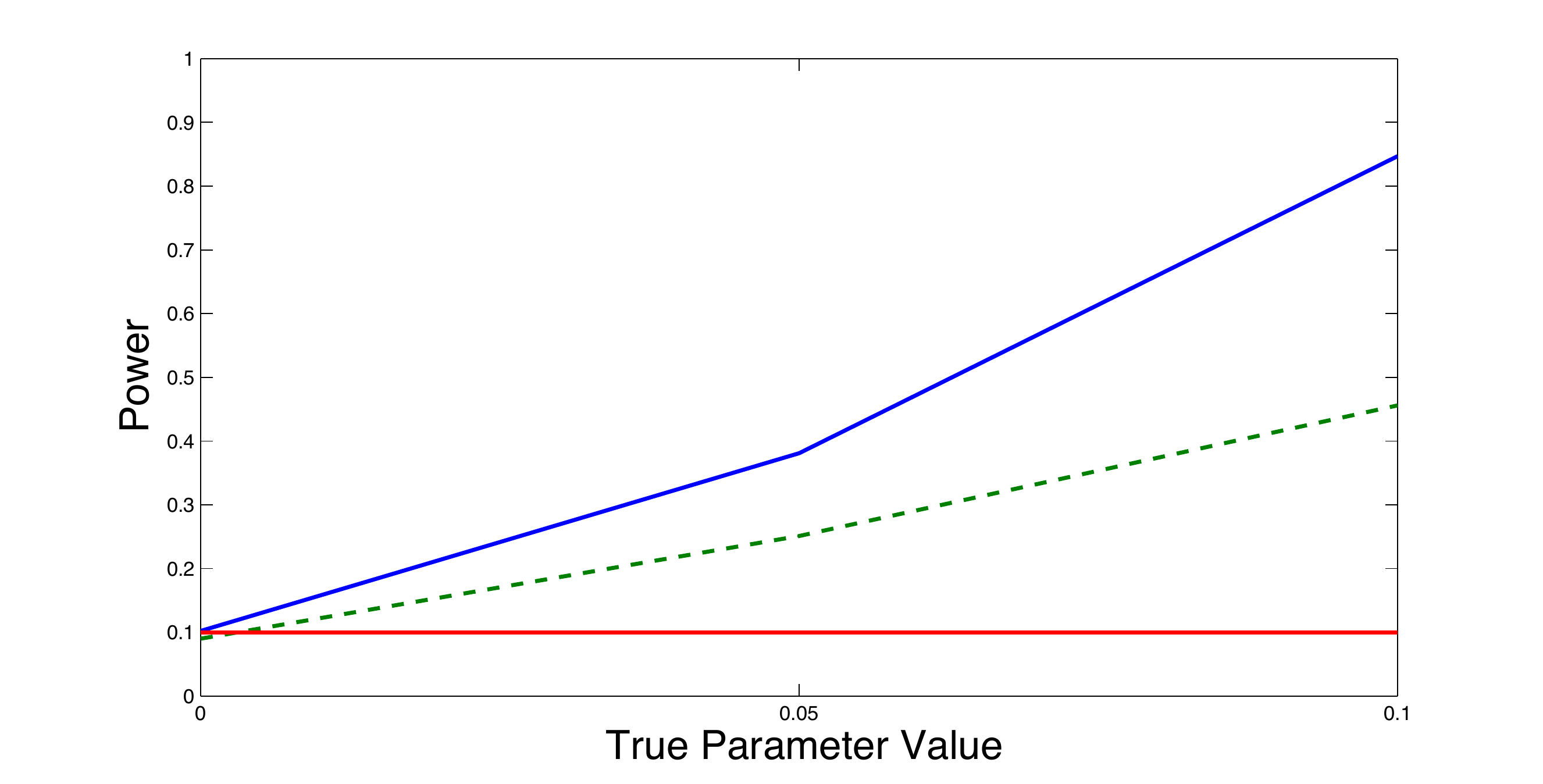} &
\includegraphics[height=1.5in,width=2in]{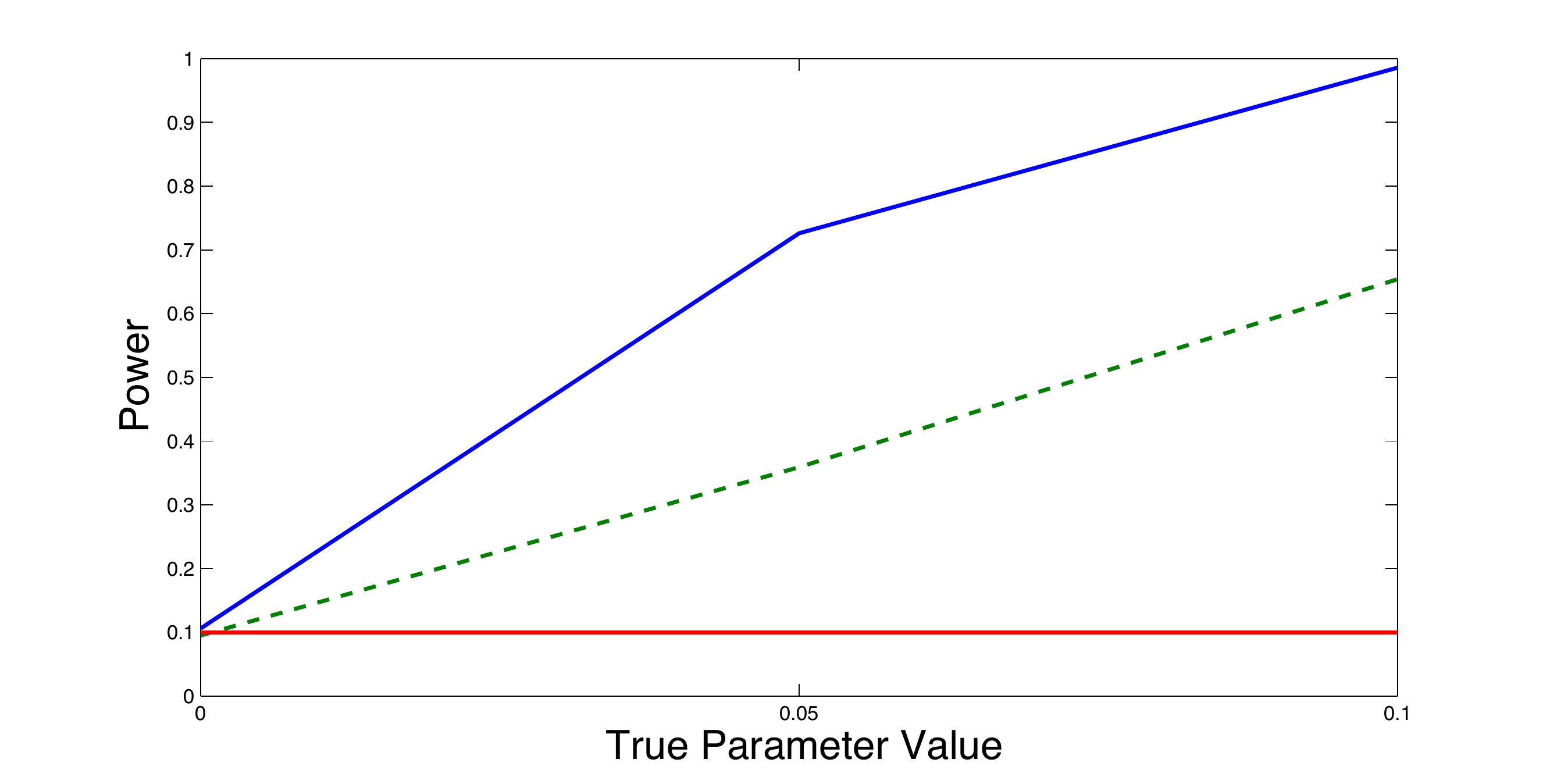} \\
\includegraphics[height=1.5in,width=2in]{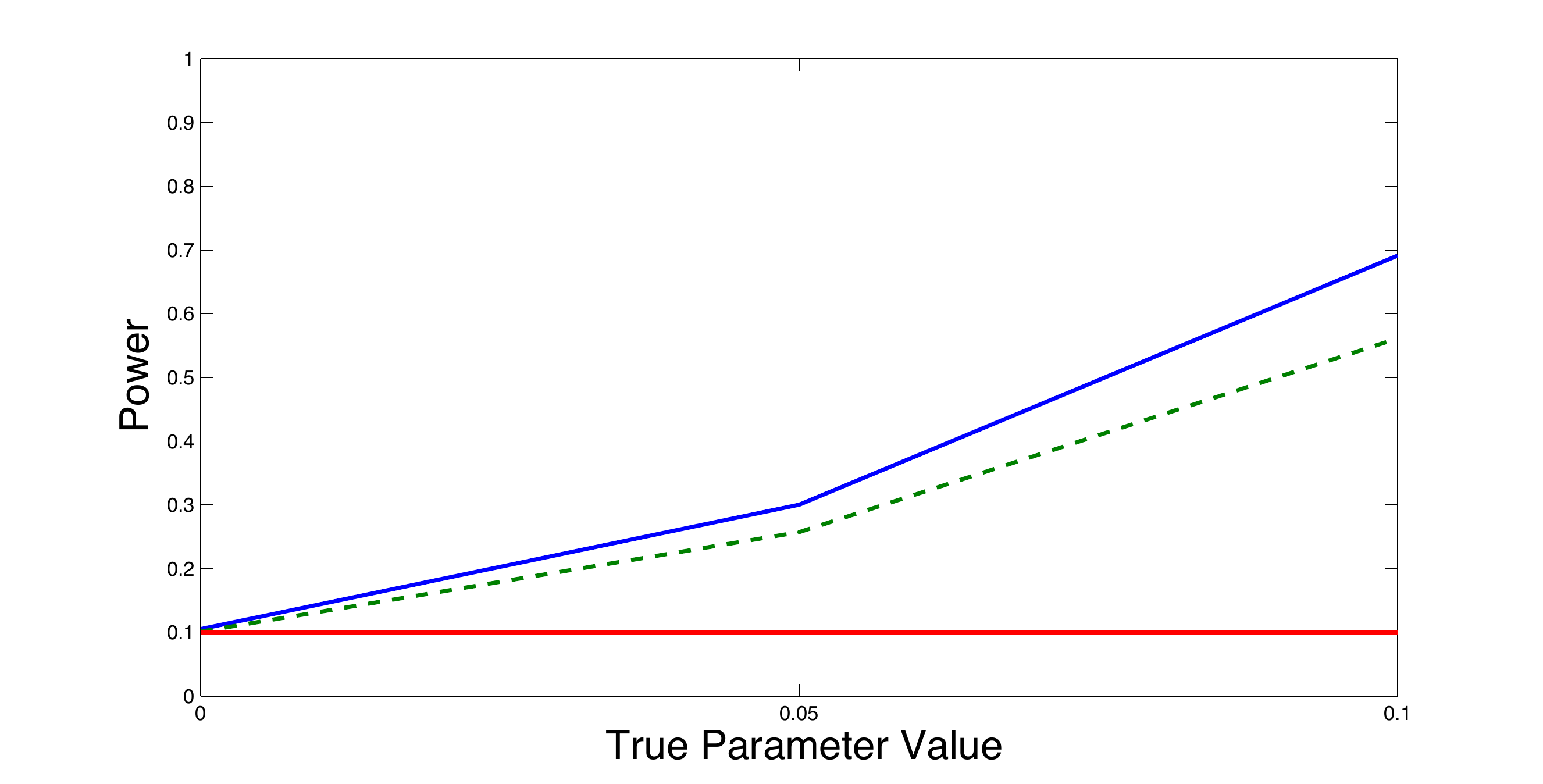} & 
\includegraphics[height=1.5in,width=2in]{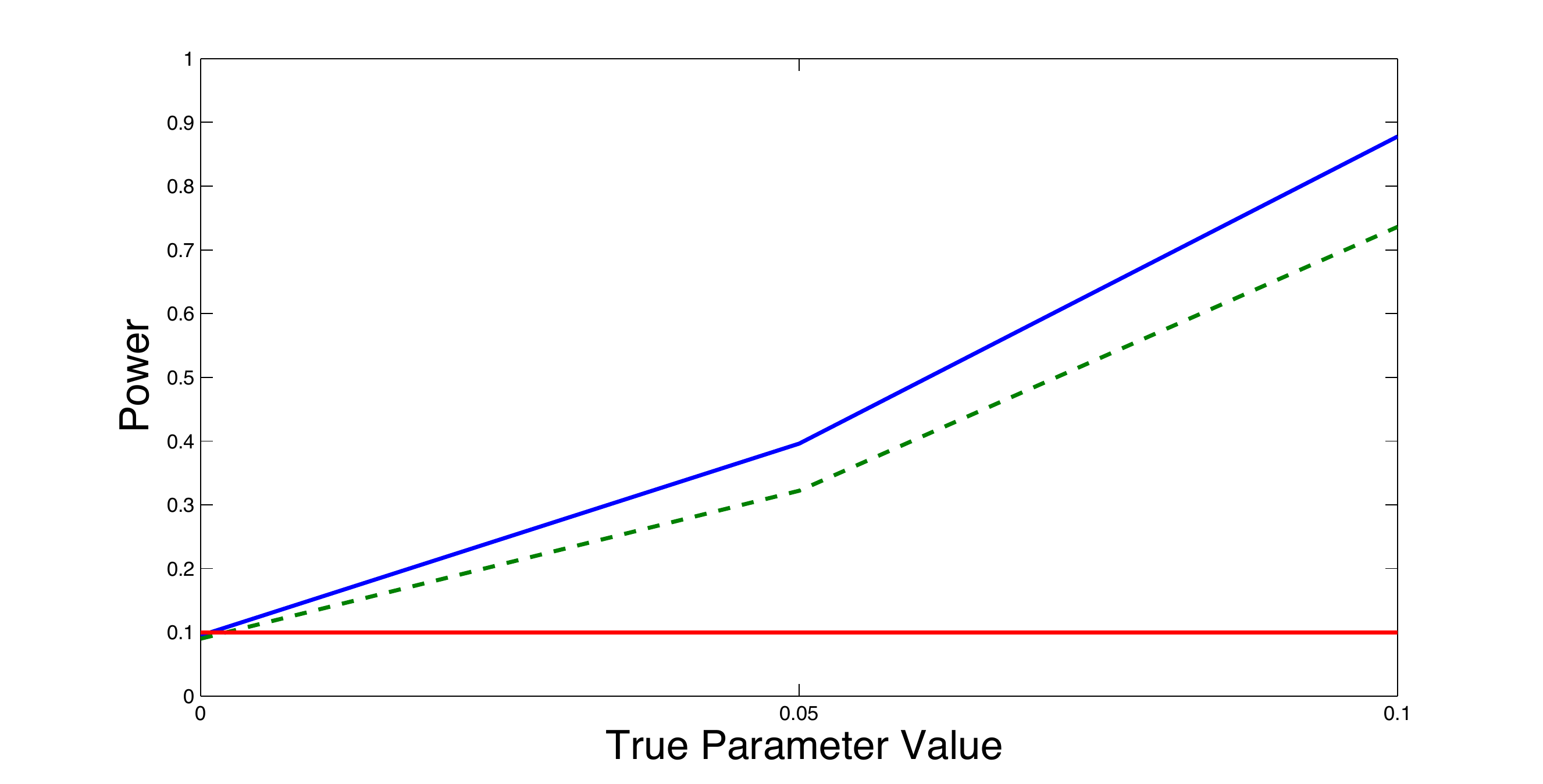} &
\includegraphics[height=1.5in,width=2in]{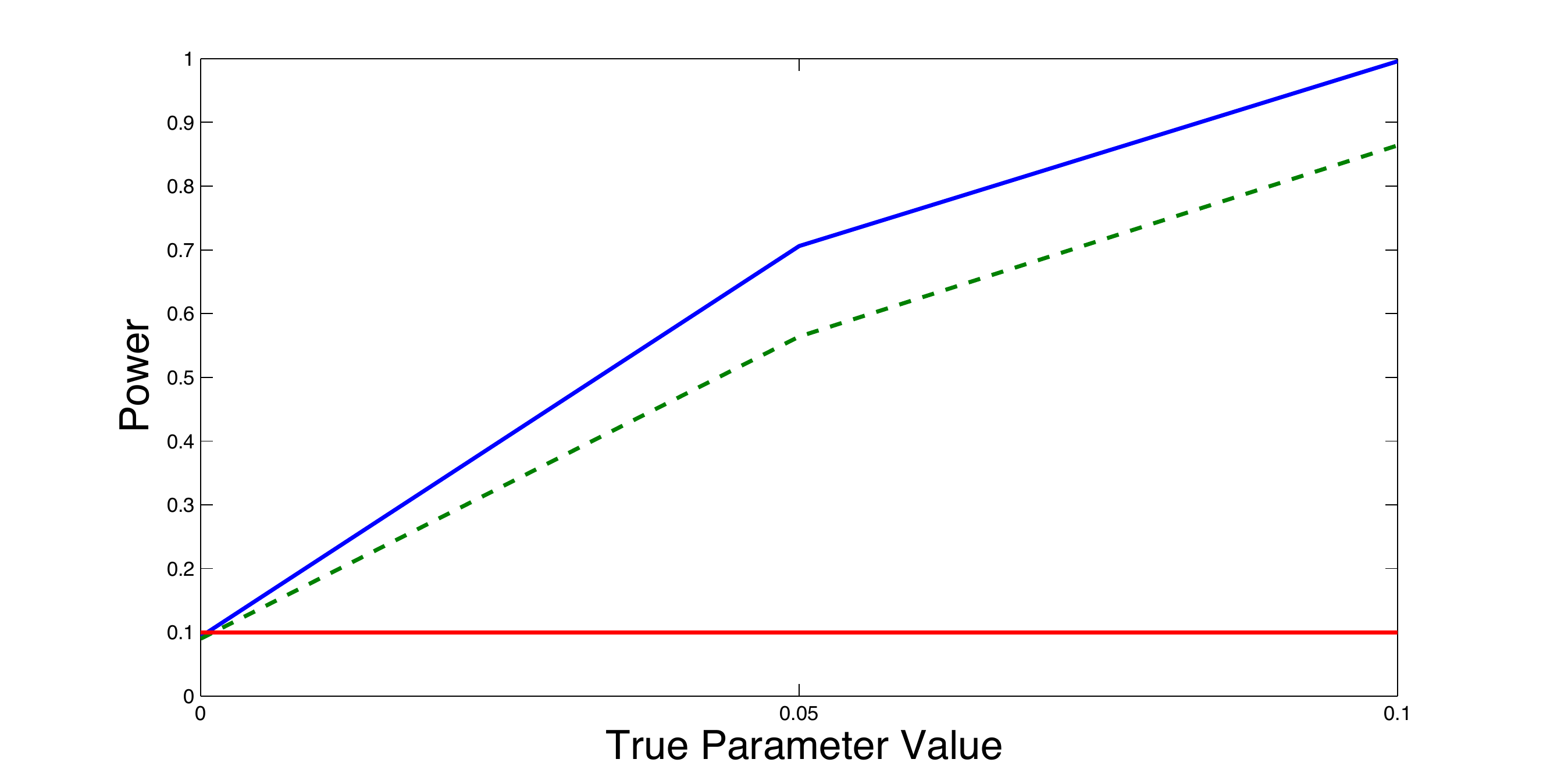} 
\end{array}$
\caption{\scriptsize{
 {\bf Power Curves.} We plot the proportion of rejections of the null hypothesis ${\cal H}_0 :
\beta_{0}=0$, when the true parameter value is
$\beta_{0}\in\{0,0.05,0.1\}$. In the first row, we consider the bias-corrected method proposed in Amihud, Hurvich and Wang (2008). In the second row, we consider the Bonferroni approach for the local-to-unity asymptotic theory introduced in Campbell and Yogo (2006). In the third row, we consider the conventional subsampling, while in the last row we present our robust subsampling. In the first, second and third columns, the degree of persistence is $\rho=0.9$, $\rho=0.95$, and $\rho=0.99$, respectively.
We consider noncontaminated samples (straight line) and contaminated samples (dashed
line).}}\label{poweracy}
\end{figure}

\begin{figure}[!h]
\center $ \begin{array}{ccc}
\includegraphics[height=1.5in,width=2in]{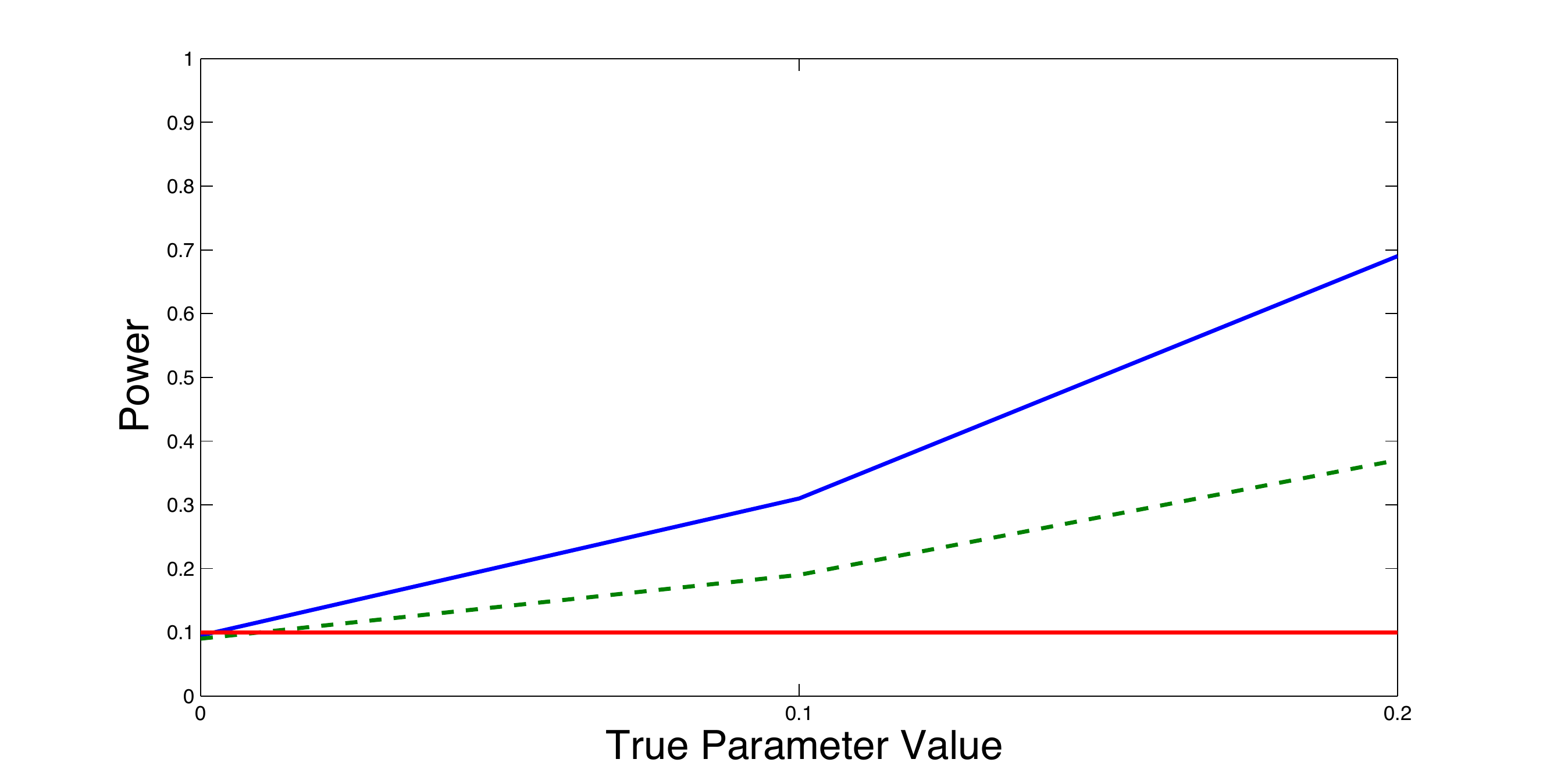} & 
\includegraphics[height=1.5in,width=2in]{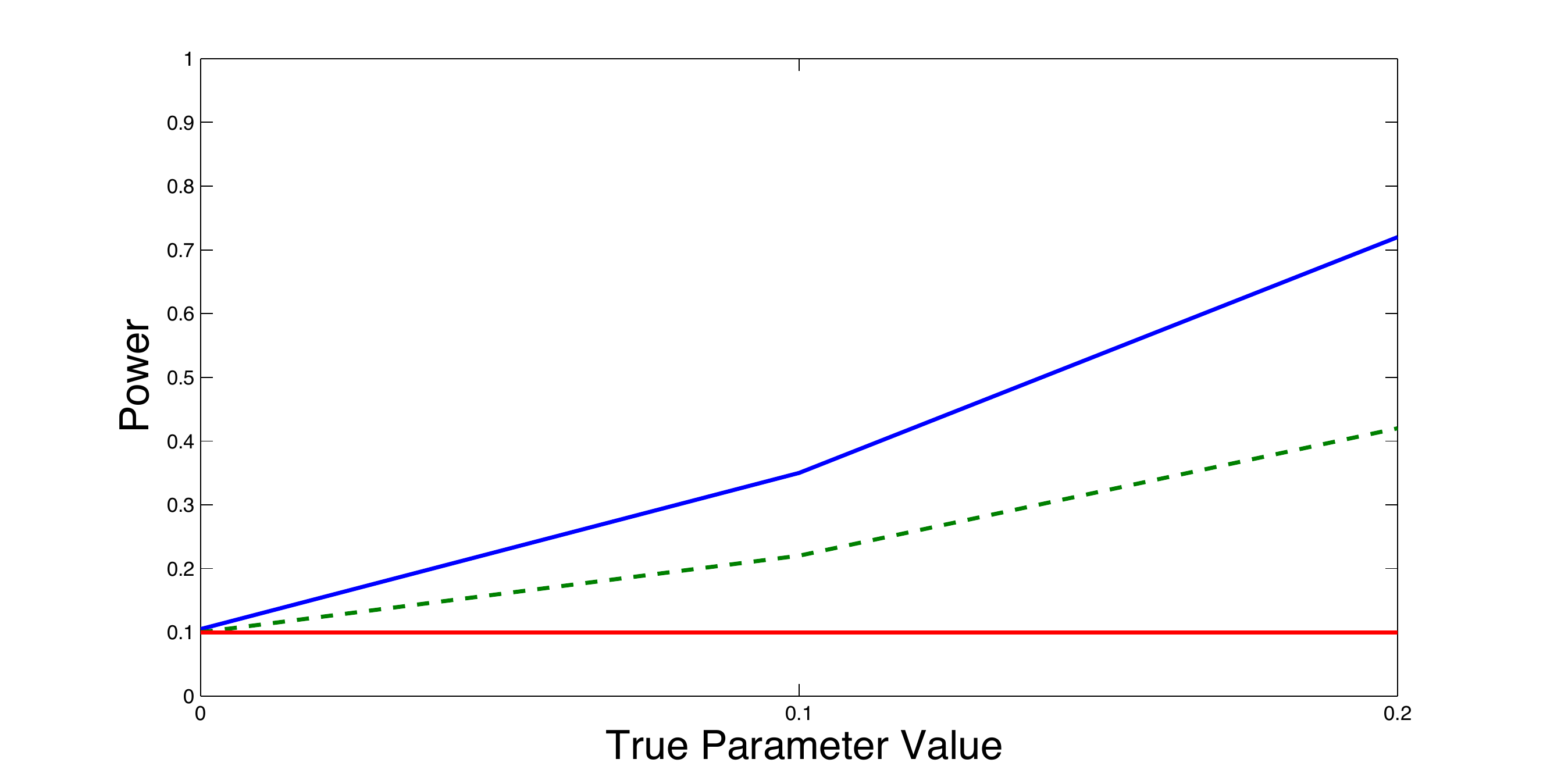} &
\includegraphics[height=1.5in,width=2in]{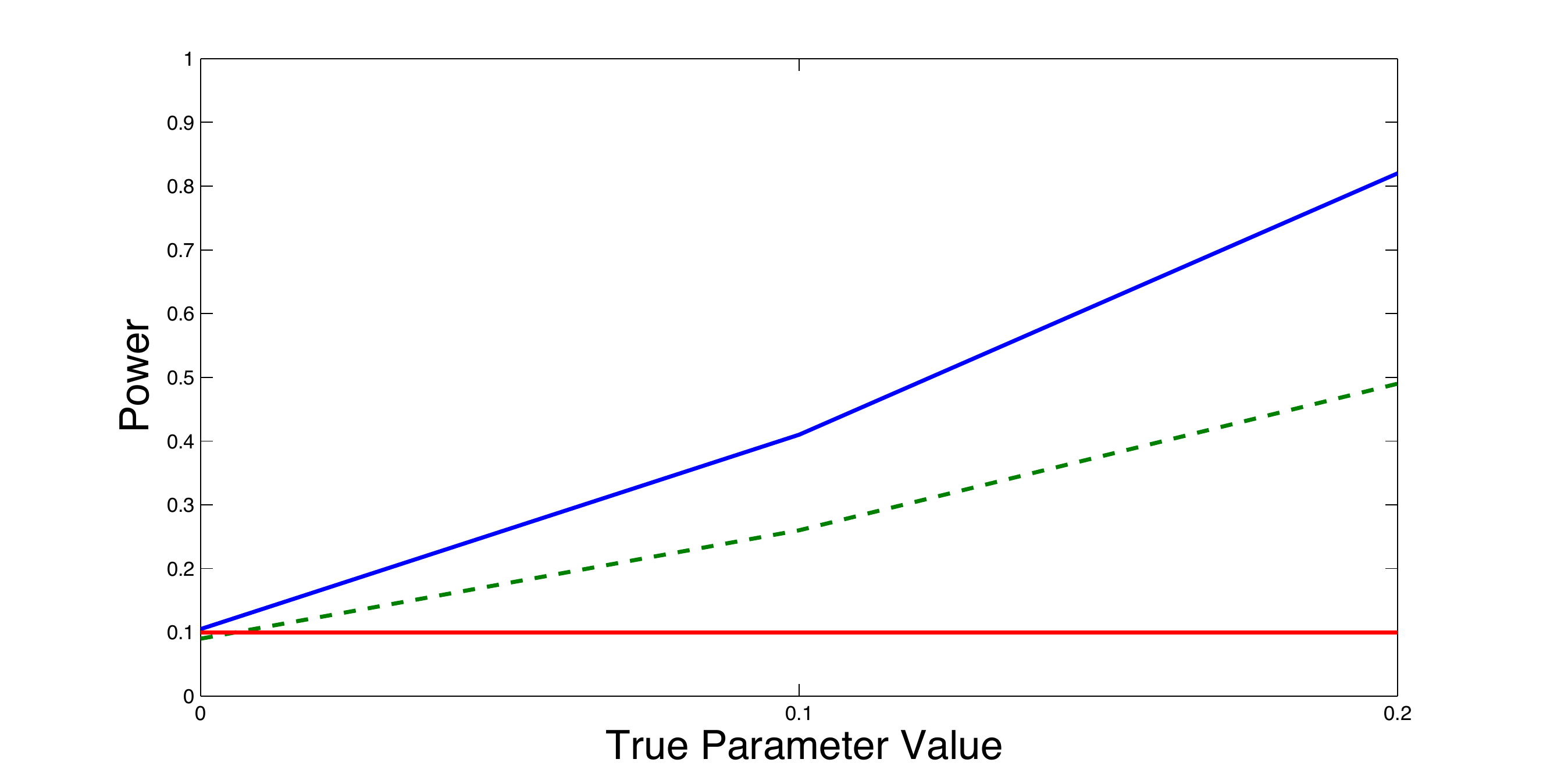} \\
\includegraphics[height=1.5in,width=2in]{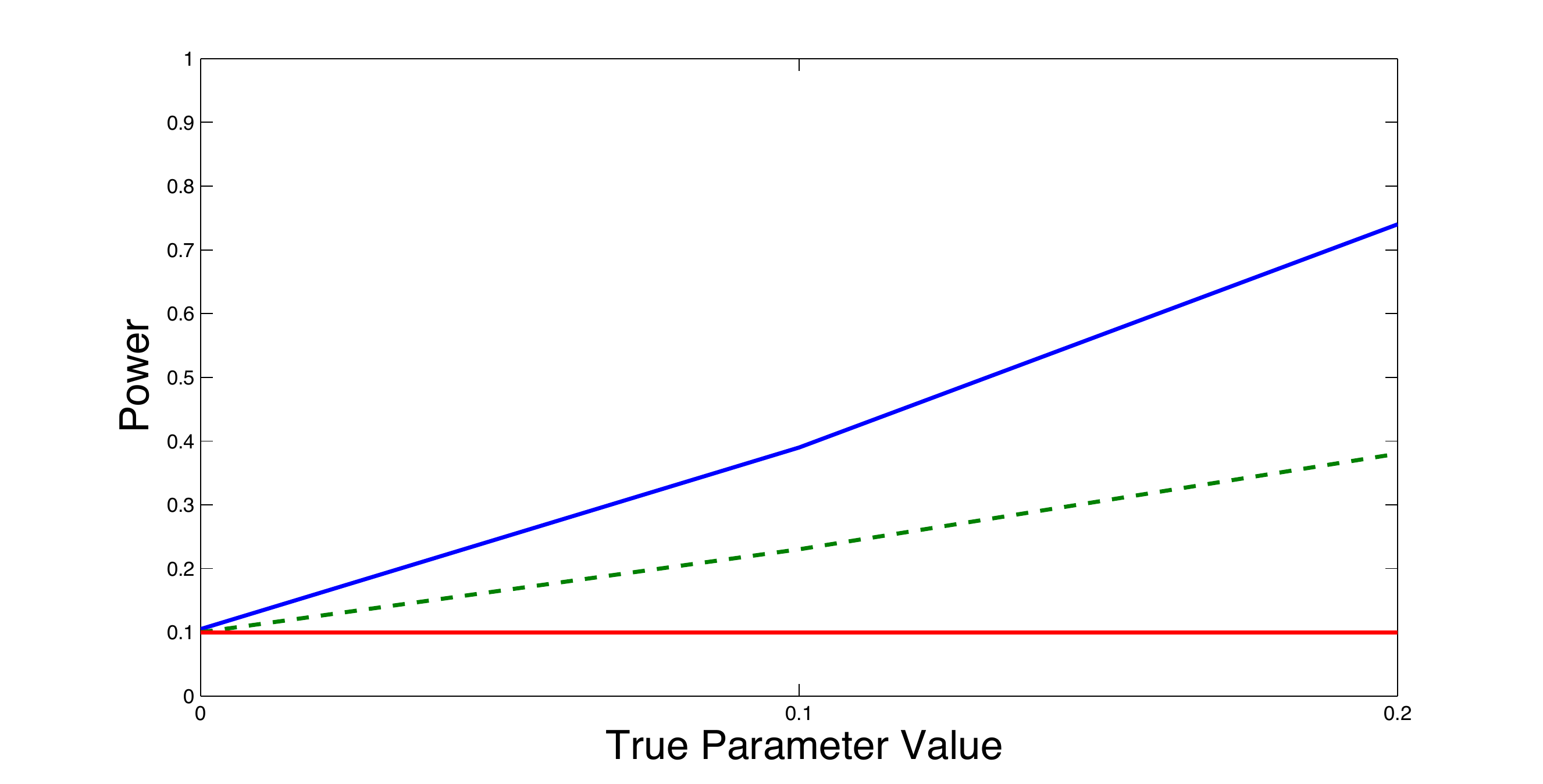} & 
\includegraphics[height=1.5in,width=2in]{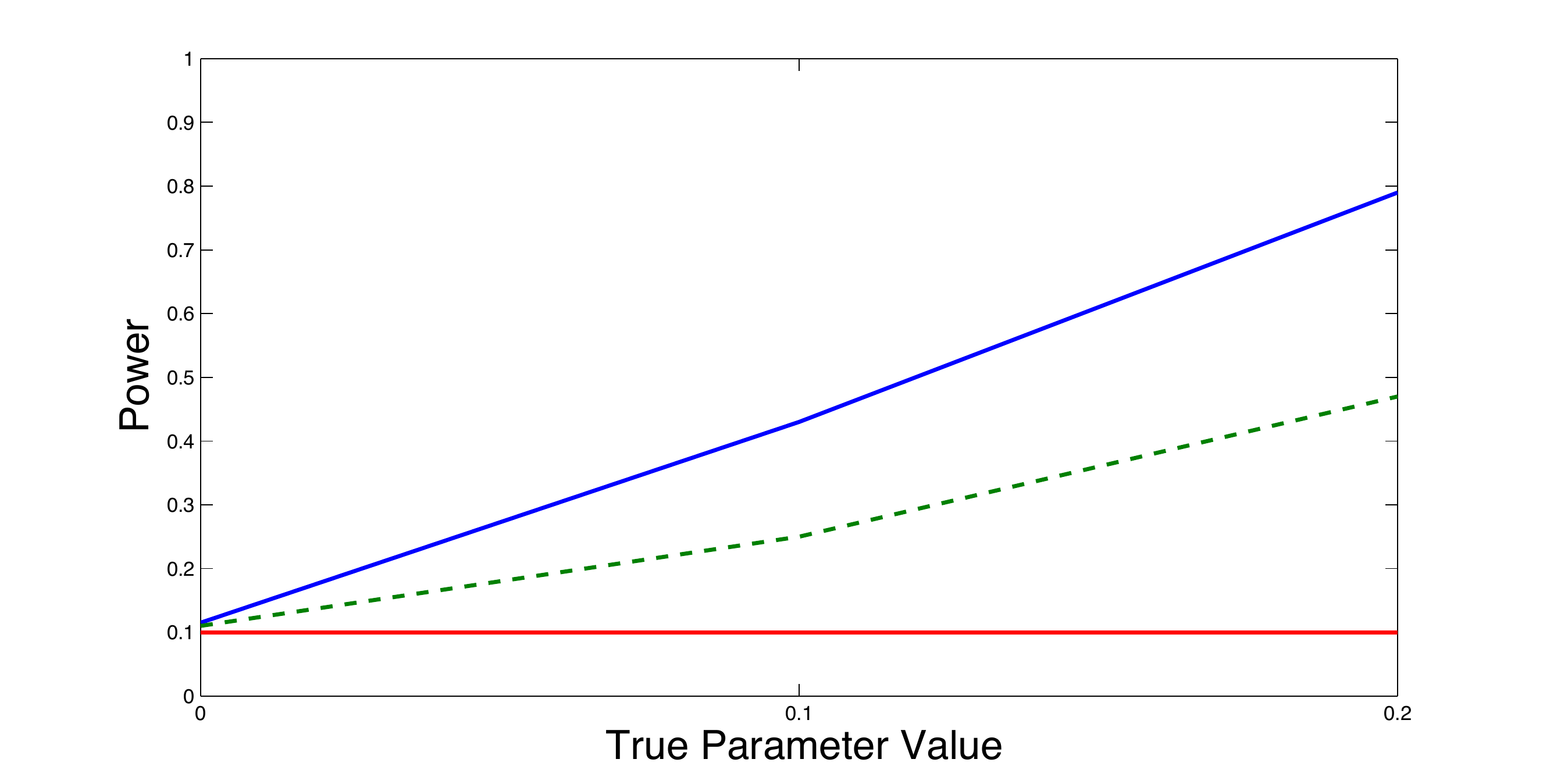} &
\includegraphics[height=1.5in,width=2in]{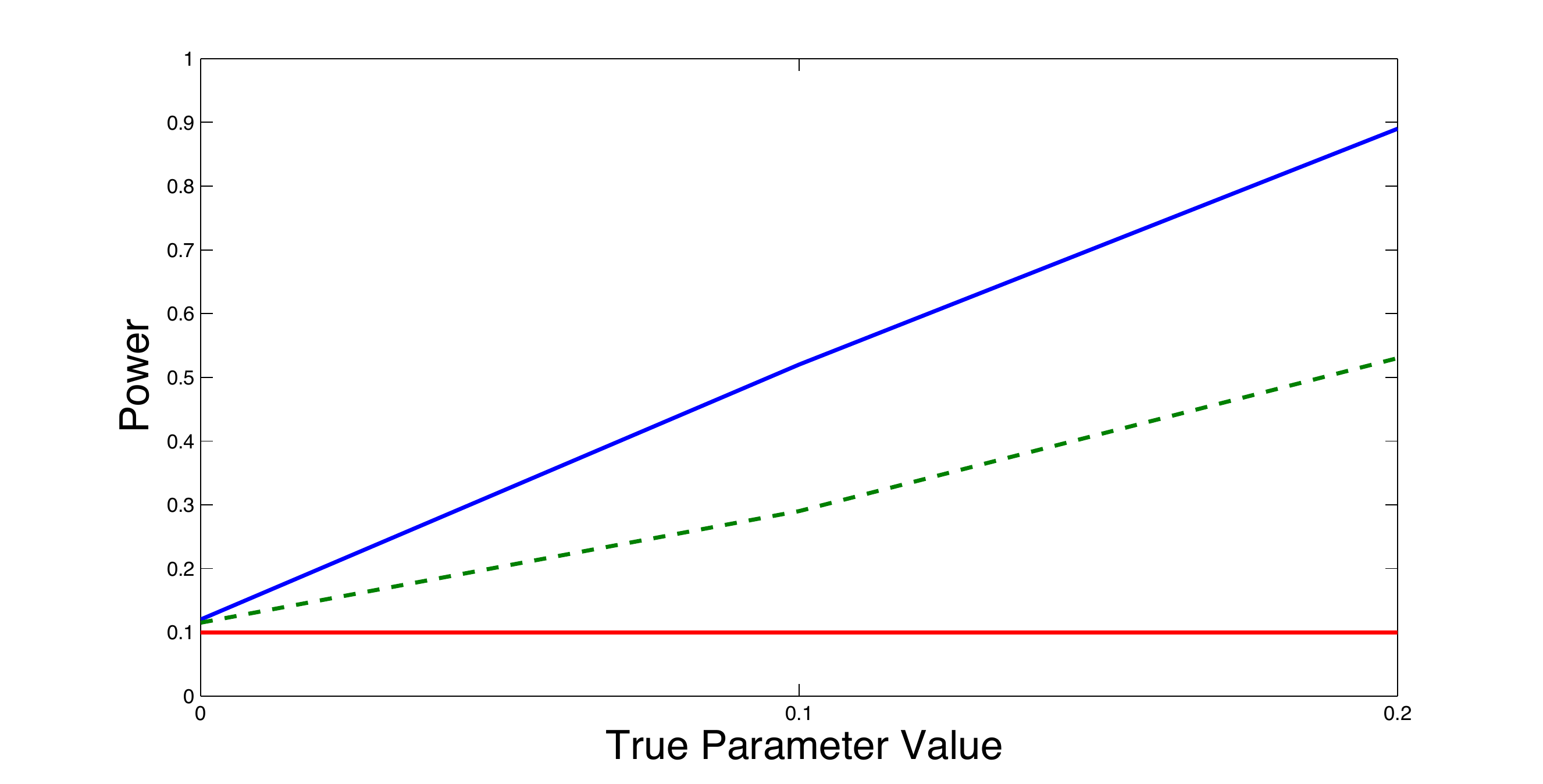} \\
\includegraphics[height=1.5in,width=2in]{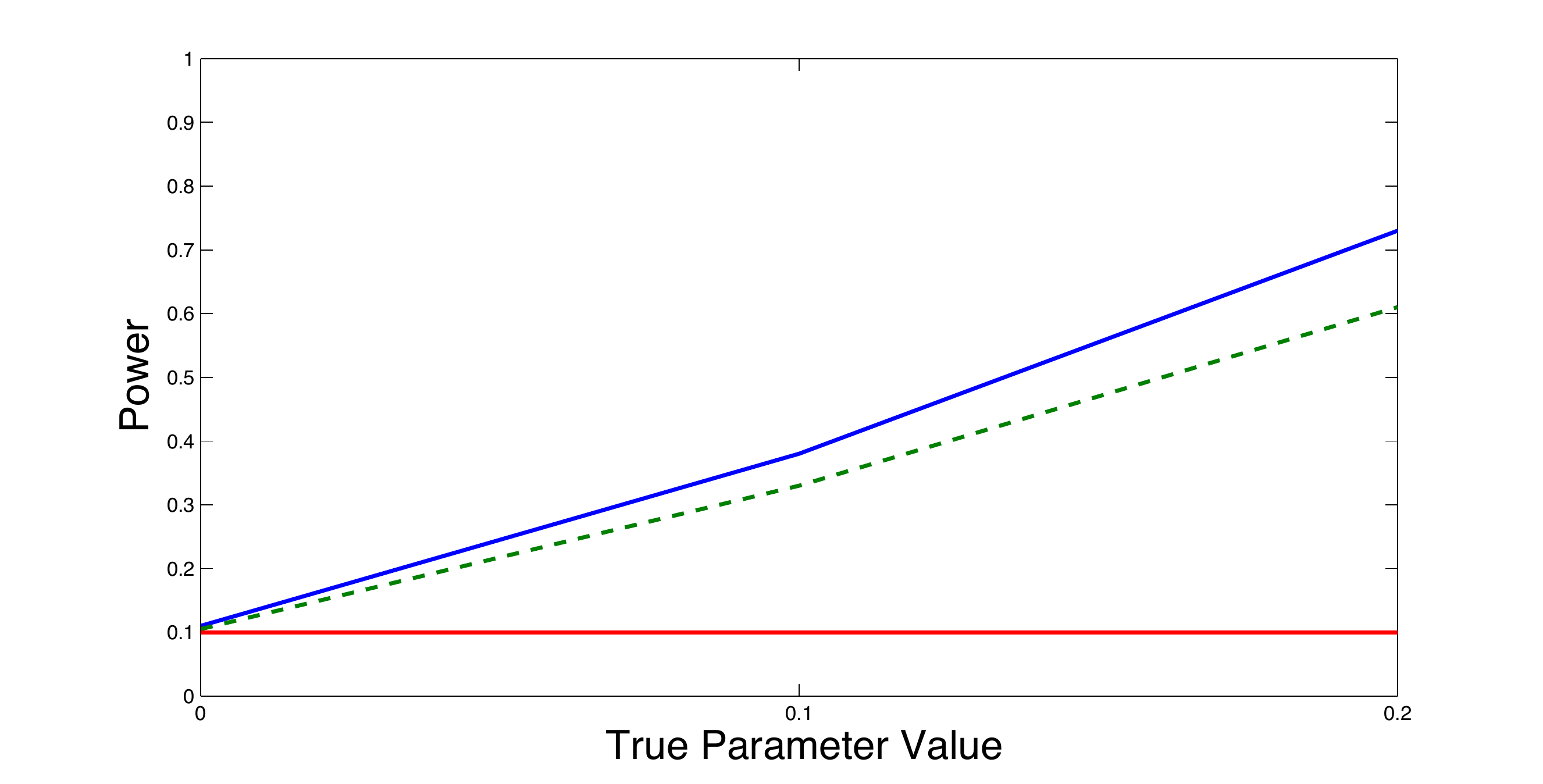} & 
\includegraphics[height=1.5in,width=2in]{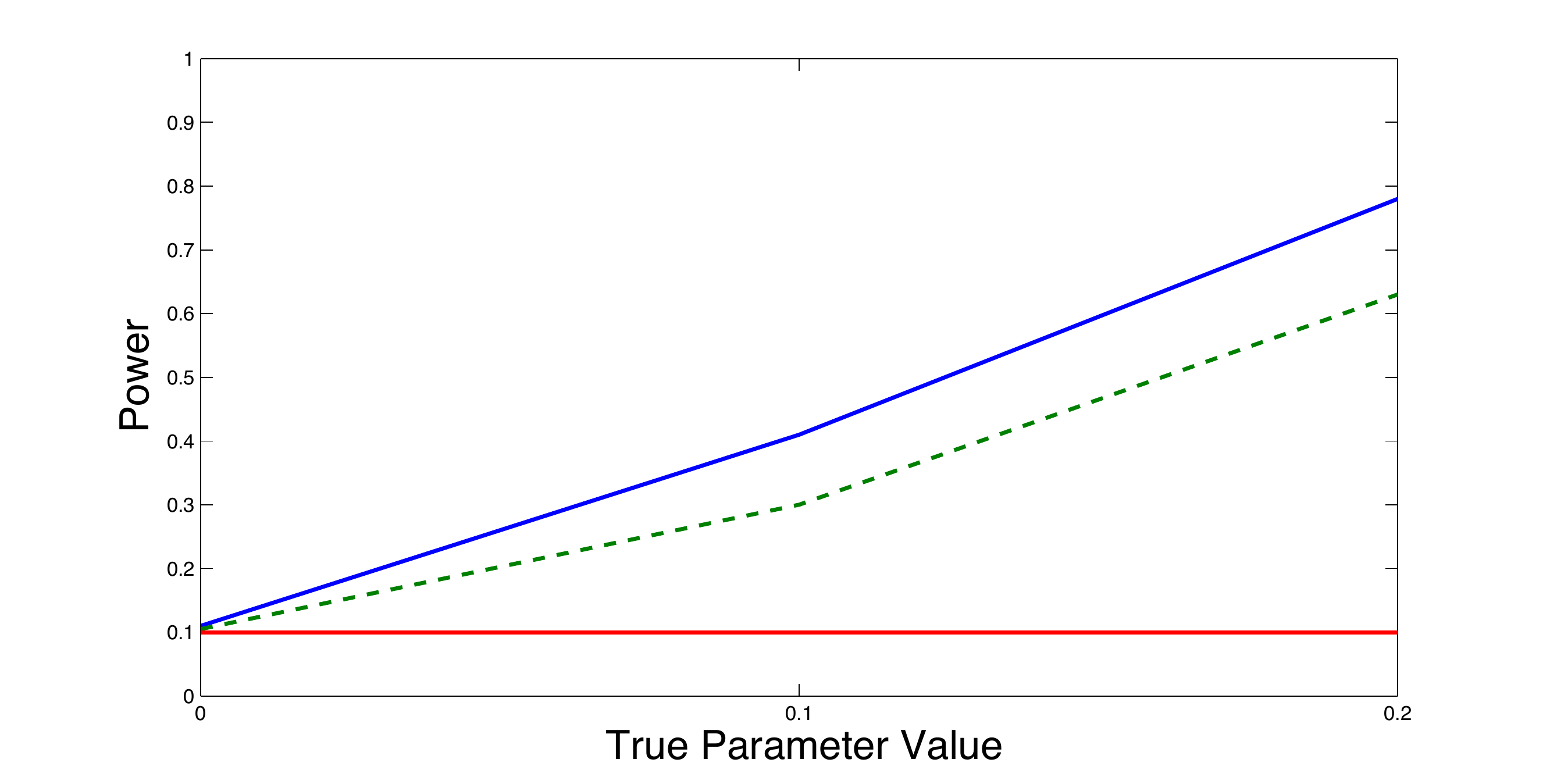} &
\includegraphics[height=1.5in,width=2in]{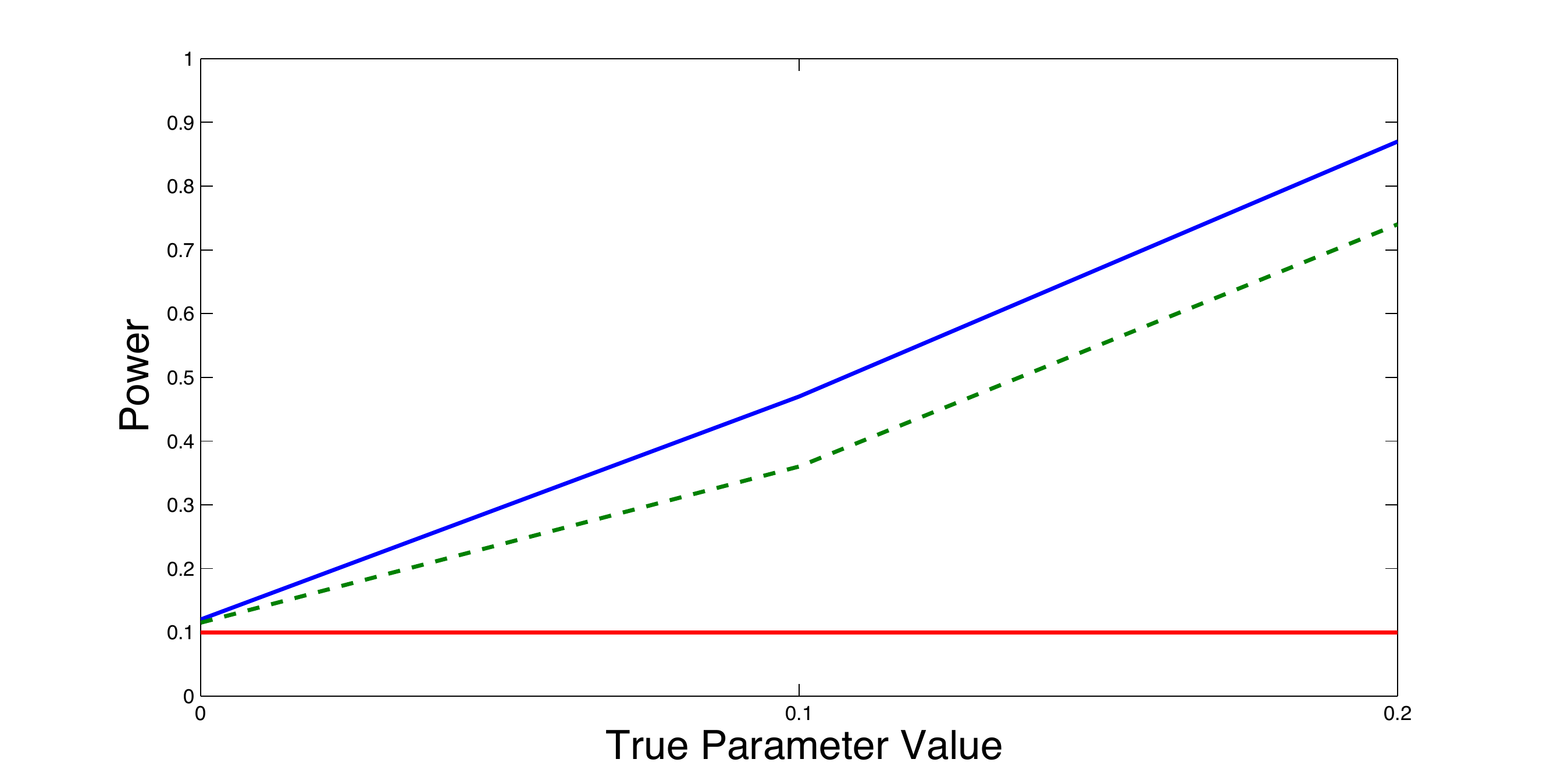} \\
\includegraphics[height=1.5in,width=2in]{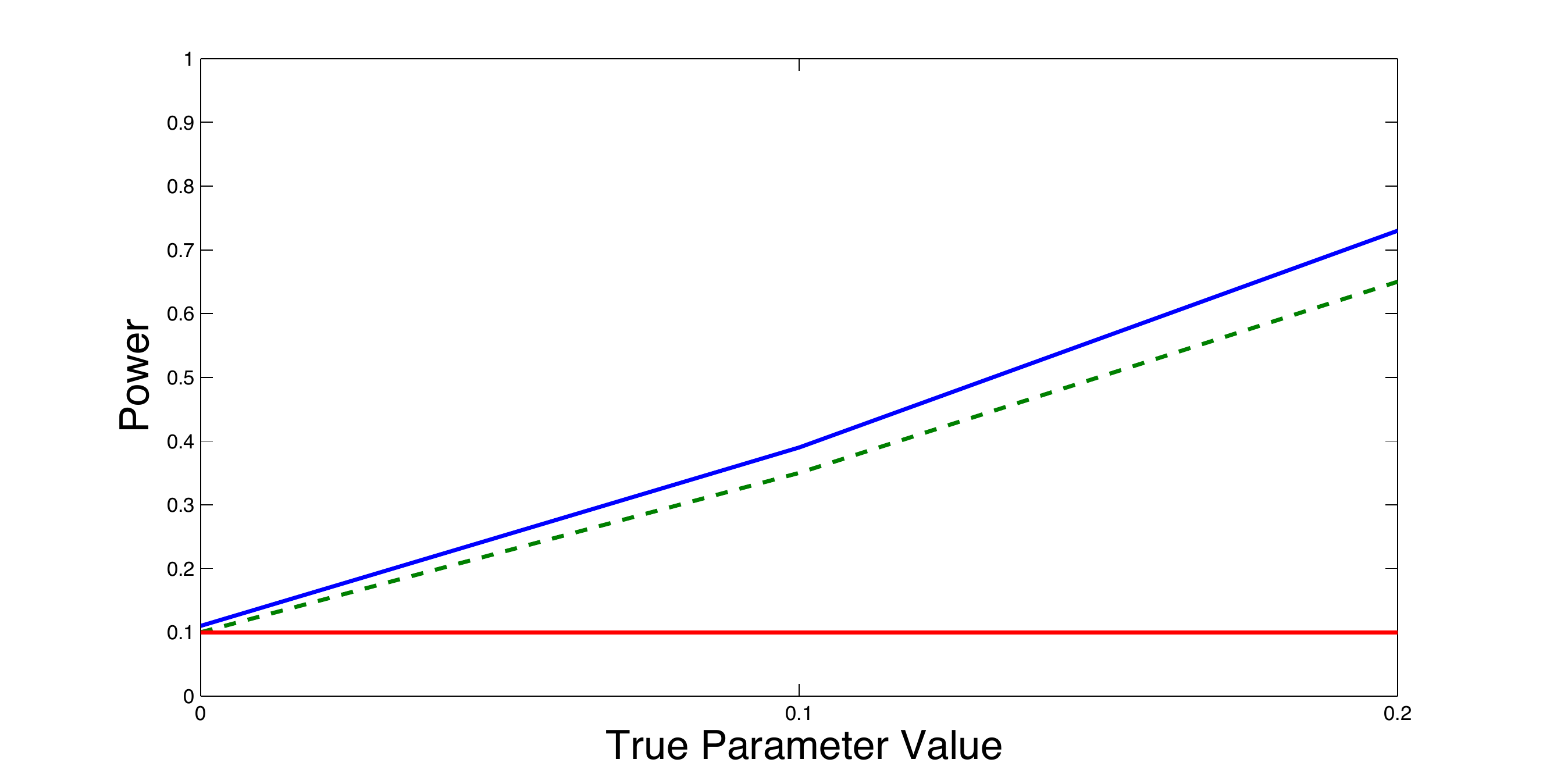} & 
\includegraphics[height=1.5in,width=2in]{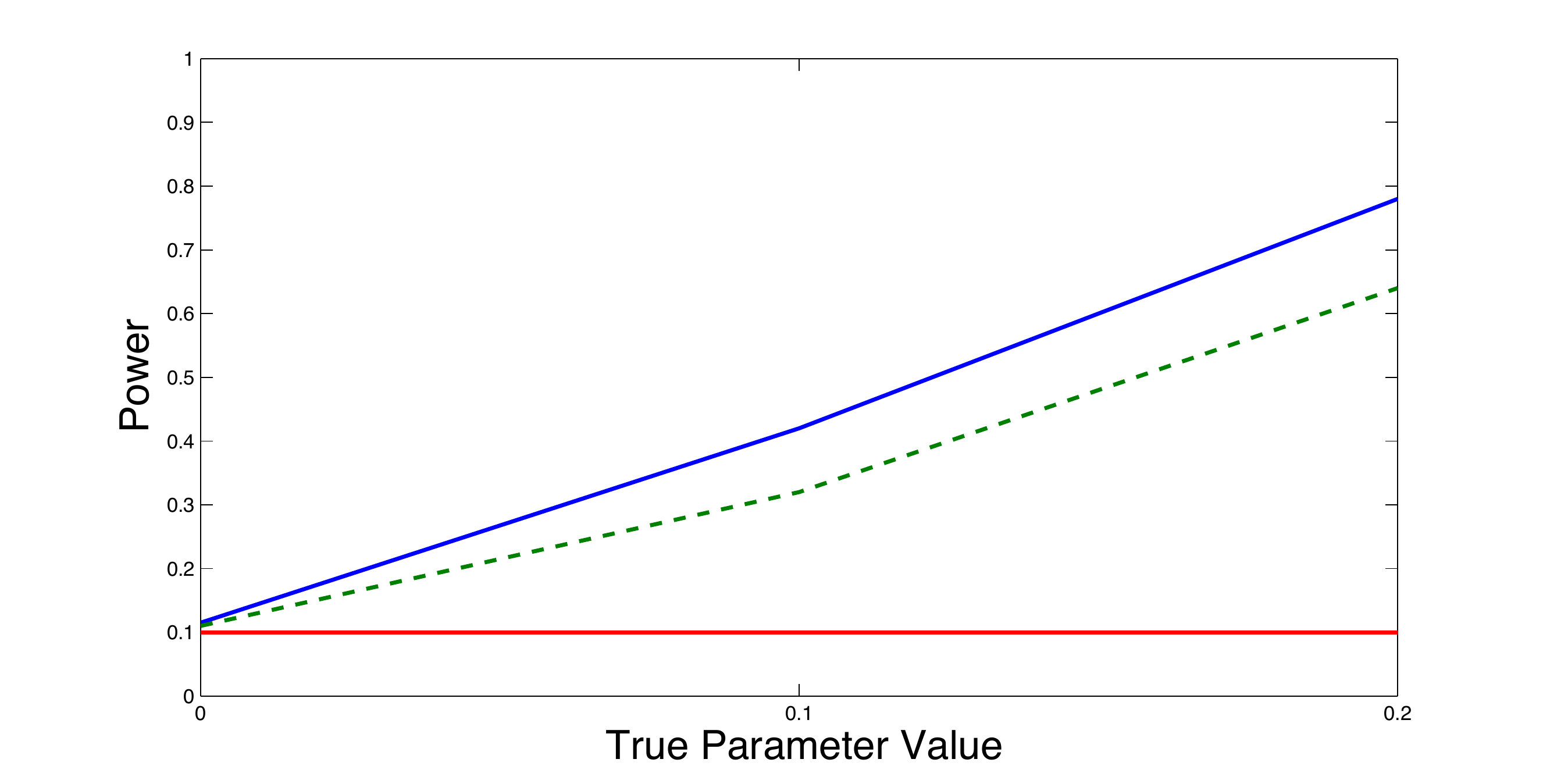} &
\includegraphics[height=1.5in,width=2in]{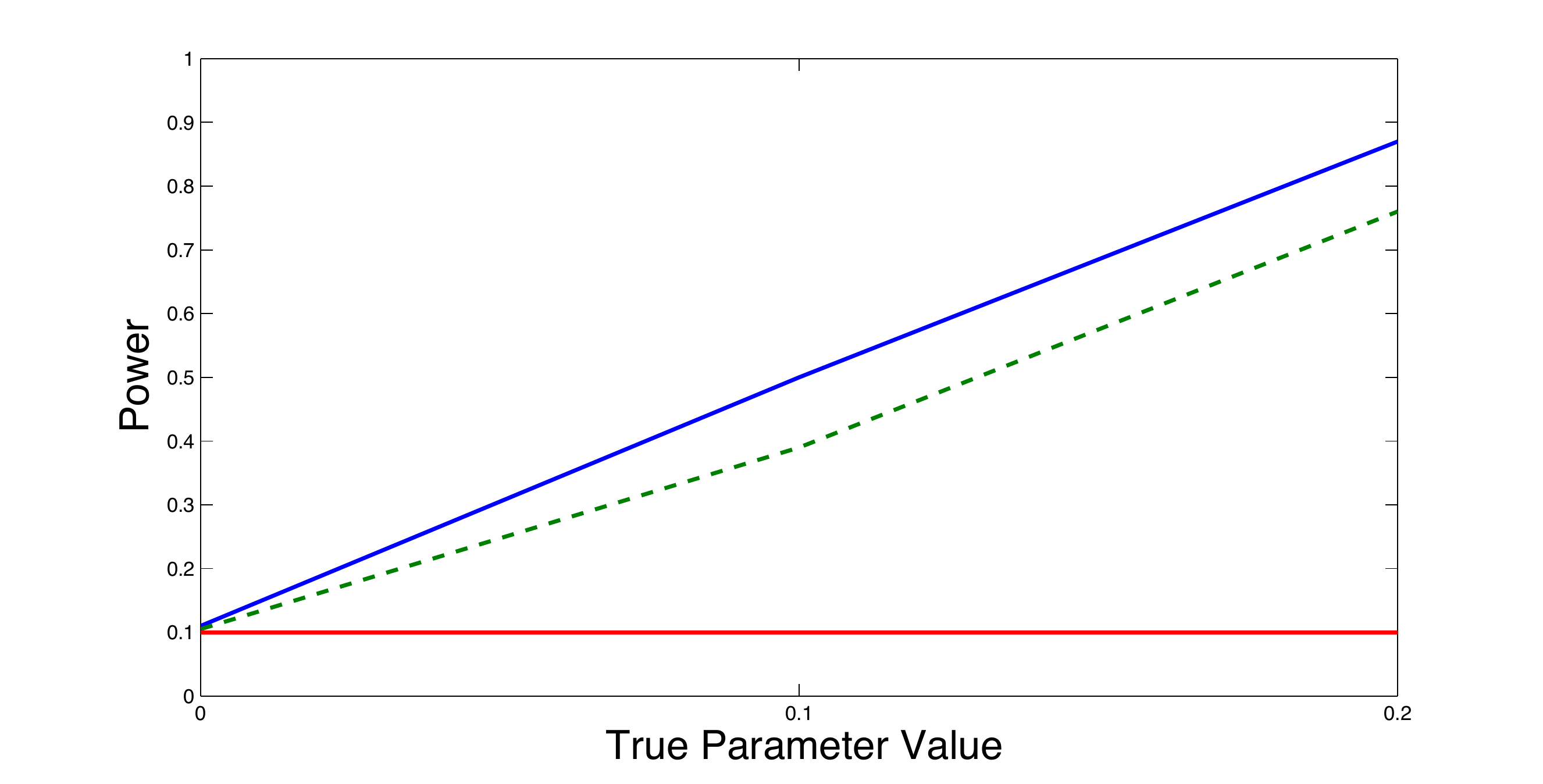} 
\end{array}$
\caption{\scriptsize{
 {\bf Power Curves.} We plot the proportion of rejections of the null hypothesis ${\cal H}_0 :
\beta_{0}=0$, when the true parameter value is
$\beta_{0}\in\{0,0.1,0.2\}$. In the first row, we consider the conventional subsampling. In the second row, we consider the conventional bootstrap. In the third row, we consider our robust subsampling, while in the last row we present our robust bootstrap. In the first, second and third columns, the degree of persistence is $\rho=0.3$, $\rho=0.5$, and $\rho=0.7$, respectively.
We consider noncontaminated samples (straight line) and contaminated samples (dashed
line).}}\label{powerstat}
\end{figure}



\newpage

\begin{figure}[!h]
\center $ \begin{array}{cc}
\includegraphics[height=1.75in,width=2.5in]{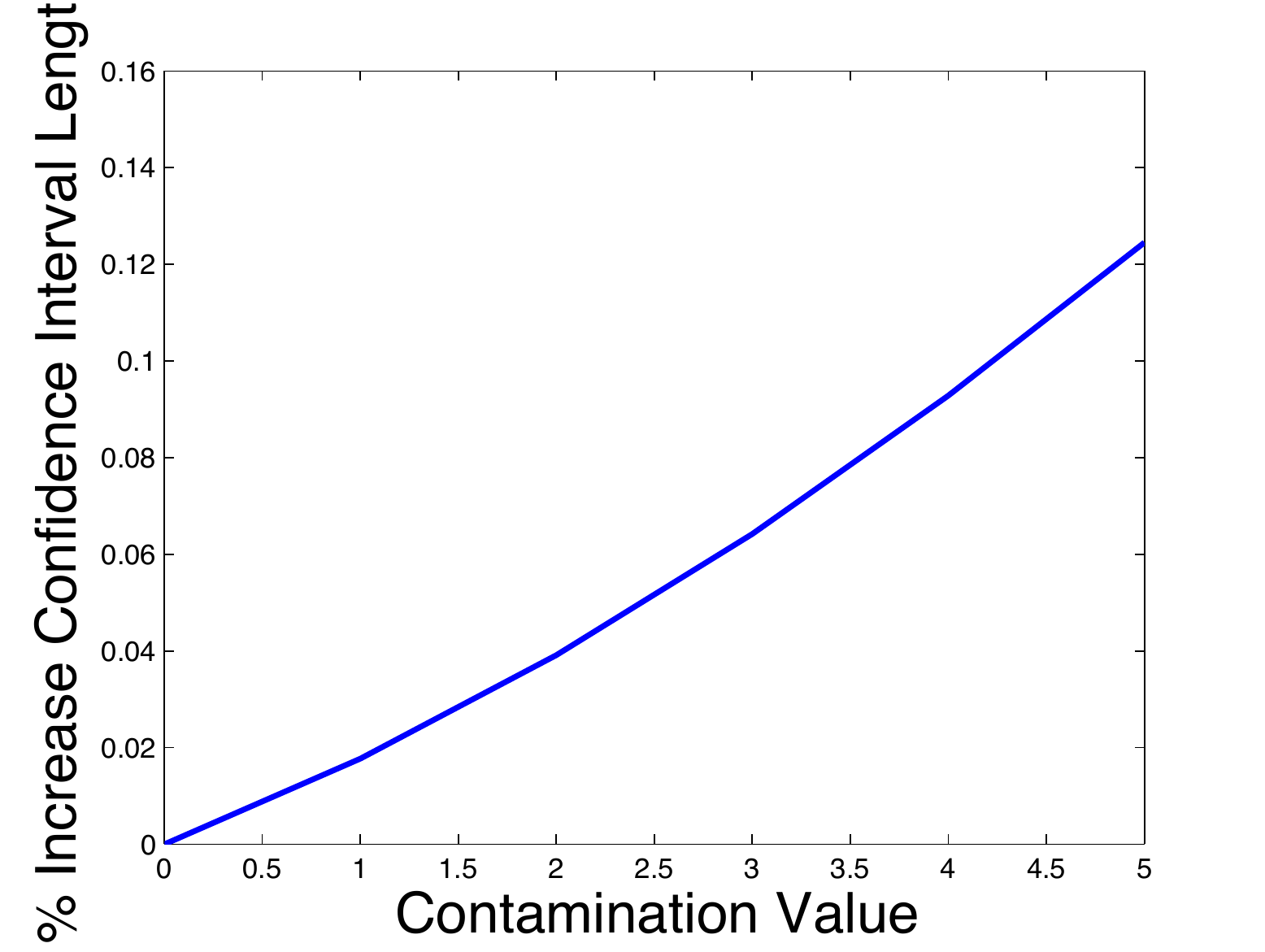} &
\includegraphics[height=1.75in,width=2.5in]{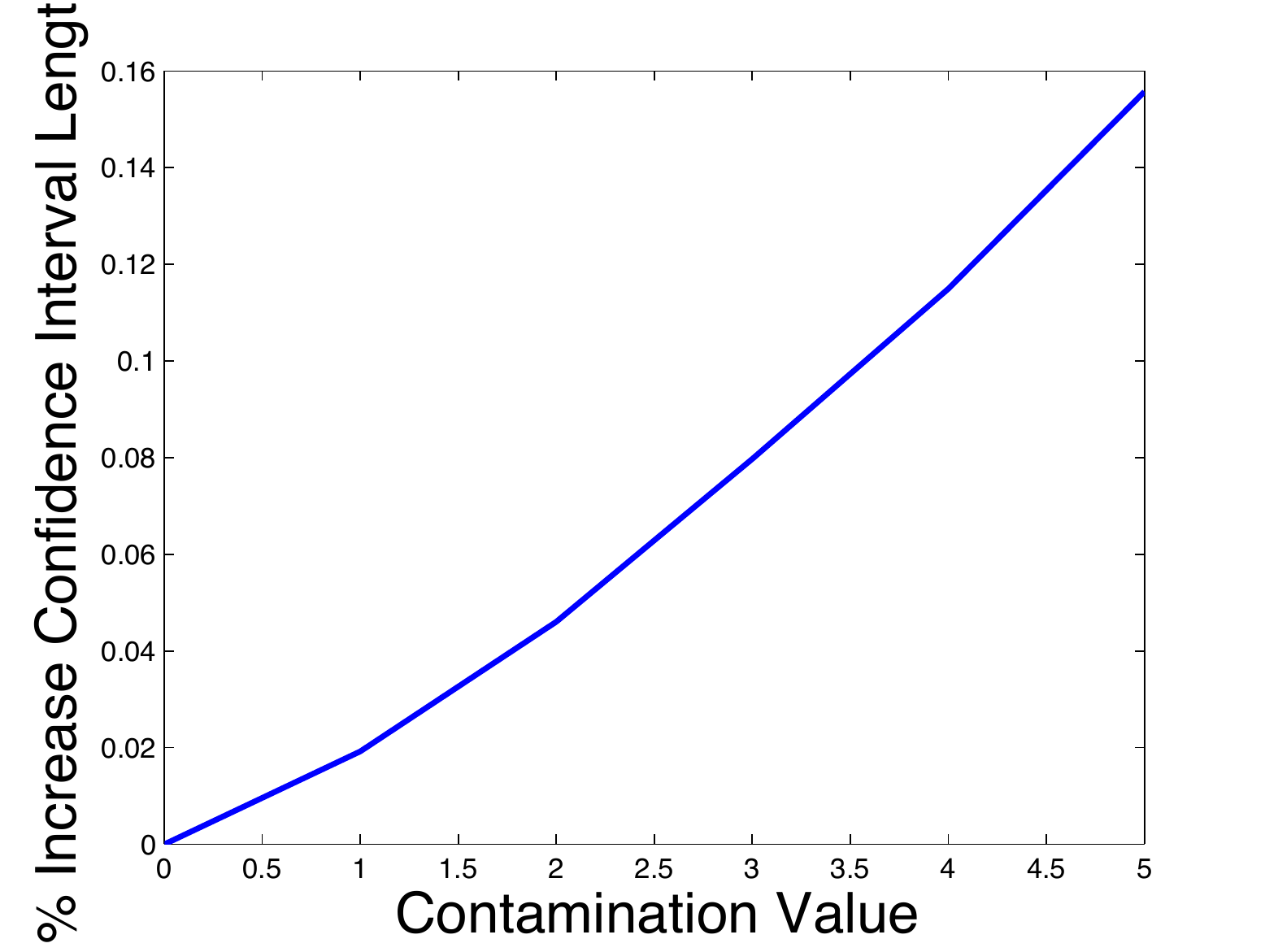}\\
& \\
& \\
\includegraphics[height=1.75in,width=2.5in]{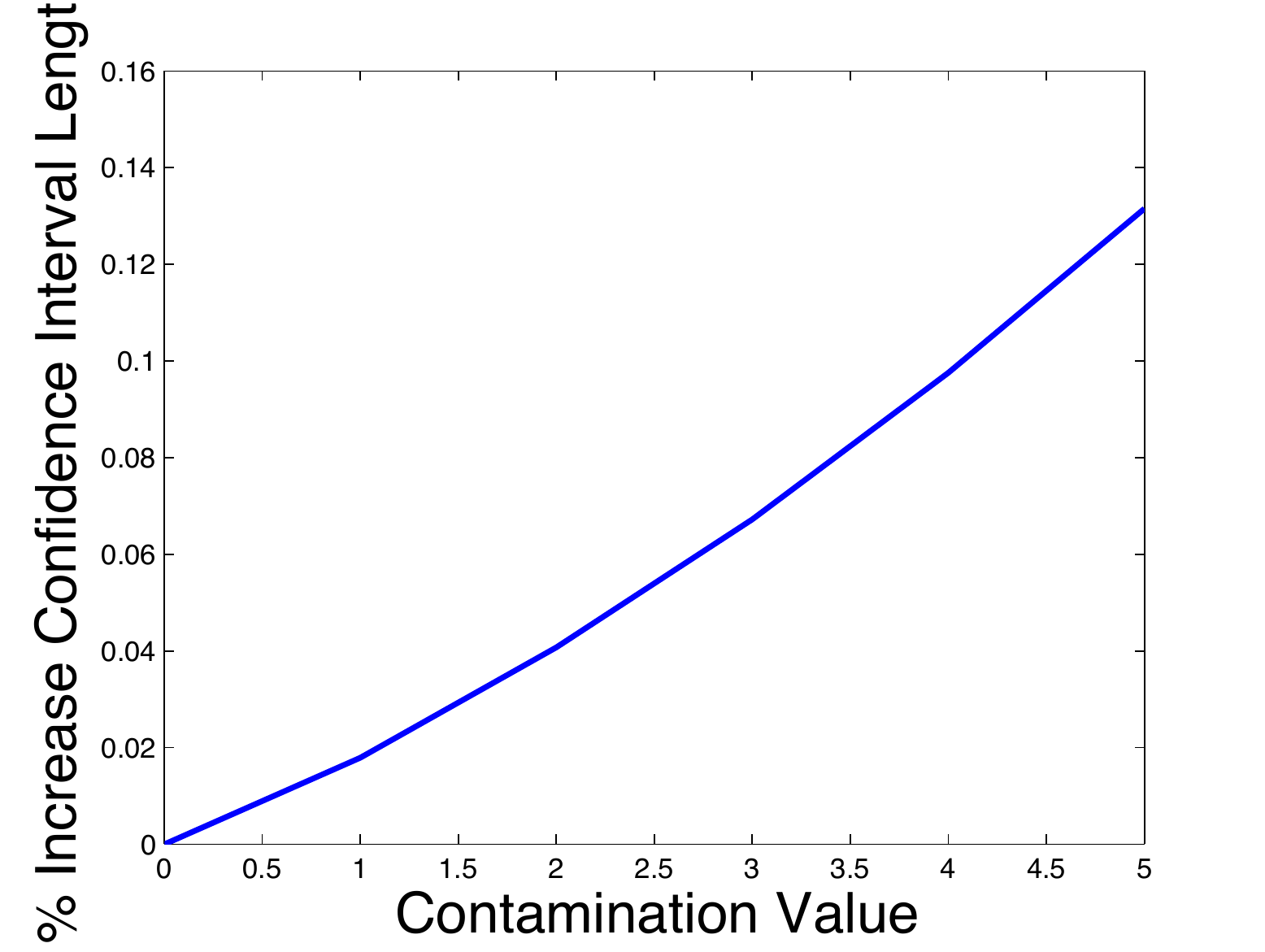} &
\includegraphics[height=1.75in,width=2.5in]{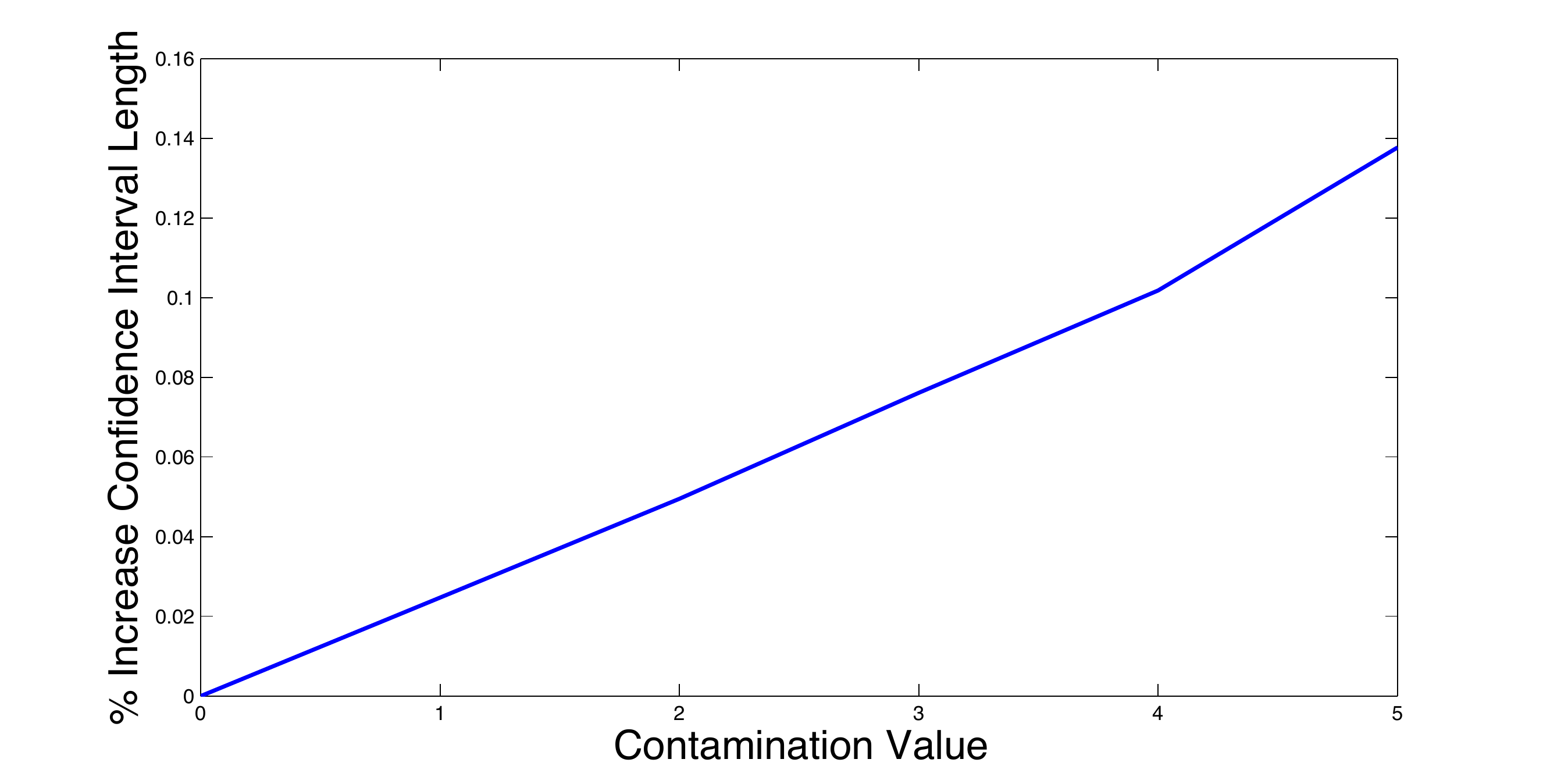}\\
& \\
& \\
\includegraphics[height=1.75in,width=2.5in]{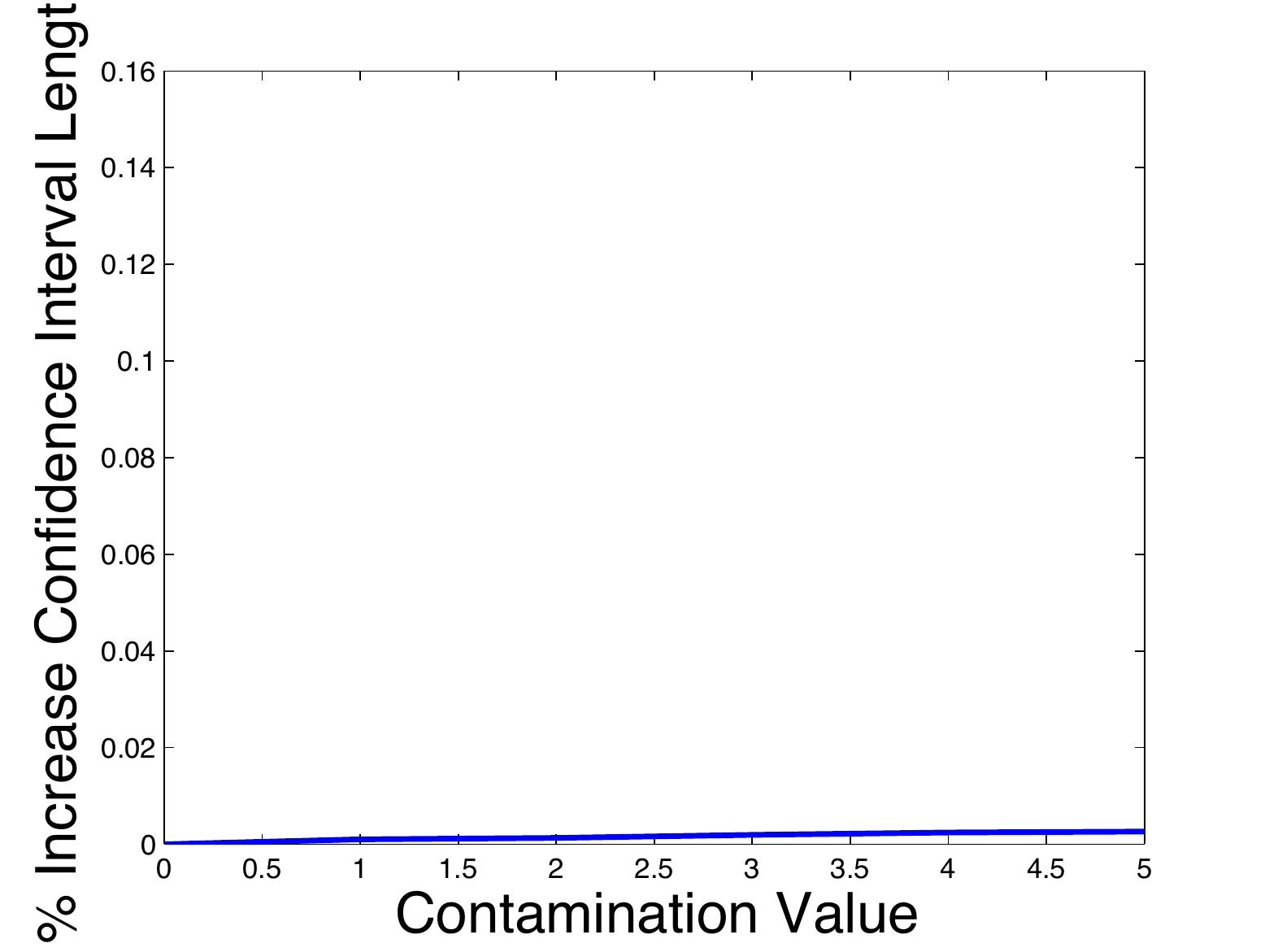} &
\includegraphics[height=1.75in,width=2.5in]{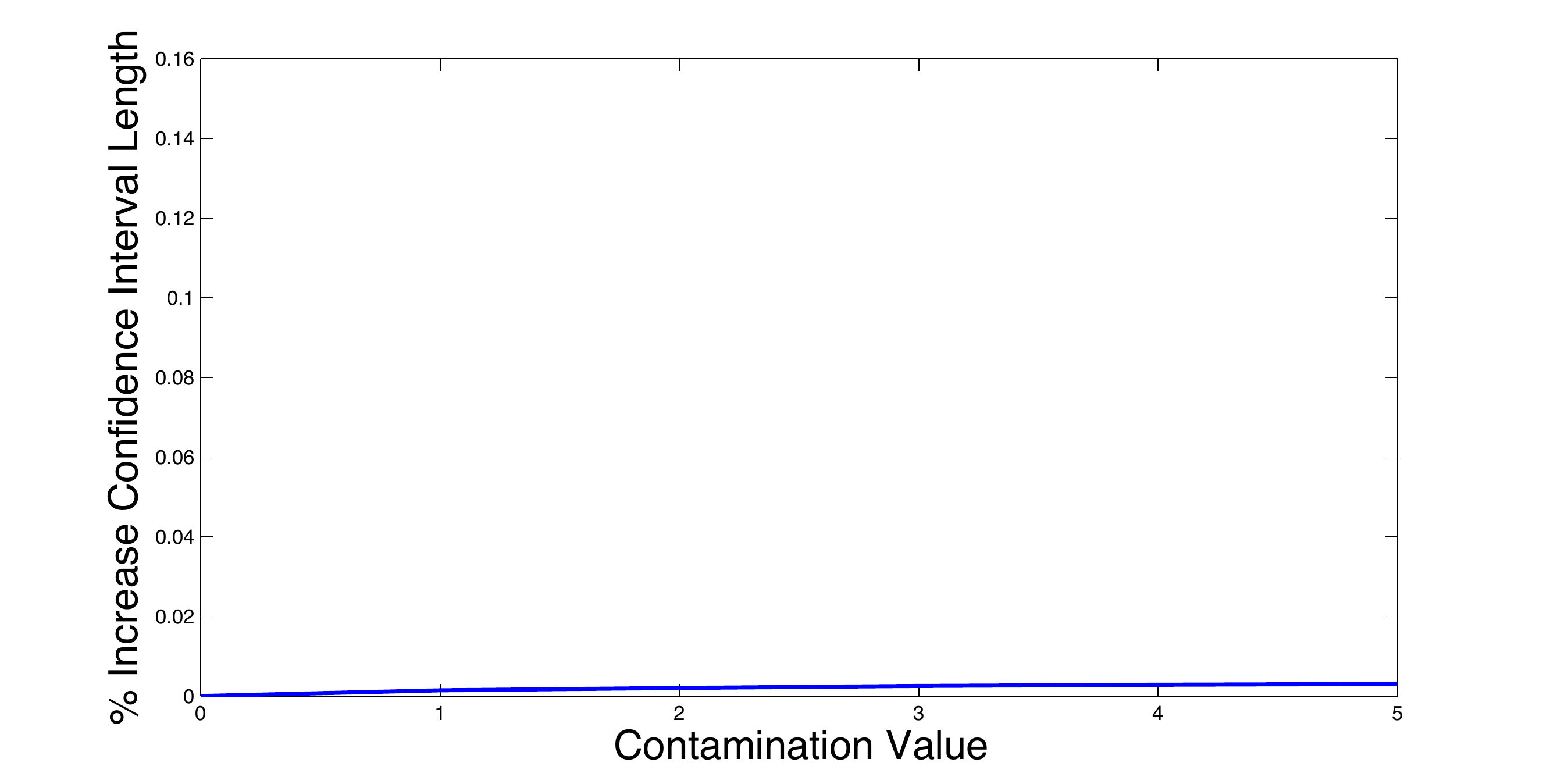}\\
\end{array}$
\caption{\scriptsize{
 {\bf Sensitivity Analysis.} We plot the percentage of  increase of the confidence interval lengths with respect to variation of $y_{max}$, in each Monte Carlo sample, within the interval $[0, 5]$. In the first row, from the left to the right, we consider the bias-corrected method proposed in Amihud, Hurvich and Wang (2008) and the Bonferroni approach for the local-to-unity asymptotic theory introduced in Campbell and Yogo (2006), respectively.
In the second row, from the left to the right, we consider the conventional subsampling and bootstrap, respectively.
Finally, in the last row, from the left to the right, we consider our robust subsampling and bootstrap, respectively.}}\label{sensacy}
\end{figure}

\newpage

\begin{figure}[!h]
\center $ \begin{array}{c}
\includegraphics[height=2in,width=5.5in]{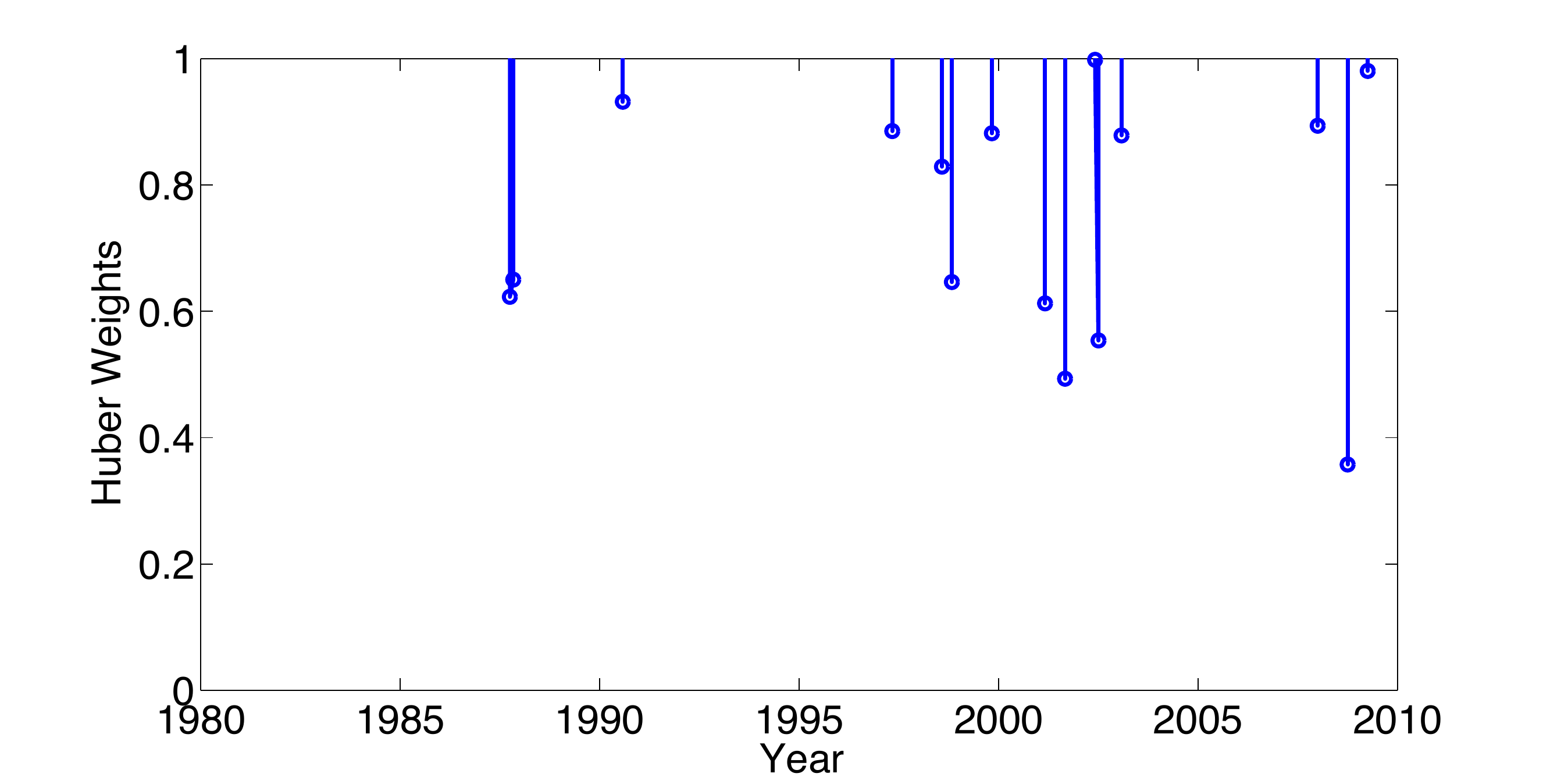}
\end{array}$
\caption{\scriptsize{
 {\bf Huber Weights under the Predictive Regression Model (\ref{pregmodel31}).} We plot the Huber weights for the predictive regression model (\ref{pregmodel31}) in the period 1980-2010.}}\label{huberweights}
\end{figure}

\begin{figure}[!h]
\center $ \begin{array}{c}
\includegraphics[height=2in,width=5.5in]{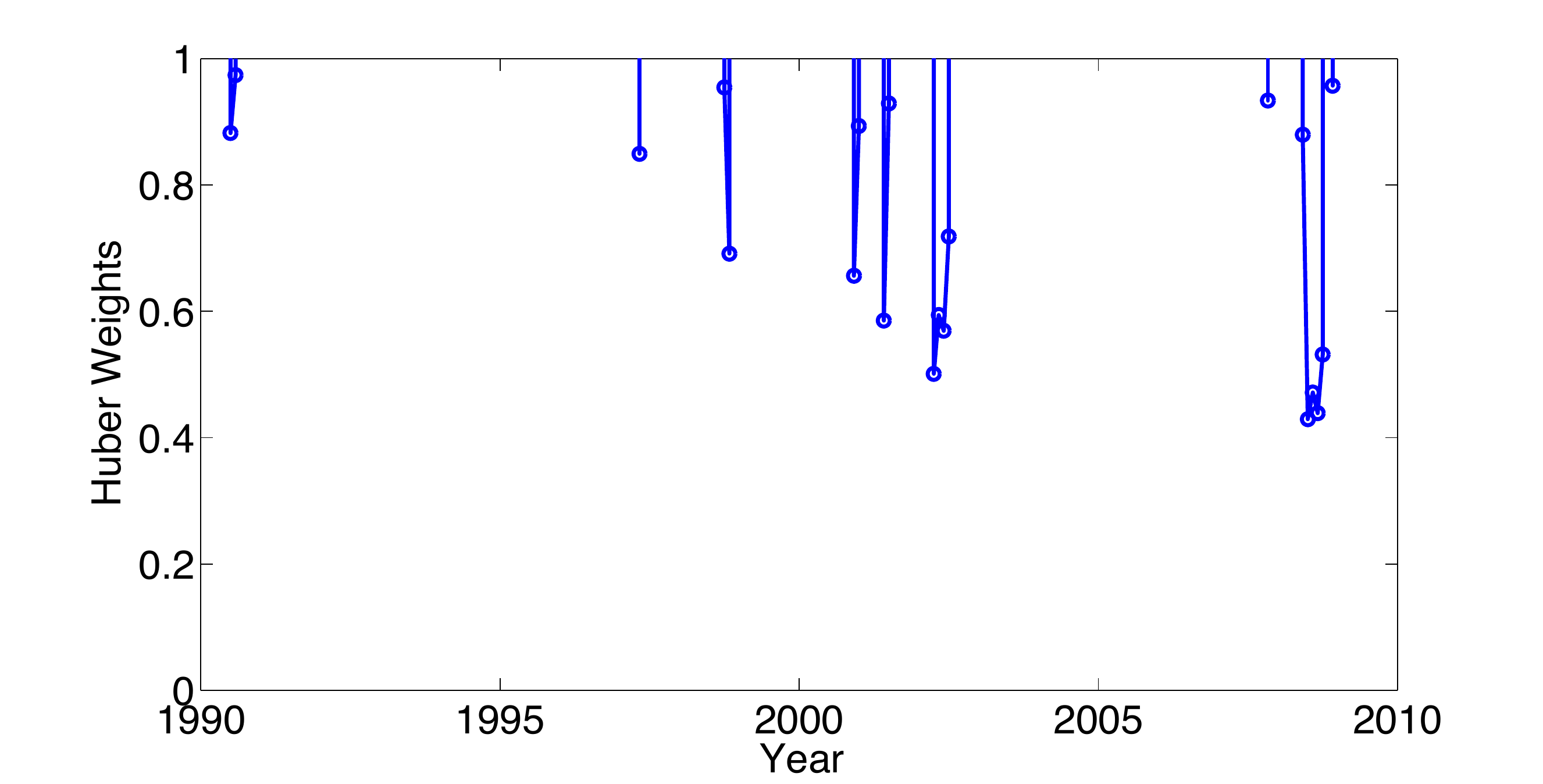}
\end{array}$
\caption{\scriptsize{
 {\bf Huber Weights under the Predictive Regression Model (\ref{pregmodel3f}).} We plot the Huber weights for the predictive regression model (\ref{pregmodel3f}) in the period 1990-2010.}}\label{huberweights1}
\end{figure}

\begin{figure}[!h]
\center $ \begin{array}{c}
\includegraphics[height=2in,width=5.5in]{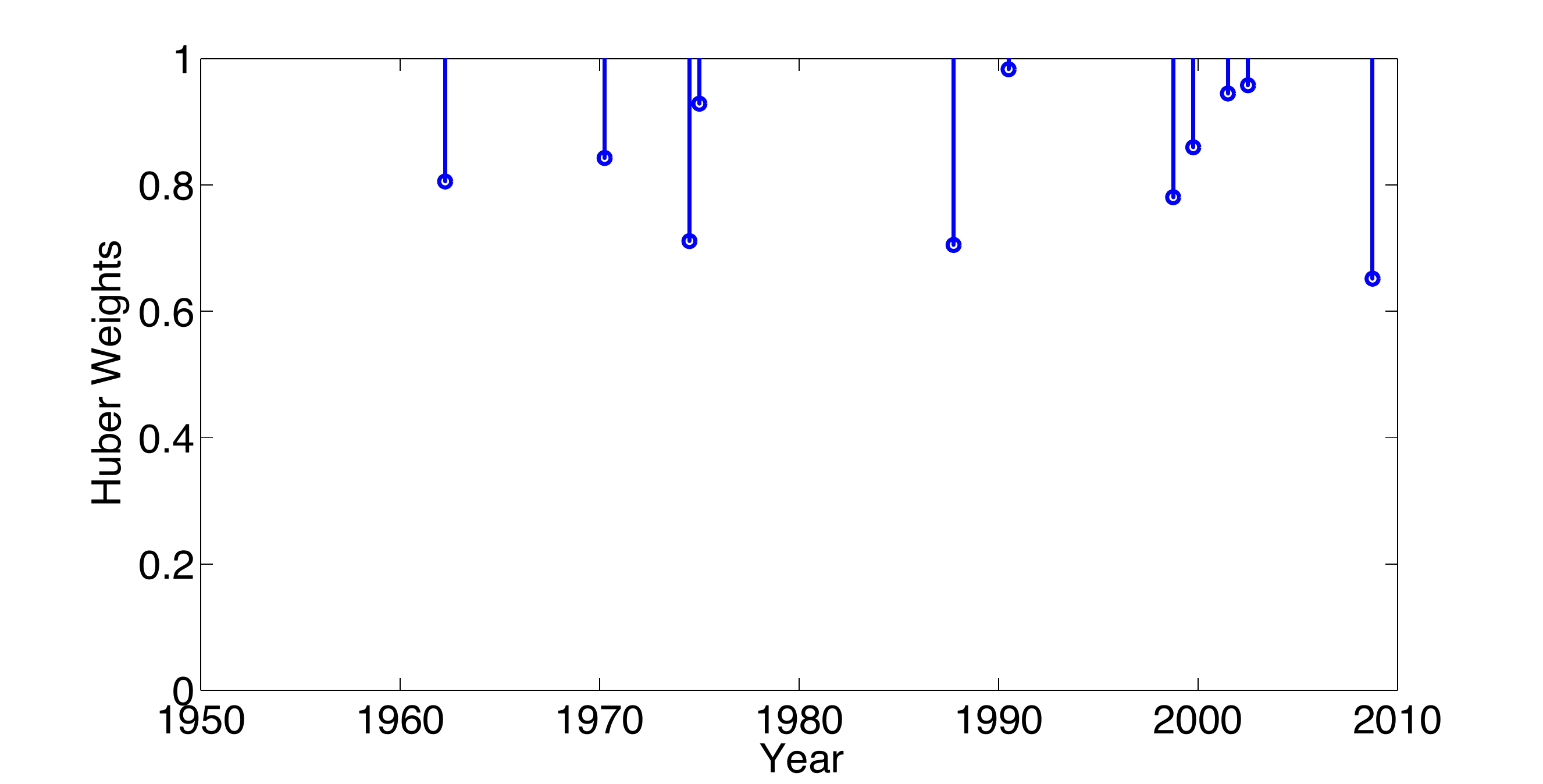}
\end{array}$
\caption{\scriptsize{
 {\bf Huber Weights under the Predictive Regression Model (\ref{pregmodel3}).} We plot the Huber weights for the predictive regression model (\ref{pregmodel3}) in the period 1950-2010.}}\label{huberweights2}
\end{figure}

%
%
%
%
%
%
%
%
%
%

\newpage

\begin{figure}[!h]
\center $ \begin{array}{c}
\includegraphics[height=2.5in,width=3.5in]{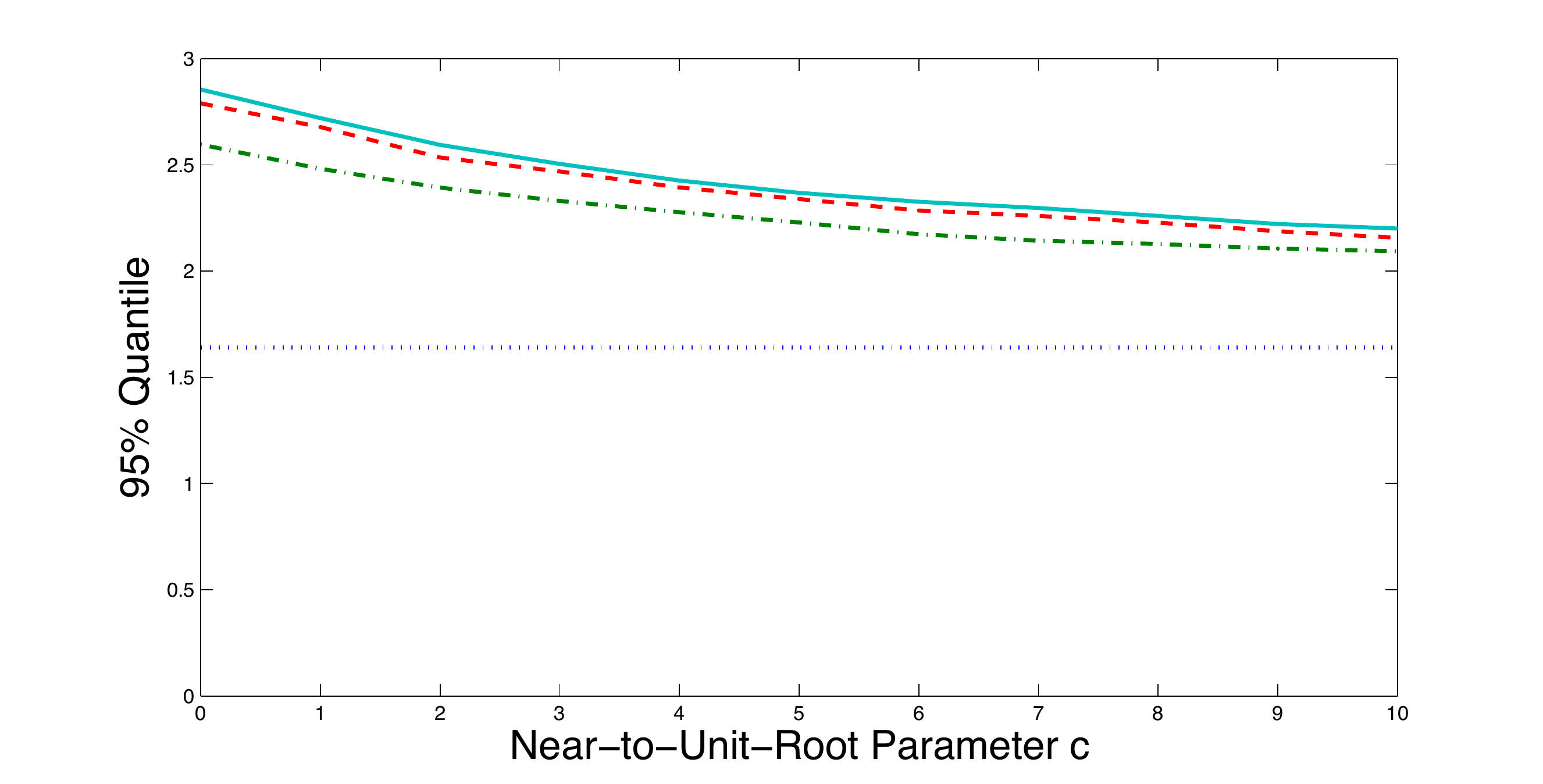} \\
\includegraphics[height=2.5in,width=3.5in]{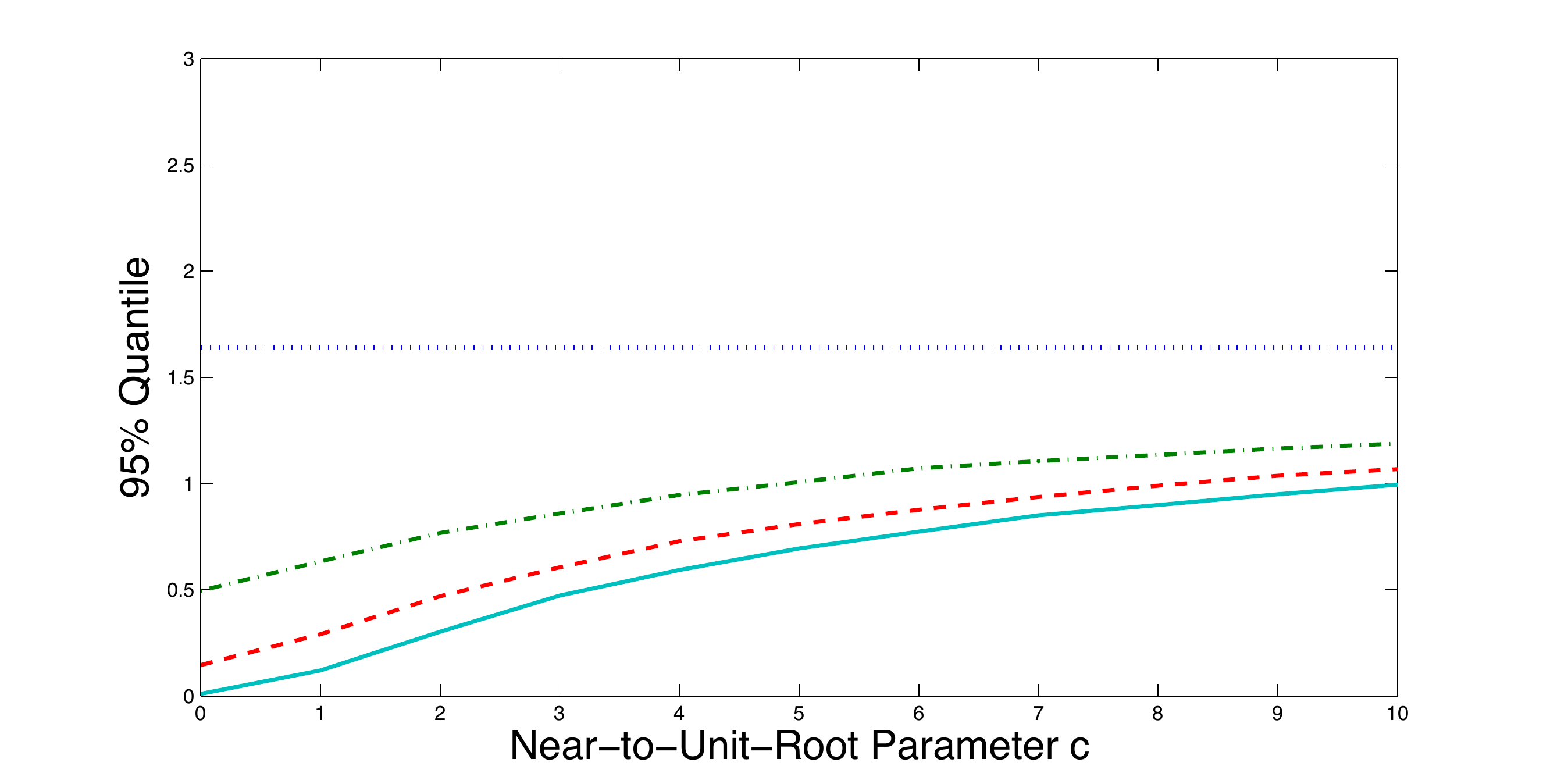} \\
\includegraphics[height=2.5in,width=3.5in]{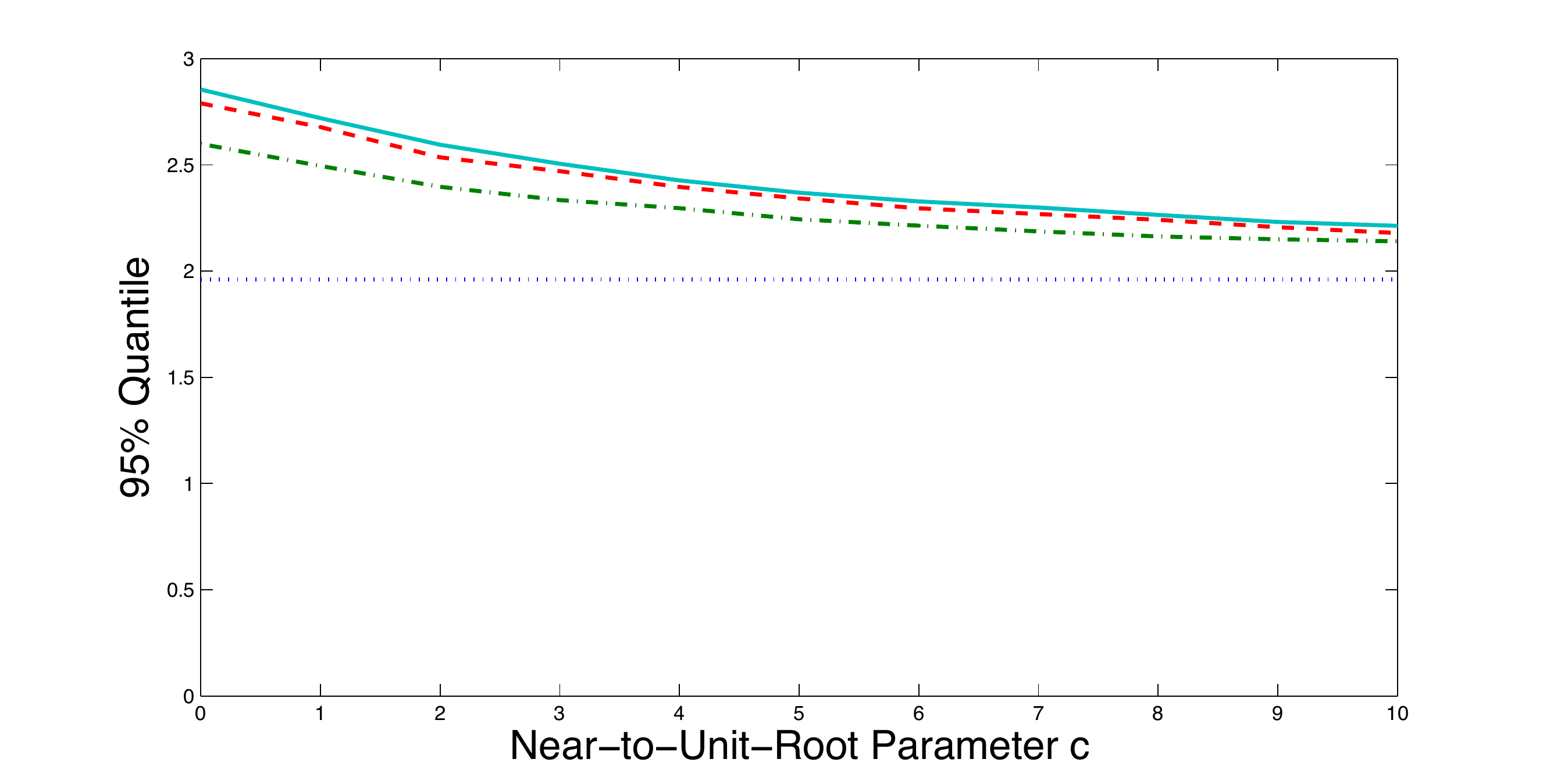} 
\end{array}$
\caption{\scriptsize{
 {\bf 0.95-Quantiles.} Let 
 $T_n^R=\sqrt{n}(\hat{\beta}_n^R-\beta_0)/\hat{\sigma}_n^R$.
 From the top to the bottom, we plot the $0.95$-quantiles of the limit distribution of statistics $T_n^R$ $-T_n^R$, and $\vert T_n^R\vert$ for different values of the degree of persistence $\rho=1-c/n$, with $c\in[0,10]$.
 The covariance parameter of the error terms is $\phi=0$ (dotted line), $\phi=-1$ (dash-dotted line), $\phi=-2$ (dashed line), and $\phi=-5$ (solid line).}}\label{asize}
\end{figure}

\newpage

\begin{table}[!h]
\begin{center}
\begin{tabular}{|c|c|c|}
\hline
$n=120$, $b=0.5$  & $0.9$ & $0.95$ \\
\hline
Subsampling ($m=10$) & $[0.0417;0.0833]$ & $[0.0417;0.0417]$ \\
Subsampling ($m=20$) & $[0.0833;0.0833]$ & $[0.0833;0.0833]$ \\
Subsampling ($m=30$) & $[0.1250;0.1250]$ & $[0.1250;0.1250]$ \\
\hline
Bootstrap ($m=10$) & $[0.0417;0.3750]$ & $[0.0417;0.3333]$ \\
Bootstrap ($m=20$) & $[0.0833;0.3333]$ & $[0.0833;0.3333]$ \\
Bootstrap ($m=30$) & $[0.1250;0.3333]$ & $[0.1250;0.2500]$ \\
\hline
\end{tabular}
\end{center}
\caption{\scriptsize{{\bf Subsampling and Block Bootstrap Lower and Upper
Bounds for the Quantile Breakdown Point.} Breakdown point of the
subsampling and the block bootstrap quantiles. The sample size is $n=120$,
and the block size is $m=10,20,30$. We assume a statistic with breakdown
point $b=0.5$ and confidence levels $t=0.9,0.95$. Lower and upper
bounds for quantile breakdown points are computed using Theorem
\ref{bsubboot}.}}\label{table1}
\end{table}


\begin{table}[!h]
\begin{center}
\begin{tabular}{|c|c|c|}
\hline
$n=120$ & $0.9$ & $0.95$ \\
\hline
Subsampling ($m=10$) & $\phantom{PPp}0.1750\phantom{PPp}$ & $\phantom{PPp}0.1250\phantom{PPp}$ \\
Subsampling ($m=20$) & $0.2500$ & $0.2083$ \\
Subsampling ($m=30$) & $0.3250$ & $0.2833$ \\
\hline
Bootstrap ($m=10$) & $0.5000$ & $0.5000$ \\
Bootstrap ($m=20$) & $0.5000$ & $0.5000$ \\
Bootstrap ($m=30$) & $0.4250$ & $0.3583$ \\
\hline
\end{tabular}
\end{center}
\caption{\scriptsize{{\bf Robust Subsampling and Robust Block Bootstrap for the studentized Statistic $T_n$.} Breakdown point of the
robust subsampling and the robust block bootstrap quantiles for the studentized statistic $T_n$, in the predictive regression model (\ref{pregmodel1})-(\ref{pregmodel2}). The sample size is $n=120$,
and the block size is $m=10,20,30$. The quantile breakdown points are computed using Theorem
\ref{ftheo}.}}\label{table111}
\end{table}

\begin{table}
\begin{center}
\begin{tabular}{|c|c|c|c|c|c|}
\hline
$\phantom{P}$  & $1980-1995$ & $1985-2000$ & $1990-2005$ & $1995-2010$  \\
\hline
Bias-Corrected  & $0.0292^{(\ast\ast)}$ & $0.0167^{(\ast\ast)}$ & $0.0191^{(\ast\ast)}$ & $0.0156$ \\
\hline
Bonferroni  & $0.0236^{(\ast\ast)}$ & $0.0134^{(\ast\ast)}$ & $0.0117^{(\ast)}$ & $0.0112$ \\
\hline
Subsampling  & $0.0430^{(\ast\ast)}$ & $0.0175$ & $0.0306^{(\ast\ast)}$ & $0.0355^{(\ast\ast)}$ \\
\hline
R.Subsampling  & $0.0405^{(\ast\ast)}$ & $0.0174^{(\ast\ast)}$ & $0.0245^{(\ast\ast)}$ & $0.0378^{(\ast\ast)}$ \\
\hline
\end{tabular}
\end{center}
\caption{\scriptsize{{\bf Point Estimates of Parameter $\beta$.} We report the point estimates of the parameter $\beta$ in the predictive regression model (\ref{pregmodel31}) for the subperiods, 1980-1995, 1985-2000, 1990-2005 and 1995-2010, all consisting of $180$ observations. In the second and third line we consider the bias-corrected method and the Bonferroni approach, respectively. In the fourth and fifth line we consider the conventional subsampling and our robust subsampling
respectively. $(\ast)$ and $(\ast\ast)$ mean rejection at $10\%$ and $5\%$ significance level, respectively.}}\label{tablef1}
\end{table}

\begin{table}
\begin{center}
\begin{tabular}{|c|c|c|c|c|c|}
\hline
$\phantom{P}$  & $1991-2006$ & $1992-2007$ & $1993-2008$ & $1994-2009$ & $1995-2010$ \\
\hline
Subsampling   & $0.0369^{(\ast\ast)}$ & $0.0402^{(\ast\ast)}$ & $0.0454^{(\ast\ast)}$ & $0.0368^{(\ast)}$ & $0.0415^{(\ast\ast)}$ \\
\hline
R.Subsampling  & $0.0368^{(\ast\ast)}$ & $0.0402^{(\ast\ast)}$ & $0.0437^{(\ast\ast)}$ & $0.0366^{(\ast\ast)}$ & $0.0412^{(\ast\ast)}$ \\
\hline
\end{tabular}

$\phantom{P}$

\begin{tabular}{|c|c|c|c|c|c|}
\hline
$\phantom{P}$  & $1991-2006$ & $1992-2007$ & $1993-2008$ & $1994-2009$ & $1995-2010$ \\
\hline
Subsampling  & $0.4700^{(\ast\ast)}$ & $0.4648^{(\ast\ast)}$ & $0.4968^{(\ast\ast)}$ & $0.3859^{(\ast\ast)}$ & $0.3993^{(\ast\ast)}$ \\
\hline
R.Subsampling  & $0.4821^{(\ast\ast)}$ & $0.4771^{(\ast\ast)}$ & $0.5276^{(\ast\ast)}$ & $0.3932^{(\ast\ast)}$ & $0.4083^{(\ast\ast)}$ \\
\hline
\end{tabular}
\end{center}
\caption{\scriptsize{{\bf Point Estimates of Parameters $\beta_1$ and $\beta_2$.} We report the point estimates of parameters $\beta_1$ (fist table) and $\beta_2$ (second table) in the predictive regression model (\ref{pregmodel3f}) for the subperiods
1991-2006, 1992-2007, 1993-2008, 1994-2009 and 1995-2010, all consisting of $180$ observations. In the second and third line we consider the conventional subsampling and our robust subsampling, respectively. 
$(\ast)$ and $(\ast\ast)$ mean rejection at $10\%$ and $5\%$ significance level, respectively.}}\label{tablef23}
\end{table}

\newpage

\begin{table}
\begin{center}
\begin{tabular}{|c|c|c|c|c|}
\hline
$\phantom{P}$  & $1950-1995$ & $1955-2000$ & $1960-2005$ & $1965-2010$  \\
\hline
Subsampling  & $0.0724^{(\ast\ast)}$ & $0.0301$ & $0.0467^{(\ast\ast)}$  & $0.0480^{(\ast)}$ \\
\hline
R.Subsampling  & $0.0721^{(\ast\ast)}$ & $0.0305^{(\ast\ast)}$ & $0.0474^{(\ast\ast)}$ & $0.0488^{(\ast\ast)}$  \\
\hline
\end{tabular}

$\phantom{P}$

%
\begin{tabular}{|c|c|c|c|c|}
\hline
$\phantom{P}$  & $1950-1995$ & $1955-2000$ & $1960-2005$ & $1965-2010$  \\
\hline
Subsampling  & $-0.1509$ & $-0.2926^{(\ast\ast)}$ & $-0.2718^{(\ast\ast)}$ & $-0.2187$  \\
\hline
R.Subsampling  & $-0.1532^{(\ast\ast)}$ & $-0.2920^{(\ast\ast)}$ & $-0.2701^{(\ast\ast)}$ & $-0.2173^{(\ast\ast)}$  \\
\hline
\end{tabular}
\end{center}
\caption{\scriptsize{{\bf Point Estimates of Parameters $\beta_1$ and $\beta_2$.} We report the point estimates of parameters $\beta_1$ (first table) and $\beta_2$ (second table) in the predictive regression model (\ref{pregmodel3}) for the subperiods
1950-1995, 1955-2000, 1960-2005 and 1965-2010, all consisting of $180$ observations.In the second and third line we consider the conventional subsampling and our robust subsampling, respectively. 
$(\ast)$ and $(\ast\ast)$ mean rejection at $10\%$ and $5\%$ significance level, respectively.}}\label{tablef45}
\end{table}

\begin{table}
\begin{center}
\begin{tabular}{|l|c|c|c|}
\hline
$\phantom{P}$  & $R^2_{OS}$ & $R^2_{OS,OLS}$ & $R^2_{OS,ROB}$  \\
\hline
Shiller  & $\hphantom{-}0.0051$ & $\hphantom{-}0.0351$ & $\hphantom{-}0.0404$   \\
\hline
Bollerslev et al.  & $\hphantom{-}0.0140$ & $\hphantom{-}0.0437$ & $\hphantom{-}0.0570$   \\
\hline
Santos and Veronesi  & $\hphantom{-}0.0113$ & $-0.0389$ & $-0.0273$   \\
\hline
\end{tabular}
\end{center}
\caption{\scriptsize{{\bf Out-of-Sample $R^2$ Statistics.} We report the out-of-sample $R^2$ statistics for the single predictor model introduced in Section \ref{spm1} (Shiller), and the two-predictor models analyzed in Sections \ref{dyvrp} and \ref{dyli} (Bollerslev et al., and Santos and Veronesi), respectively.}}\label{tableos}
\end{table}

\end{document}